\documentclass{amsart}

\usepackage{amsmath,amssymb,bbm}
\usepackage{amsthm}
\usepackage{latexsym}
\usepackage{graphicx}
\usepackage{color,wrapfig}
\usepackage{subfigure}
\usepackage{multirow}
\usepackage{bm,url,comment}
\usepackage{mathtools,mathrsfs}
\usepackage{cool}
\usepackage{bm}
\usepackage{array}

\usepackage[neverdecrease]{paralist}

\usepackage{booktabs}

\newcommand{\ii}{\mathsf{i}}

\theoremstyle{definition}
\newtheorem{theorem}{Theorem}[section]

\newtheorem{lemma}[theorem]{Lemma}
\newtheorem{corro}[theorem]{Corollary}
\newtheorem{thm}{Theorem}[section]
\newtheorem{lem}[thm]{Lemma}

\theoremstyle{remark}
\newtheorem*{remark}{Remark}
\newtheorem*{example}{Example}
\newtheorem{assumption}[thm]{Assumption}

\usepackage{hyperref}
\hypersetup{
    colorlinks,
    citecolor=black,
    filecolor=black,
    linkcolor=blue,
    urlcolor=black
}

\begin{document}

\title[]{Uncertainty quantification of synchrosqueezing transform under complicated nonstationary noise}

\author{Hau-Tieng Wu}

\address{Courant Institute of Mathematical Sciences, New York University, New York, 10012, USA} 

\author{Zhou Zhou}

\address{Department of Statistics, University of Toronto, Toronto, ON,  Canada}

\begin{abstract}
We propose a bootstrapping framework to quantify uncertainty in time-frequency representations (TFRs) generated by the short-time Fourier transform (STFT) and the STFT-based synchrosqueezing transform (SST) for oscillatory signals with time-varying amplitude and frequency contaminated by complex nonstationary noise.
To this end, we leverage a recent high-dimensional Gaussian approximation technique to establish a sequential Gaussian approximation for nonstationary processes under mild assumptions. This result is of independent interest and provides a theoretical basis for characterizing the approximate Gaussianity of STFT-induced TFRs as random fields.
Building on this foundation, we establish the robustness of SST-based signal decomposition in the presence of nonstationary noise. Furthermore, assuming locally stationary noise, we develop a Gaussian autoregressive bootstrap for uncertainty quantification of SST-based TFRs and provide theoretical justification. We validate the proposed methods with simulations and illustrate their practical utility by analyzing spindle activity in electroencephalogram recordings. Our work bridges time-frequency analysis in signal processing and nonlinear spectral analysis of time series in statistics.
\end{abstract}

\maketitle

\section{Introduction}

Uncertainty quantification (UQ) is crucial for applying time-frequency (TF) analysis \cite{Flandrin2018explorations} tools, particularly the synchrosqueezing transform (SST) \cite{Daubechies2011synchro}, to study nonstationary time series, as it ensures the reliability and interpretability of the extracted features. In real-world applications, such as biomedical signal processing, climate modeling, and finance, time series data are often contaminated by nonstationary noise and influenced by various external factors. Without proper UQ, the interpretation of the resulting TF representations (TFR), a function defined on the TF domain, might be misleading and may lead to incorrect conclusions about the underlying dynamics. By incorporating UQ, one can more reliably differentiate true oscillatory patterns from artifacts, improve the robustness of feature extraction, and guide decision-making in practical applications. Furthermore, rigorous uncertainty estimation provides a foundation for statistical inference, enabling hypothesis testing and model validation in TF analysis. In this paper, we establish UQ for SST in the presence of complex nonstationary noise.

Note that SST is a nonlinear TF analysis tool. In the statistics literature, existing results in spectral domain time series analysis have focused primarily on linear methods, particularly spectra based on Fourier and wavelet transforms. It is impossible to have a complete list of relevant literature here and we only list a few representative works. For results on Fourier-based spectra, see \cite{brillinger2001time,dahlhaus1997fitting,liu2010asymptotics,HafnerKirch2017,KreissPaparoditis2015,yang2022spectral};  for wavelet-based spectra, see \cite{nason2000wavelet,nason2013test}. On the other hand, however, in real signal processing applications nonlinear TF analysis tools have gained popularity and showed better accuracy and efficiency in a wide range of situations; see, e.g., recent work \cite{alian2023amplitude,su2024reconsider}. Therefore we see an important gap to fill with respect to statistical inference for nonlinear TF analysis algorithms. Several major challenges need to be tackled with significant theoretical and methodological innovations in order to successfully conduct statistical inference such as UQ for nonlinear type TF analysis tools, which we shall elaborate in the following.

The first challenge is to characterize the distribution of the TFR determined by the short-time Fourier transform (STFT) \cite{Flandrin2018explorations}, which forms the foundation of the SST, when the input is a non-Gaussian and nonstationary random process with nontrivial dependence. In much of the existing literature, particularly work on SST \cite{Thakur2013,chen2014non,yang2018statistical,sourisseau2022asymptotic}, this difficulty is bypassed by modeling the signal as a continuous function and the noise as a stationary Gaussian generalized random process $\Phi$. In this setting, when the window function $h$ is a Schwartz function, the STFT defined by $V_\Phi(t,\eta):=\Phi(h(\cdot-t)e^{\ii2\pi\eta(\cdot-t)})$, where $t\in \mathbb{R}$ is time and $\eta\in \mathbb{R}$ is frequency, is automatically a complex Gaussian random field. This setting allows one to focus primarily on the nonlinearity of SST. 
However, this assumption is often unrealistic in practice: real data are discrete, and the noise is typically non-Gaussian and nonstationary. In the discrete setting, the distribution of the TFR produced by the discrete STFT remains largely unexplored, with the exception of recent work \cite{yang2022spectral} examining finite collections of TF pairs. Intuitively, since the numerical implementation of the STFT consists of weighted sums of random variables, the resulting TFR should be well approximated by a Gaussian random field when the window length is sufficiently large and the noise satisfies mild moment and dependence conditions, by a central limit effect. While this intuition is supported by numerical evidence, precise conditions under which the discretely computed TFR is close to a Gaussian random field have not been rigorously established.
In this paper, we apply a recently developed {\em high-dimensional Gaussian approximation} technique to establish a {\em high-dimensional sequential Gaussian approximation} suitable for analyzing the discrete STFT. We show that, under appropriate discretization and mild conditions on the noise, the TFR generalized by the discrete STFT is asymptotically close to a complex Gaussian random field. This result rigorously extends existing theory to discrete, non-Gaussian, and nonstationary settings that are directly relevant to real-world applications.

The second challenge in developing a bootstrapping algorithm is determining how accurately we can recover the underlying signal, and hence the noise, so that the noise can be effectively resampled.
We focus on oscillatory signals that can be well modeled by the adaptive harmonic model (AHM) \cite{Daubechies2011synchro}, where the signal exhibits oscillations with slowly varying amplitude and frequency.
To recover the underlying signal, we adopt the reconstruction formula based on the SST, motivated by empirical observations suggesting that SST exhibits strong robustness to various forms of nonstationary noise. To date, the robustness of this SST-based reconstruction has been theoretically established under noise models involving distributed stationary random processes \cite{Thakur2013} and distributed nonstationary random processes with stationary correlation structures \cite{chen2014non}. However, real-world noise is often non-Gaussian and nonstationary, and theoretical guarantees in this setting have been lacking.
In this paper, we establish the robustness of the SST-based reconstruction formula for signals satisfying the AHM, even when the signal is contaminated by nonstationary noise.

The third challenge is designing a bootstrapping method to {\color{blue} uniformly} quantify the uncertainty of the TFR determined by SST. In practice, we often lack detailed information about the noise structure, including its stationarity. To address this, we assume the noise is locally stationary and apply the recently developed time-varying autoregressive (tvAR) approximation \cite{ding2023autoregressive} to obtain a reasonable model of the noise. This approximation enables us to resample the noise. However, a key difficulty lies in understanding how the approximation error propagates through the nonlinear operations involved in SST during the bootstrapping process. We show that the TFR of a locally stationary noise process determined by SST can be accurately approximated through bootstrapping based on the tvAR approximation,  thereby enabling the desired UQ.

This paper makes three contributions. First, we show that for a nonstationary time series generated by a filtration mechanism with mild moment and dependence conditions, the TFR from the discrete STFT can be uniformly approximated by a Gaussian random field (Theorem~\ref{Theorem: STFT distribution like Gaussian}). This is achieved by combining a recent high-dimensional Gaussian approximation technique \cite{LiuYangZhou2024} with a blocking strategy \cite{csorgHo1975new} to construct a sequential Gaussian approximation, which is of independent interest and provides a rigorous justification of the Gaussianity of the discrete STFT-based TFR.
Second, we investigate the robustness of the SST-based reconstruction formula for signals satisfying the AHM corrupted by nonstationary noise generated by the same filtration mechanism (Theorem \ref{section:theorem:stability}). A critical aspect is quantifying the discretization error linking discrete and continuous STFT/SST.
Third, under locally stationary noise, we propose a tvAR-based bootstrap for STFT/SST TFR to enable rigorous UQ (Theorem~\ref{Bootstrapping main theorem}). We may want to mention that the bootstrap extends the auto-regressive sieve bootstrap \cite{kreiss2011range} to non-stationary and nonlinear frequency domain setting. To our knowledge, the established UQ results are the first uniform statistical inference results on nonlinear TF analysis in the literature.

{\it Notation:} For a random vector $v\in \mathbb{R}^d$, denote $|v|$ to be its Euclidean norm and $|v|_\infty$ to be its $\ell^\infty$ norm. For a random variable $X$, denote $\|X\|_q$ to be its $\mathcal{L}^q$ norm. For two sequences $a_n$ and $b_n$, denote $a_n=O_{\ell^2}(b_n)$ to state that $a_n/b_n$ is bounded in $\mathcal{L}^2$ norm. Similarly, $a_n=o_{\ell^2}(b_n)$ says that $a_n/b_n$ converges to $0$ in $\mathcal{L}^2$ norm.

\section{Mathematical model}\label{section background SST STFT}

Take a real time series $Y=\{Y_{i}\}_{i=1}^n$ that follows the model
\begin{eqnarray}\label{eq:model}
Y_{i} = f_{i}+\sigma \epsilon_{i}\,,
\end{eqnarray} 
where $f_{i}$ is a deterministic sequence,  $\sigma:=\sigma_n>0$ might depend on $n$, and $\epsilon_i$ is a random process modeling noise.  
We detail $f$ and $\epsilon$ term by term below.

\subsection{Adaptive harmonic model}
Consider an oscillatory function
\begin{eqnarray}\label{eq:mean}
f(t) =  \sum_{k=1}^KA_k(t) \cos( 2\pi \phi_k(t) ) 
\end{eqnarray}
where $\phi_k(t)\in C^2(\mathbb{R})$ is a monotonically increasing function and $A_k(t)\in C^1(\mathbb{R})$ is a positive function. Usually $A_k(t) \cos( 2\pi \phi_k(t) )$ is called the $k$-th {\em intrinsic mode type (IMT)} function, $A_k(t)$,  $\phi_k(t)$ and  $\phi_k'(t)$ are called {\em amplitude modulation} (AM), {\em phase}, and {\em instantaneous frequency} (IF) respectively of the $k$-th IMT function. $A_k(t)$ and $\phi'_k(t)$ model how strong and fast the $k$-th IMT function oscillates at time $t$.
We need the following {\em slowly varying} and {\em separation} assumption: 
\begin{assumption}\label{model assumption 1}
For a nonnegative small $\varepsilon<1$, we assume for each $k$,
\begin{itemize}
\item $\Xi_s<\phi_k'(t)<\Xi$ for all $t\in \mathbb{R}$, where $\varepsilon^{1/3}<\Xi_s<\Xi$ are constants;
\item $|A'(t)| \leq \varepsilon A(t)$ for all $t\in \mathbb{R}$; $\Xi_a\leq A_k(t)\leq \Xi_A$ for all $t\in \mathbb{R}$, where $\varepsilon^{1/3}<\Xi_a<\Xi_A$ are constants.
\item $|\phi_k''(t)|\leq \varepsilon \phi_k'(t)$ for all $t\in \mathbb{R}$ and $\sup_t|\phi_k''(t) |<M$ for some $M\geq 0$.
\item $\phi_k'(t)-\phi_{k-1}'(t)>\Xi_s$ for $k=2,\ldots,K$.
\end{itemize}
\end{assumption}
We call this model {\em the adaptive harmonic model} (AHM) \cite{Daubechies2011synchro}. See \cite{chen2014non} for the model's identifiability. When $\varepsilon=0$, $\phi_k(t)$ is linear and $A_k(t)$ is constant so that a IMT function is harmonic, which is commonly assumed in the traditional time series literature. However, in many practical applications, signals oscillate with time-varying speed and amplitude, making the harmonic model inadequate \cite{Daubechies2011synchro,chen2014non}. While more complex models for biomedical signals, such as those involoving non-sinusoidal oscillations \cite{CYLin2018,lin2021wave}, can be considered, they offer little additional insight for the UQ focus of this work. Therefore, the AHM is sufficient for our purposes.

We sample $f(t)$ at $t_i:=i/\sqrt{n}$ to model the discrete sequence $f_i$; that is,
\begin{align}\label{discretization of AHM function sqrt n Hz}
f_i=f(i/\sqrt{n})\,,
\end{align}
where $i=1,\ldots,n$; that is, sample $f$ over the interval $[0,\sqrt{n}]$ with a uniform sampling period of $1/\sqrt{n}$. In the frequency domain, this setup implies a Nyquist rate of $\sqrt{n}$ and a corresponding frequency resolution (or canonical frequency bin width) of $1/\sqrt{n}$.

\begin{remark}
In contrast to previous SST-related work, e.g., \cite{Daubechies2011synchro,chen2014non}, in Assumption \ref{model assumption 1} we introduce additional uniform lower and upper bounds on $\phi'(t)$ and $A(t)$ over $\mathbb{R}$. These conditions do not meaningfully restrict the model of interest; rather, they are adopted to simplify the exploration and simplify the derivation of our theorems that are uniform in nature. It is important to note that both the IF and AM are permitted to grow at the rate of order $\varepsilon \sqrt{n}$ as $n\to\infty$, so the upper bound constraints are essential for establishing the desired uniform error bounds unless alternative bounding conditions are imposed.
The lower bounds on frequency and amplitude serve to prevent degeneracy. Particularly, the subscript $s$ in $\Xi_s$ suggests ``separation'', which controls spectral leakage. When analyzing multiple components is of interest, the analysis is a straightforward generalization with assumption that IF of different components are separated by $\Xi_s$. Moreover, we modify the control of $A'$, replacing the condition $|A'(t)|\leq \varepsilon \phi'(t)$ with $|A'(t)|\leq \varepsilon A(t)$. This adjustment is made to highlight the role of the amplitude more clearly in the final theoretical results. 
\end{remark}

\subsection{Non-stationary noise (NSN) model}
Next, to put structures on $\epsilon_i$, consider the following filtration-based high dimensional non-stationary (HDNS) time series
\begin{equation}\label{noise model in the HDNS setup}
z_i=(z_{i,1},\ldots,z_{i,d}):=\mathcal{G}_i(\mathcal{F}_i)\in \mathbb{R}^d
\end{equation} 
for $i=1,\ldots,n$, where $d\in \mathbb{N}$, $\mathcal{G}_i=(\mathcal{G}_{i,1},\ldots,\mathcal{G}_{i,d}):\mathbb{R}^\infty\to \mathbb{R}^d$ is a measurable function, $\mathcal{F}_i:=(\ldots,e_{i-1},e_i)$ is a sequence of i.i.d. random variables $\{e_i\}$, $\mathcal{G}_i$ is the causal filtration mechanism at time $i$ that takes a sequence of i.i.d. random variables and outputs a $d$-dimensional random vector. When $d=1$, we simply say the time series satisfies the non-stationary noise (NSN) model.

 To quantify the dependence structure of $z_i$, consider an i.i.d. copy of $\{e_i\}$, denoted as $\{\hat{e}_i\}$, and define $\mathcal{F}_{i,i-k}:=(\ldots,e_{i-k-1},\hat{e}_{i-k},e_{i-k+1},\ldots,e_i)$. The temporal dependence is quantified by the entrywise uniform functional dependence measure: 
\begin{align}
\theta_{q,j}(k):=&\,\sup_{1\leq i\leq n}(\mathbb{E}[|\mathcal{G}_{i,j}(\mathcal{F}_i)-\mathcal{G}_{i,j}(\mathcal{F}_{i,i-k})|^{q}])^{1/q}\,,\label{definition of dependence measure}
\end{align}
and 
$\theta_{q}(k):=\max_{1\leq j\leq d}\theta_{q,j}(k)$, 
which uniformly quantify how the $k$-step historical input impacts the current output. 
We shall mention that the quantity $\theta_{q}(k)$ is different from the $L^r$ norm-based physical dependence measure, where $r\geq 2$, considered in \cite[Equation (3)]{Mies2024}. 
The cumulative tail dependence is needed to control the auto-covariance structure of the time series, and it is quantified by $\Theta_{q}(k):=\max\limits_{1\leq j\leq d}\sum_{l=k}^\infty \theta_{q,j}(l)$.

The noise $\epsilon_i$ is then modeled as \eqref{noise model in the HDNS setup} with $d=1$ so that it is a centered non-stationary noise process whose data generating mechanism may evolve both smoothly and abruptly over time. In general $z_i$ is not a martingale.

\begin{example}
The piecewise locally stationary (PLS) with $r$ break points (PLS($r$)) \cite{zhou2013,zhou2014} is a special case of the HDNS model. $\{\epsilon_{i}\}_{i = 1}^n$ is PLS($r$) if there exist constants $0 = s_0 < s_1 <\ldots< s_r < s_{r+1} = 1$ and $r+1$ measurable functions $\mathcal G_0,\ldots,\mathcal G_{r}$ as nonlinear causal filters such that 
$\epsilon_{i} = \mathcal G_j(t_i,\mathcal{F}_i), \quad \text{if} \quad s_j < t_i \leq s_{j+1}$,
$j=0,1,\cdots, r$, where $t_i = i/n$, and the time series is locally stationary \cite{dahlhaus1997fitting} between $s_j$ and $s_{j+1}$. The data generation mechanism changes abruptly at $s_j$, $j=1,2,\cdots, r$, which better models real-world noise and artifact. 
\end{example}

We need the following assumptions regarding the distribution behavior of $z_i$ and the physical dependence structure.
\begin{assumption}\label{model assumption 4}
Assume that for some $p>2$, 
$\theta_{p}(k)=O((k+1)^{-(\chi+1)}(\log(k+1))^{-A})$
for some constants $\chi>1$ and $A>0$ and
$\sup_{i\geq 0}\max_{1\leq j\leq d}\mathbb{E}[|z_{i,j}|^p]<B$
for some constant $B>0$; that is, $z_i$ fulfill the uniform finite $p$-th moment assumption for some $p>2$ and the dependence measure $\theta_{p}(k)$ decays polynomially. 
\end{assumption}

Note that when $z_i$ has a finite exponential moment uniformly; that is, $\sup_{1\leq j\leq d}\mathbb{E}[\exp(|z_{i,j}|)]<\infty$, or when the physical dependence decays exponentially; that is, $\theta_{p}(k)=O(\exp(-C(k+1)))$ for some constants $C>0$, we can obtain better bounds in the upcoming theorem. Since the proof technique is similar, we focus on the above assumption to simplify the discussion.

\section{Algorithm}
\subsection{discrete short-time Fourier transform (STFT)}
\label{sec:STFT summary}

Take a time series $\{X_l\}_{l\in \mathbb{Z}}$. The {\em discrete STFT} of $\{X_l\}_{l\in \mathbb{Z}}$ associated with a unit vector $\mathbf{h}\in \mathbb{R}^{2m+1}$, where $m\in \mathbb{N}$, as the window, is defined as
\begin{align}\label{STFT definition discrete}
\mathbf V^{(\mathbf{h})}_X(l,\eta)
:=\sum_{j=l-m}^{l+m}X_j \mathbf{h}(j-l)e^{-\ii 2\pi \eta(j-l)}\,,
\end{align}
where $l\in \mathbb{Z}$ and $\eta\in [0,1/2)$ is the frequency to explore. Here, $m$ is the length of truncation related to the kernel bandwidth and is chosen by the user.

\subsection{Continuous STFT and its discretization}
\label{sec:STFT summary}
We shall mention that the discrete STFT \eqref{STFT definition discrete} is directly related to the discretized version of the {\em continuous STFT}. Recall the continuous STFT:
\begin{align}\label{STFT definition continuous}
V^{(\mathsf{h})}_f(t,\xi):=\int_{-\infty}^\infty f(x)\mathsf{h}(x-t)e^{-\ii 2\pi \xi(x-t)}dx\,,
\end{align}
where $f$ is a continuous tempered distribution and $\mathsf{h}$ is a symmetric Schwartz function with unit $L^2$ norm and supported on $[-\beta,\beta]$, 
where $\beta>0$ is chosen by the user to control the spectral resolution. A more general setup is certainly possible; however, it does not provide additional insight into the main focus of this paper. Therefore, we adopt this specific setup for the purposes of this study.
In general, when the sampling rate is $q$ Hz and we have $n$ sampling points over the period $[0,n/q]$, \eqref{STFT definition continuous} is discretized by the Riemann sum as
\begin{align}\label{STFT definition continuous discretization}
V^{(\mathsf{h})}_f(t_l,\eta)\approx \frac{1}{q}\sum_{j=1}^n f(t_j)\mathsf{h}(t_j-t_l)e^{-\ii 2\pi \eta(t_j-t_l)}\,,
\end{align}
where $f(t_j)$ is set to $0$ when $t_j<0$ or $t_j>n/q$ (that is, pad the signal by $0$, $t_l=l/q$, $l=1,\ldots,n$, and $\eta\in[0, q/2)$ is the frequency of interest. We call the right hand side of \eqref{STFT definition continuous discretization} the {\em discretized STFT}. This discretization approximation $\approx$ holds when $q$ is sufficiently large, which we assume from now. We will precisely quantify it below for our analysis purpose.   Set $\mathbf{h}\in \mathbb{R}^{2\lceil\beta q\rceil+1}$ as
$\mathbf{h}(k)=\frac{1}{\sqrt{q}}\mathsf{h}\left(-\beta+\frac{k-1}{q}\right)\,, \ \ k=1,\ldots,2\lceil\beta q\rceil+1$,
where the normalization $\frac{1}{\sqrt{q}}$ guarantees that $|\mathbf{h}|$ is of order $1$ when $q$ is large, since $\sum_{k=1}^{2\lceil\beta q\rceil}\mathbf{h}(k)^2\to \int_{-\beta}^\beta|\mathsf{h}(t)|^2 dt=1$, when $q\to \infty$. In other words, we discretize the window $\mathsf{h}$ by $2\lceil\beta q\rceil+1$ points. 
In this setup, we connect the discretized STFT and the discrete STFT \eqref{STFT definition discrete} with a normalization factor $1/\sqrt{q}$; that is, if we set $X_j=f(t_j)$, $V^{(\mathsf{h})}_f(t_l,\eta)$ can be approximated by
\begin{align}\label{STFT relationship between discrete and continuous}
\frac{1}{\sqrt{q}}\sum_{j=l-\lceil\beta q\rceil}^{l+\lceil\beta q\rceil} X_j\mathbf{h}(j-l)e^{-\ii 2\pi \eta(t_j-t_l)}=:V^{(\mathbf{h})}_{X}(t_l,\eta)  =\frac{1}{\sqrt{q}} \mathbf V^{(\mathbf{h})}_X(l,\eta)\,.
\end{align}
In Section \ref{section GA of STFT on NSN}, we will show that under \eqref{eq:model}, the TFR of the discretized STFT is asymptotically a complex Gaussian random field on the TF domain.

\subsection{Synchrosqueezing transform (SST)}
\label{sec:SST summary}

SST is derived from the continuous STFT, and numerically implemented via a direct discretization. The {\em STFT-based SST} of $f$ with {\em resolution} $\alpha > 0$ is defined as
\begin{equation}\label{eq: SSTconti}
S^{(\mathsf{h})}_f(t,\xi) := 
\int_0^\infty
V^{(\mathsf{h})}_f(t,\eta)
\, g_\alpha\big(\xi- O^{(\nu)}_f(t,\eta)\big)d\eta \,,
\end{equation}
where $\xi>0$ is the frequency  and $ O^{(\nu)}_f(t,\eta)$ is the \emph{reassignment rule} defined by
\begin{equation}\label{eq: STFT-based reassignment rule cont}
 O^{(\nu)}_f(t,\eta)
:= 
\frac{-1}{2\pi i} 
\frac{V^{(\mathsf{h}')}_f(t,\eta)}{ V^{(\mathsf{h})}_f(t,\eta)}+\eta 
 \text{ if } |V^{(\mathsf{h})}_f(t,\eta)|> \nu\,,\ \mbox{ and } 
-\infty 
\ \text{ otherwise}\,,
\end{equation}
where $\nu\geq0$ is the chosen threshold, $g_\alpha:\mathbb{C}\to \mathbb{R}$ approximates the $\delta$ measure with the support at $\{0\}$ when $\alpha\to 0$.
The parameter $\alpha$ dictates SST's resolution in the frequency-axis.  For clarity, below we adopt ${g}_\alpha(z) = \frac{1}{\sqrt{\pi\alpha}} e^{-|z|^2/\alpha}$, which has $L^1$ norm $1$. SST's nonlinearity arises from reassignment rule, which extracts instantaneous frequency information \cite{Daubechies2011synchro}. 
Compared with the STFT, SST enhances the TFR contrast by leveraging phase information, mitigating uncertainty-principle effects under oscillatory conditions such as the AHM. This facilitates detection of oscillatory components and more accurate estimation of instantaneous frequency, amplitude modulation, and denoising via the reconstruction formula.

Numerically,  when the sampling rate is $q$ Hz and we have $n$ sampling points sampled at $i/q$, where $i=1,\ldots,n$, over the period $[0,n/q]$, STFT-based SST is implemented by a direct discretization.
Define $\mathbf{Dh}\in \mathbb{R}^{2\lceil \beta q\rceil+1}$ as $\mathbf{Dh}(k)=\frac{1}{\sqrt{q}}\mathsf{h}'\left(-\beta+\frac{k-1}{q}\right)$, where $k=1,\ldots,2\lceil\beta q\rceil+1$. 
The STFT-based SST with threshold $\nu\geq 0$ is thus numerically implemented by the Riemann sum
\begin{equation}\label{eq: SST}
S^{(\mathbf{h})}_X(t_l,\xi_j) := 
\frac{C}{d}\sum_{k=1}^{d}
V^{(\mathbf{h})}_X(t_l,\eta_k)
\, g_\alpha\big(\xi_j- O^{(\nu)}_X(t_l,\eta_k)\big) \,,
\end{equation}
where $d\in \mathbb{N}$ means $d$ uniform points on $[0,C)$ in the frequency axis, $0<C\leq \frac{q}{2}$ indicates the chosen spectral range of interest, $\xi_j\in [0, q/2)$ is the positive frequency we have interest, and $ O_X(t_l,\eta_k)$ is the \emph{reassignment rule} defined by
\begin{equation*}
 O^{(\nu)}_X(t_l,\eta_k)
:= 
\frac{-1}{2\pi i} 
\frac{V^{(\mathbf{Dh})}_X(t_l,\eta_k)}{ V^{(\mathbf{h})}_X(t_l,\eta_k)}+\eta_k 
\ \text{ if } |V^{(\mathbf{h})}_X(t_l,\eta_k)|> \nu\,,\ \mbox{ and }
-\infty \ \text{ otherwise}.
\end{equation*}
$C$ can be chosen as large as $q/2$ when no background information is available, or it is set to a sufficiently large constant depending on the application. 

\begin{remark}
The choice of $\xi_l$ in practice is important. While a natural choice based on Nyquist-Shannon sampling and the window $\mathsf{h}$ is $\xi_l=\frac{l}{2\beta}$ for $l=1,\ldots,\lceil \beta q\rceil$, results from discrete Fourier analysis suggest that a finer grid might mitigate spectral leakage \cite{Genton_Hall:2007}. Empirically, we find that using a finer grid in SST is beneficial, although the optimal choice and its interaction with $\alpha$ remain unexplored. These considerations motivate alternative grids, such as $\xi_l= l/(2\beta q)$ for $l=1,\ldots,\lceil \beta q^{3/2}\rceil$, which may offer both practical and theoretical advantages. This topic lies beyond the scope of this paper and will be pursued in future work.
\end{remark}

\subsection{Signal reconstruction by SST}

To establish UQ of TFRs determined by STFT and SST, we require an algorithm capable of accurately separating the underlying oscillatory signal satisfing the AHM. To this end, we consider a SST-based reconstruction and aim to demonstrate the robustness of SST-based signal reconstruction under the AHM framework.

Recall that when $\mathsf{h}(0)\neq 0$, the SST-based reconstruction formula \cite{Wu2011adaptive} for recovering each IMT function, $f_k(t)=A_k(t)\cos(2\pi \phi_k(t))$, from a noisy observation $Y=f+\Phi$, where $f$ satisfies the AHM and $\Phi$ is a tempered distributed random process modeling the noise \cite{chen2014non,sourisseau2022asymptotic}, is defined as
\begin{align}\label{alogithm:sst:reconstruction continuous formula}
{f}_k^{\mathbb{C}}(t):=\frac{1}{\mathsf{h}(0)}\int_{\xi\in 
R_{k,t}}  S^{(\mathsf{h})}_Y(t,\xi)d\xi\,,
\end{align}
where $R_{k,t}:=[-\Delta_r,\,\Delta_r]+ \phi_k'(t)$ with $\Delta_r>\varepsilon^{1/3}$ a small constant chosen by the user.\footnote{In the literature \cite{Daubechies2011synchro,chen2014non}, $R_t$ is defined with $\Delta_r=\varepsilon^{1/3}$ when $\varepsilon>0$. To handle the degenerate case that $\varepsilon=0$, we modify the reconstruction formula and consider a non-degenerate $\Delta_r$.} 
In practice, $\phi_k'(t_l)$ can be estimated using ridge extraction algorithms applied to the TFR derived from the STFT \cite{laurent2021novel} or the SST \cite{chen2014non}. 
It is important to note that while several ridge extraction algorithms have been developed and are empirically robust and accurate to estimate the IF, a quantitative analysis of these algorithms remains lacking. Since analyzing ridge extraction algorithm is out of the scope of this paper, in the following analysis, we assume that the ridge extraction algorithm is sufficiently accurate and we can robustly estimate $\phi_k'(t)$. We refer readers to \cite{liu2024analyzing} for a review of existing algorithms and recent effort in ridge analysis.

In practice, we discretize \eqref{alogithm:sst:reconstruction continuous formula} directly to analyze a sampled version of $Y$ or a time series. 
Set $m = \lceil \beta q \rceil \asymp \sqrt{n}$. 
Numerically implement STFT via the discretized STFT $V^{(\mathbf{h})}_Y(t_l,\eta_j)$ defined in \eqref{STFT relationship between discrete and continuous} at time $t_l$ and frequency $\eta_j=j\Xi/d$, where $j=1,\ldots,d$, and $d\in \mathbb{N}$ is the number of frequency bins. 
Next, the SST is computed as $S^{(\mathbf{h})}_Y(t_l,\xi_k)$ according to \eqref{eq: SST}, where $C=\Xi$, $\xi_k=k\Xi/d$ for $k=1,\ldots,d$ and the threshold $\nu\geq0$.
Finally, the SST reconstruction formula  is implemented by
\begin{align}
\widetilde{f}^{\mathbb{C}}(t_l):=\frac{1}{\mathsf{h}(0)}\frac{\Xi}{d}\sum_{\xi_k\in 
R_l}  S^{(\mathbf{h})}_Y(t_l,\xi_k) \,,\label{alogithm:sst:reconstruction}
\end{align} 
which is a direct discretization of \eqref{alogithm:sst:reconstruction continuous formula}.

While pointwise robustness for a similar reconstruction formula based on the continuous wavelet transform has been established in \cite{chen2014non}, to our knowledge, no robustness results are available for \eqref{alogithm:sst:reconstruction continuous formula} under general nonstationary noise, let alone uniform robustness.
In Section \ref{section:reconstruction formula proof}, we establish a robustness of the reconstruction $\widetilde{f}^{\mathbb{C}}$ satisfying a uniform error bound in time with high probability when $Y_i = f_i + \sigma \epsilon_i$ is defined in \eqref{eq:model} following the discretization scheme \eqref{discretization of AHM function sqrt n Hz}. The noise-free case follows by letting $\sigma \to 0$. 
The robustness of the reconstruction formula also ensures the robustness of AM and phase reconstruction; that is, we estimate $A_k(l/\sqrt{n})$ by $|\widetilde{f}_k^{\mathbb{C}}(l)|$, and the phase $\phi_k(l/\sqrt{n})$ by unwrapping $\widetilde{f}_k^{\mathbb{C}}(l)/|\widetilde{f}_k^{\mathbb{C}}(l)|$, with the error being uniformly controlled.

\subsection{Numerical implementation}
For numerical implementation, we recommend the following practical guidelines for an input signal $X_{1},\ldots,X_{n}$. For the STFT, use a truncated Gaussian window $\mathsf{h}$. When no prior information is available, the R\'enyi entropy \cite{sheu2017entropy} may be applied to select $\beta$; otherwise, choose $\beta$ so that $\mathsf{h}$ spans approximately 8-15 cycles of the target oscillatory component. For the SST, we suggest setting $\nu=10^{-6}\times \texttt{std}(X_{1},\ldots,X_{n})$, and $\alpha=10/\sqrt{n}$. For reconstruction, start from the lowest-frequency component and proceed iteratively. A practical choice is $\Delta_r=\sqrt{\alpha} \min\limits_{i=1,\ldots,n}\{\phi'_2(i/\sqrt{n})-\phi'_1(i/\sqrt{n}),\, \phi'_1(i/\sqrt{n})\}$. After reconstructing the lowest-frequency component, repeat the procedure for higher components. Empirical evidence suggests that the reconstruction is not sensitive to these parameter choices. A systematic study of optimal parameter selection is beyond the scope of this work and is left for future investigation.

\section{Uncertainty quantification of SST by bootstrapping}

In many real-world settings, oscillatory components may or may not be present, and their onset times and signal-to-noise ratios are typically unknown. A key practical goal is therefore to determine whether an oscillatory component is present in the TFR obtained by the STFT or SST and to quantify the associated uncertainty.
Recovering the underlying model of $\epsilon_i$ is generally challenging. Nonetheless, when $\epsilon_i$ satisfies mild regularity conditions, particularly local stationarity, we can effectively approximate the noise via a time-varying autoregressive (tvAR) approach \cite{ding2023autoregressive} while preserving its covariance structure. By our Gaussian approximation result for the STFT, the STFT of this tvAR process can be uniformly approximated in time and frequency by that of a Gaussian tvAR process with the same coefficients.
Although the SST is nonlinear, the preserved covariance structure under the Gaussian approximation implies that the distributional behavior of the SST of $\epsilon_i$ can be well approximated by that of a Gaussian tvAR process. 
Motivated by this observation, we propose a bootstrap algorithm for $\epsilon_i$ under the locally stationary assumption. This yields principled UQ for both the STFT and the SST and leads naturally to a noise-thresholding framework.
A variety of bootstrap methods for locally stationary processes have been proposed, e.g., wavelet-based approach \cite{nason2013test}, STFT-based approach \cite{KreissPaparoditis2015,HafnerKirch2017}, moving block approach \cite{sourisseau2022asymptotic,LinSongvanderSluis2025}, and singular spectrum approach \cite{Poskitt2025}.
Empirically, we find these approaches perform comparably for bootstrapping the SST when the noise is estimated via the reconstruction formula \eqref{alogithm:sst:reconstruction}. We focus on the tvAR approach here because it integrates most naturally with our theoretical framework and allows us to establish convergence rates. See Section \ref{section bootstrap theorem proof}.

\subsection{Bootstrap the noise}\label{section bootstrap algorithm}
Suppose the given time series is $\{X_i\}_{i=1}^n$, and consider the null hypothesis $f_i=0$. If the null hypothesis is known to hold, we simply set $\tilde{f}_i=0$ and $\tilde{\epsilon}_i:=X_i$.
Otherwise, we first apply the SST to reconstruct candidate IMT components and denote the reconstructed $k$-th component by $\tilde{f}_{k,i}$, assuming that the number of components $K$ is known.\footnote{Estimating $K$ directly from the time series without prior knowledge, particularly under our model, is a challenging problem. Since assuming knowledge of $K$ is reasonable in many biomedical time series applications, and investigating this open problem does not directly advance our goal of UQ, we leave it for future work.} Then, estimate the noise by $\tilde{\epsilon}_i:=X_i-\sum_{k=1}^K\tilde{f}_{k,i}$. 

Given $\tilde{\epsilon}_i$, we approximate the error process $\{\epsilon_n\}$ by a tvAR process $\{x_n\}$ with short range dependence:
$x_i=\sum_{j=1}^b\phi_{j}(i/n)x_{i-j}+\varepsilon_i$,
where $\{\varepsilon_i\}$ is a locally stationary white noise, $b\in \mathbb{N}$ is the tvAR order and $\phi_j$ is a smooth function on $[0,1]$ with proper conditions. To approximate $\phi_j$, consider an orthonormal basis $\{\psi_i\}_{i=1}^m$ of a finite dimensional subspace of $L^2[0,1]$, so that $P_S\phi_j\approx \phi_j$ for $j=1,\ldots,b$ with $S:=\texttt{span}\{\psi_i\}_{i=1}^m$. Then,
$x_i\approx \sum_{j=1}^b\sum_{k=1}^ma_{ik}\psi_{j}(i/n)x_{i-j}+\varepsilon_i$.
$a_{ik}$ are estimated via linear regression on $\{\tilde{\epsilon}_i\}_{i=1}^n$, yielding
$\tilde{\phi}_{j}(i/n)=\sum_{k=1}^m\tilde{a}_{ik}\psi_{j}(i/n)$, and
the innovation process is estimated as
$\tilde{\varepsilon}_i:=x_i-\sum_{j=1}^b\tilde{\phi}_{j}(i/n)\tilde{\epsilon}_{i-j}$
when $i=b+1,\ldots,n$ and when $i=1,\ldots,b$, set $\tilde{\varepsilon}_i=\tilde{\epsilon}_i$ or estimate it by the reverse process. The time-varying standard deviation of the estimated innovation process is estimated via local averaging process:
$\tilde{\sigma}_i=\texttt{STD}\{ \tilde{\varepsilon}_j,\, \max\{1,i-I\}\leq j\leq \min\{n, i+I\}\}$, 
where $I\in \mathbb{N}$ typically set to $20$ but can be optimized. 
To generate bootstrap replicates $\{\epsilon^{(*m)}_i\}$, draw i.i.d. standard Gaussian $\eta^{(*m)}_i$ and set
$\epsilon^{(*m)}_i:=
\tilde{\sigma}_i\eta^{(*m)}_i$  when $i=1,\ldots,b$, and 
$\epsilon^{(*m)}_i:=\sum_{j=1}^b\tilde{\phi}_{j}(i/n)\epsilon^{(*m)}_{i-j}+\tilde{\sigma}_i\eta^{(*m)}_i$ when $i=b+1,\ldots,n$.
This yields $M$ bootstrap replicates that preserve the locally stationary structure of the original process.

\subsection{Application to uncertainty quantification of SST}
The first application concerns quantifying the uncertainty of SST. Given the time series $\{X_i\}_{i=1}^n$, we run SST on $X_i^{(*m)}:=\tilde{f}_i+\epsilon^{(*m)}_i$, $m=1,\ldots,M$, over a grid in the time-frequency domain, denoted by $\mathcal{G}:=\{t_1,\ldots,t_{n'}\}\times\{\xi_1,\ldots,\xi_{d}\}\subset \mathcal{S}:=\{1/\sqrt{n},2/\sqrt{n}\ldots,\sqrt{n}\}\times \{1/\sqrt{n},\ldots,\sqrt{n}/2\}$, where $d$ is guided by Theorem \ref{Bootstrapping main theorem}. 
Clearly, $\mathcal{S}$ is the full grid that hosts $n$ sample points in the time domain and $n/2$ canonical frequencies in the frequency domain.
In practice, we set $n'=n$ and $t_l=l/\sqrt{n}$ in $\mathcal{G}$, or if we aim to speed up the algorithm, a sparser grid can be chosen. At each point in $\mathcal{G}$, we compute the empirical $(\alpha/2)$-th and $(1-\alpha/2)$-th percentiles, where $\alpha>0$ is chosen by the user, based on the $M$ realizations, $\{|S^{(\mathbf{h})}_{X^{(*m)}}(t_j,\xi_k)|\}_{m=1}^M$.
These percentiles are then interpolated from $\mathcal{G}$ to the full grid $\mathcal{S}$ using cubic splines. Denote the resulting interpolated function as $|S|_{X,\alpha/2}\in \mathbb{R}^{n\times \lfloor n/2\rfloor}$ and $|S|_{X,1-\alpha/2}\in \mathbb{R}^{n\times \lfloor n/2\rfloor}$. Plotting $|S|_{X,\alpha/2}$ and $|S|_{X,1-\alpha/2}$ alongside $|S^{(\mathbf{h})}_X|$ provides a visual representation of the uncertainty of SST in the presence of noise contamination by $\epsilon_i$. See Section \ref{section numerical simulation} for details.

The second application focuses on noise thresholding. We run SST on $\epsilon^{(*m)}_i$, $m=1,\ldots,M$, over the grid $\mathcal{G}$. At each point in $\mathcal{G}$, we estimate a threshold corresponding to a user-specified $(1-\alpha)$ confidence level, where $\alpha>0$ is chosen by the user. This threshold is then interpolated to the full grid $\mathcal{S}$ and denoted as $T\in \mathbb{R}^{n\times \lfloor n/2\rfloor}$. We then threshold the SST coefficients of $X_i$ by setting coefficients below the threshold to zero. The resulting thresholded TFR, denoted as $S^{(\mathbf{h})}_{X,T}\in \mathbb{R}^{n\times \lfloor n/2\rfloor}$, is given by $S^{(\mathbf{h})}_{X,T}(i,k):=S^{(\mathbf{h})}_X(t_i,\xi_k)I(|S^{(\mathbf{h})}_X(t_i,\xi_k)|\geq T(i,k))$, where  $I(\cdot)$ is the indicator function. A similar bootstrapping procedure can be applied to quantify the STFT uncertainty and to determine appropriate thresholds.

Third, we construct simultaneous confidence regions (SCR) statistics \cite{yang2022spectral} to detect the presence of oscillatory components. Choose the index grid $\mathcal{G}:=\{1,\ldots,n\}\times\{\lfloor n^{2/3}\rfloor,\lfloor2n^{2/3}\rfloor\ldots,\lfloor n/2\rfloor\}$. Define the STFT-SCR and SST-SCR statistics of the input signal as $r_V:=\max\limits_{(i,j)\in \mathcal{G}}\{|V(i,j)|\}$ and $r_S:=\max\limits_{(i,j)\in \mathcal{G}}\{|S(i,j)|\}$, where $V_{X}^{(\bf{h})}$ and $S_{X}^{(\bf{h})}$ are the STFT and SST of $X_i$, respectively.
Generate $M$  bootstrap replicates of the null noise process $\{\epsilon^{(*m)}_{i}\}_{i=1}^{n}$, where $m=1,\ldots,M$, and the associated STFT and SST, denoted as $V_{X}^{(*m)}$ and $S_{X}^{(*m)}$, which leads to the bootstrapped STFT-SCR and SST-SCT, denoted as $r_V^{(*m)}:=\max\limits_{(i,j)\in \mathcal{G}}\{|V_{X}^{(*m)}(i,j)|\}$ and $r_S^{(*m)}:=\max\limits_{(i,j)\in \mathcal{G}}\{|S_{X}^{(*m)}(i,j)|\}$. At level $\alpha\in (0,1)$, reject the null hypothesis that $f=0$ using STFT if $r_V>Q_\alpha(\{r_V^{(*m)},\,m=1,\ldots,M\})$, or using SST if $r_S>Q_\alpha(\{r_S^{(*m)},\,m=1,\ldots,M\})$, where $Q_\alpha$ is the $\alpha$-percentile.

\section{Numerical results}\label{section numerical simulation}

We focus on the non-null case and real signal. More results, including the null case, is postponed to Section \ref{Section more numerical results supp: null case}. The Matlab implementation to reproduce the results can be found in \url{https://github.com/hautiengwu2/UQ-SST}. In this section, we fix $M=1000$ in all bootstraps.

\subsection{Non-null case}

Generate simulated oscillatory signals in the following way. Discretize a standard Brownian motion, denoted as $B_i$, where $i=1,\ldots,n$, and convolve it with a chosen kernel $K$ with a support of $700$ points. Denote the resulting random process by $b_i$ and obtain $A(i/\sqrt{n}):=3+b_i/\|b_i\|_{\ell_\infty}$, where $\sqrt{n}>0$ is the sampling rate. Similarly, construct a monotonically increasing function by taking another standard Brownian motion $B'_i$, $i=1,\ldots,n$, that is independent of $B_i$ and smoothen it with a chosen kernel $K'$ with a support of $500$ points. Denote the resulting random process by $p_i$. Construct $\varphi'(i/\sqrt{n}):=4+\frac{0.5i}{17\sqrt{n}}+1.2\frac{p_i}{\|p_i\|_\infty}$. The monotonic random process, denoted as $\varphi(i/\sqrt{n})$, is obtained via normalized cumsum; that is, $\varphi(i/\sqrt{n})=\frac{1}{\sqrt{n}}\sum_{j=1}^i\varphi'(j/\sqrt{n}) $. The oscillatory signal is constructed by $f(i/\sqrt{n}):=\sum_{k=1}^2A_k(i/\sqrt{n})\cos(2\pi\varphi_k(i/\sqrt{n}))$. This construction of oscillatory signal is in practice closer to real world data and has been considered in various time-frequency analysis literature. The noise $\epsilon_i$ is constructed in the following way. Take an i.i.d. Gaussian process $\eta_i$ with standard deviation $1$, and construct a tvAR process via $\epsilon_i=\eta_i$ when $i=1,2$, and $\epsilon_i=\sum_{k=1}^2\phi_k(i/n)\epsilon_{i-k}+(1+0.5\cos(2\pi i/n))\eta_i$ when $i=3,4,\ldots,n$, where $\phi_1(i)=-0.5(0.7+0.3\cos(2\pi i/n))$ and $\phi_2(i)=0.3\sqrt{0.1+i/(4n)}$. The final random process is $X^{(a)}_i:=af(i/\sqrt{n})+\epsilon_i$, where $a\geq0$ is the global signal strength. When $a=0$, it is the null case, otherwise nonnull. We assume the knowledge of two oscillatory components.

With one realization of the simulated signal $X^{(1)}_i$, we apply \eqref{alogithm:sst:reconstruction} to reconstruct each IMT function and hence $f$, denoted as $\tilde{f}(i/\sqrt{n})$,  and subtract it from the signal to get the reconstructed noise, denoted as $\tilde{\epsilon}_i$. Then, bootstrap the noise on $\tilde{\epsilon}_i$ for $M\geq 1$ times, denoted as ${\epsilon}^{(*m)}_{i}$, where $m=1,\ldots,M$. 
The noisy signal and the reconstructed deterministic signal are shown in Figure \ref{Figure:Signalrecon nonnull}, and the associated TFRs are shown in the left panel of Figure \ref{Figure: BS P SST nonnull}. Clearly, SST exhibits a wider dynamic range and greater concentration than STFT, a direct consequence of the reassignment step. With the reconstructed noise, the bootstrap result is shown in Figure \ref{Figure:tvARapprox nonnull}. 
The QQ plots of distributions of STFT and SST of true noise and bootstrapped noise are shown in Figures \ref{FigureS4:STFT BS real nonnull} and \ref{FigureS5:SST BS real nonnull} respectively. 
The result supports the validity of combining SST-based reconstruction with the bootstrap approach.
When $M=1000$, the runtime of bootstrap is $476\pm 15$ s in Matlab R2025a on a 2017 MacBook Pro (3.1GHz Quad-Core Intel Core i7).

\begin{figure}[bht!]
\centering
\includegraphics[trim=0 0 0 0,clip,width=0.85\textwidth]{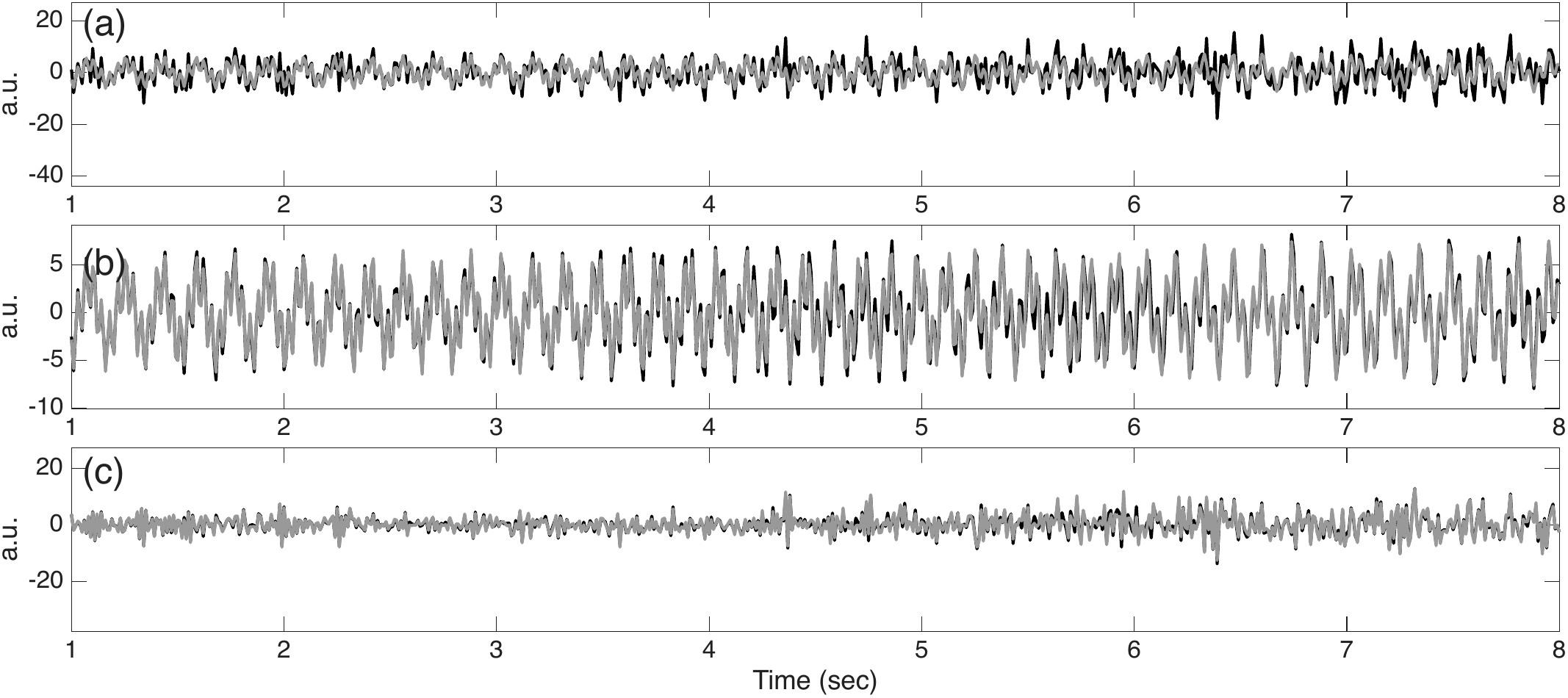}
\caption{\small An 8s segment of the noisy signal and reconstruction results. 
(a) one realization of the noisy signal (gray) and the corresponding true deterministic signal (gray). (b) reconstructed deterministic signal (gray) and the true deterministic signal (gray). (c) reconstructed noise (gray) and the true realized noise (gray).\label{Figure:Signalrecon nonnull}}  
\end{figure}

\begin{figure}[bht!]
\centering
\includegraphics[trim=0 0 0 0,clip,width=0.85\textwidth]{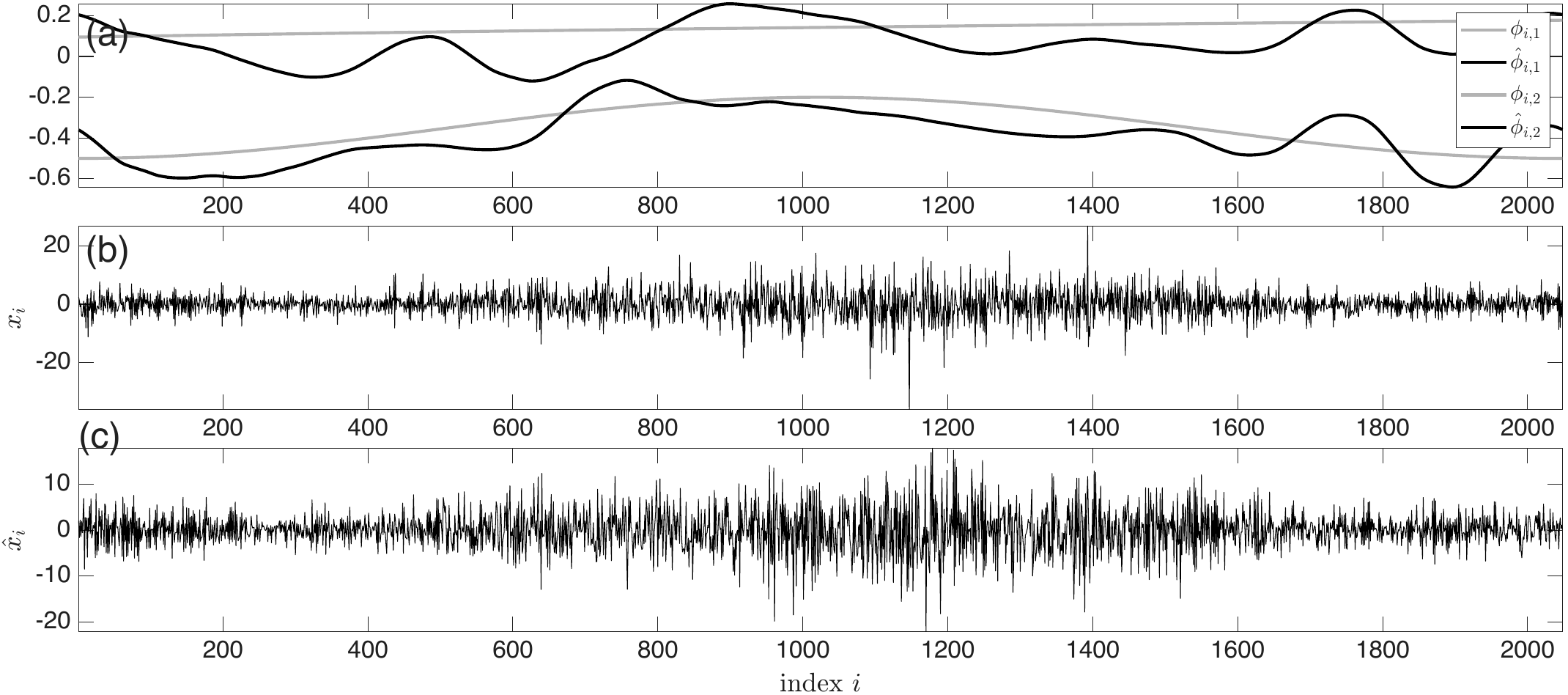}
\caption{\small (a) the true tvAR coefficients, and the estimated coefficients of the approximate tvAR process. (b) one realization of $x_i$ and the reconstructed $x_i$. (c) the bootstrapped $x_i$.\label{Figure:tvARapprox nonnull}}  
\end{figure}

The noise thresholding results with $X^{(1)}_i$ are shown in Figure \ref{Figure: BS P SST nonnull}, where the threshold is determined based on the 99th percentile of the bootstrapped noise values, $\{\epsilon^{(*m)}_{i}\}_{i=1}^{n}$, for $m=1,\ldots,1000$, using $\tilde{\epsilon}_i$.
In the TFRs of the noisy signal obtained via STFT and SST (Figures \ref{Figure: BS P SST nonnull}(a) and (d), respectively), background speckles are visible, particularly around 5-15 s above 30 Hz, which are attributable to noise. Using bootstrapped UQ from pure noise TFRs, we obtain a statistically reliable threshold (Figures \ref{Figure: BS P SST nonnull}(b) and (e)). Applying this threshold results in cleaner TFRs (Figures \ref{Figure: BS P SST nonnull}(c) and (f)), in which noise-induced speckles are effectively suppressed. The denoising effect is especially prominent in the SST.

\begin{figure}[bht!]
\centering
\includegraphics[trim=0 0 0 0,clip,width=0.8\textwidth]{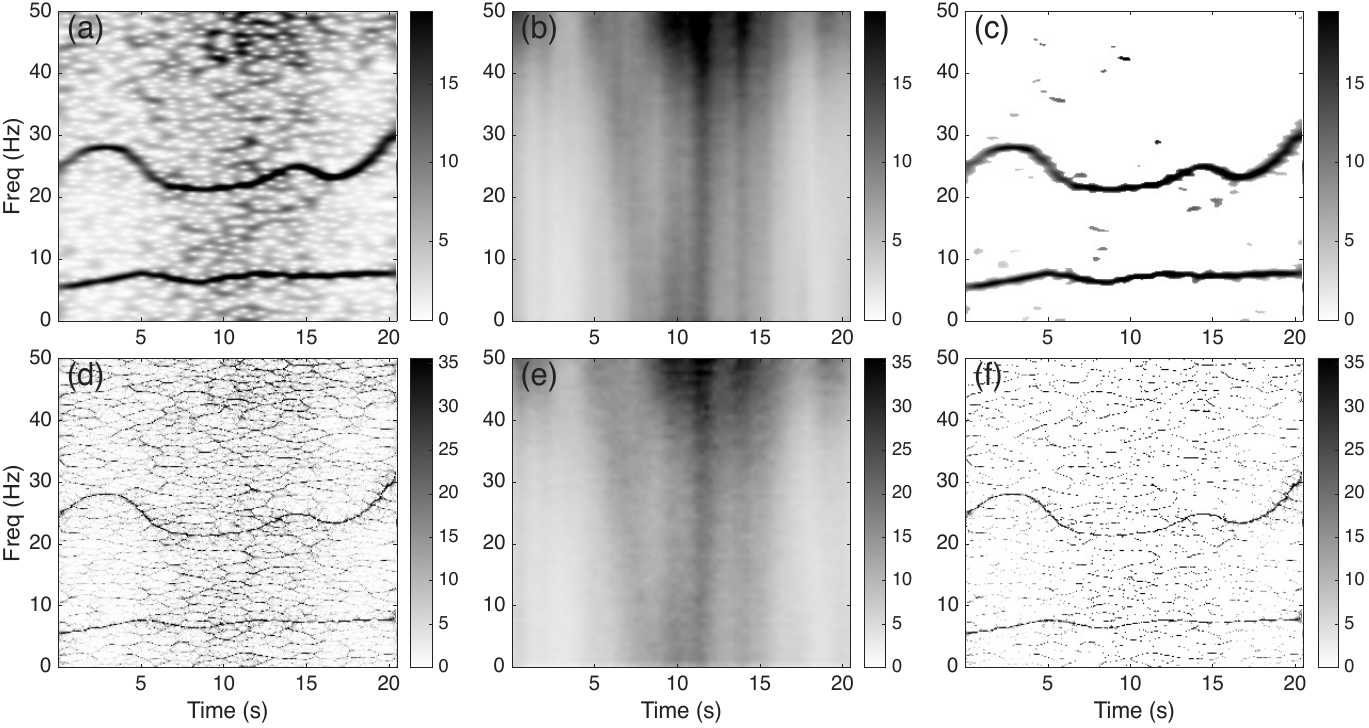}
\caption{\small Thresholding results. (a) The TFR of a nonnull signal determined by STFT. (b) 99\% percentile of the bootstrapping as the threshold, with interpolation to the whole grid. (c) Thresholded (a) by (b).
(d) The TFR of a nonnull signal determined by SST. (e) 99\% percentile of the bootstrapping as the threshold, with interpolation to the whole grid. (f) Thresholded (d) by (e).\label{Figure: BS P SST nonnull}}  
\end{figure}

The bootstrap-based UQ results are shown in Figure \ref{Figure: UQ P SST nonnull}. We generate bootstrapped signals, denoted as $\tilde{x}^{(*m)}_i:=\tilde{f}(i/\sqrt{n})+{\epsilon}^{(*m)}_{i}$, where $m=1,\ldots,1000$, and run STFT and SST on $\tilde{x}^{(*m)}_{i}$. 
Figures \ref{Figure: UQ P SST nonnull}(a)-(b) and (c)-(d) jointly depict the pointwise 95\% confidence intervals (CI) of the TFRs obtained from STFT and SST respectively, using $1000$ realizations of the random process model, while panels (e)-(f) and (g)-(h) present analogous 95\% CIs derived from bootstrap resampling of the noise and reconstructed signal.
Notably, the bootstrap-based CIs closely match those obtained from the true model. While two curves associated with the true IFs are visible in the noisy TFR in the left panel of Figure \ref{Figure: BS P SST nonnull}, the 97.5\% confidence plots shown in Figures \ref{Figure: UQ P SST nonnull}(f) and (h) offer further evidence about their existence. The difference between the 2.5\% confidence plots of the STFT and SST, particularly the presence of visible ridges in the STFT plot but less clear in the SST plot over regions corresponding to the true IFs, reflects the theoretical understanding that SST sharpens and concentrates the TFR.

\begin{figure}[bht!]
\centering
\includegraphics[trim=0 0 0 0,clip,width=0.99\textwidth]{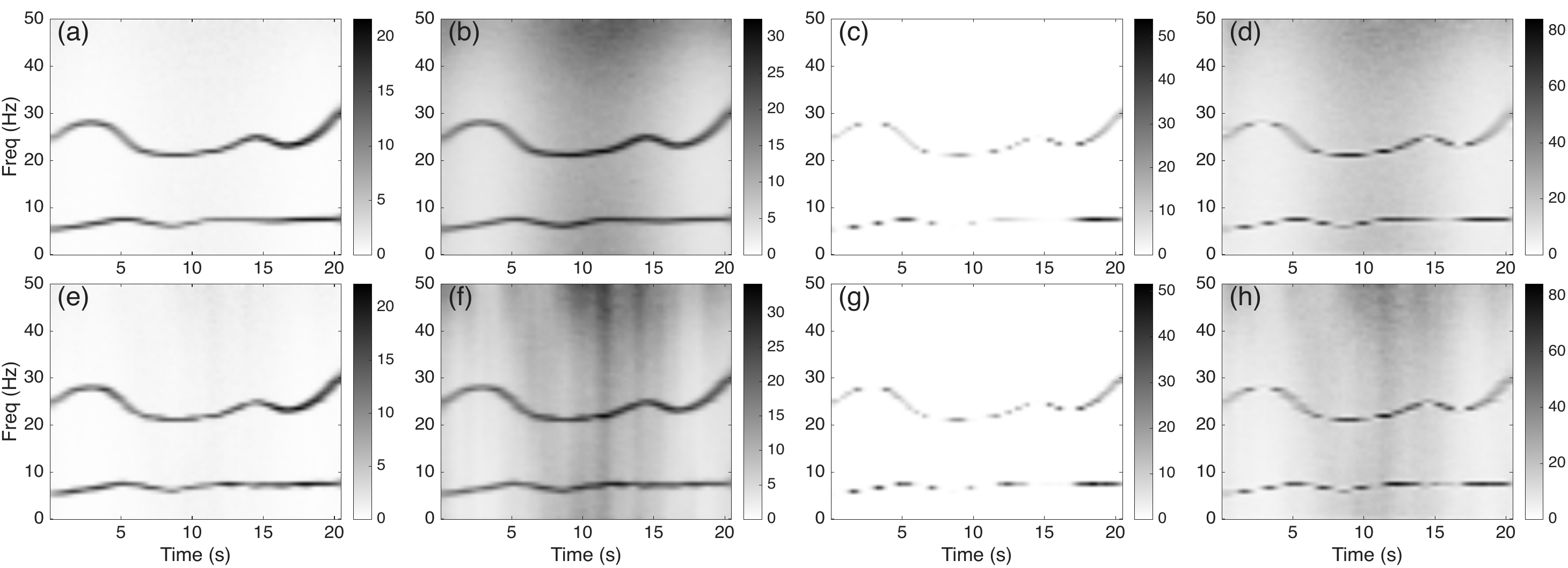}
\caption{\small UQ of TFR. Panels (a)-(b) and (c)-(d) jointly depict the pointwise 95\% confidence intervals of the TFRs obtained from STFT and SST respectively, using $5000$ realizations of the random process model. Panels (a) and (c) show the 2.5\% percentile maps, whereas (b) and (d) are the corresponding 97.5\% percentiles. 
Panels (e)-(f) and (g)-(h) present analogous 95\% confidence intervals derived from bootstrap resampling of the noise and reconstructed signal (5 000 repetitions). Panels (e) and (g) show the 2.5\% percentile maps, and (f) and (h) the 97.5\% percentile maps.\label{Figure: UQ P SST nonnull}}  
\end{figure}

Finally, we demonstrate the application of the bootstrap procedure to detect existence of oscillatory signals using the STFT-SCR and SST-SCR statistics. Consider $a=k/4$, where $k=0,1,\ldots,8$ in the signal $X^{(a)}_i$. With the significance level $0.05$, the simulated rejection rate result is shown in Figure \ref{Figure: SCR result}, where we compare the considered bootstrap approach \cite{ding2023autoregressive} with others, including \cite{HafnerKirch2017} and \cite{Poskitt2025}. We can clearly see that SST-SCR has a higher rejection rate compared with the STFT-SCR, and the considered bootstrap overall behaves better.

\begin{figure}[bht!]
\centering
\begin{minipage}{0.5\textwidth}
\includegraphics[trim=0 0 0 0,clip,width=1\textwidth]{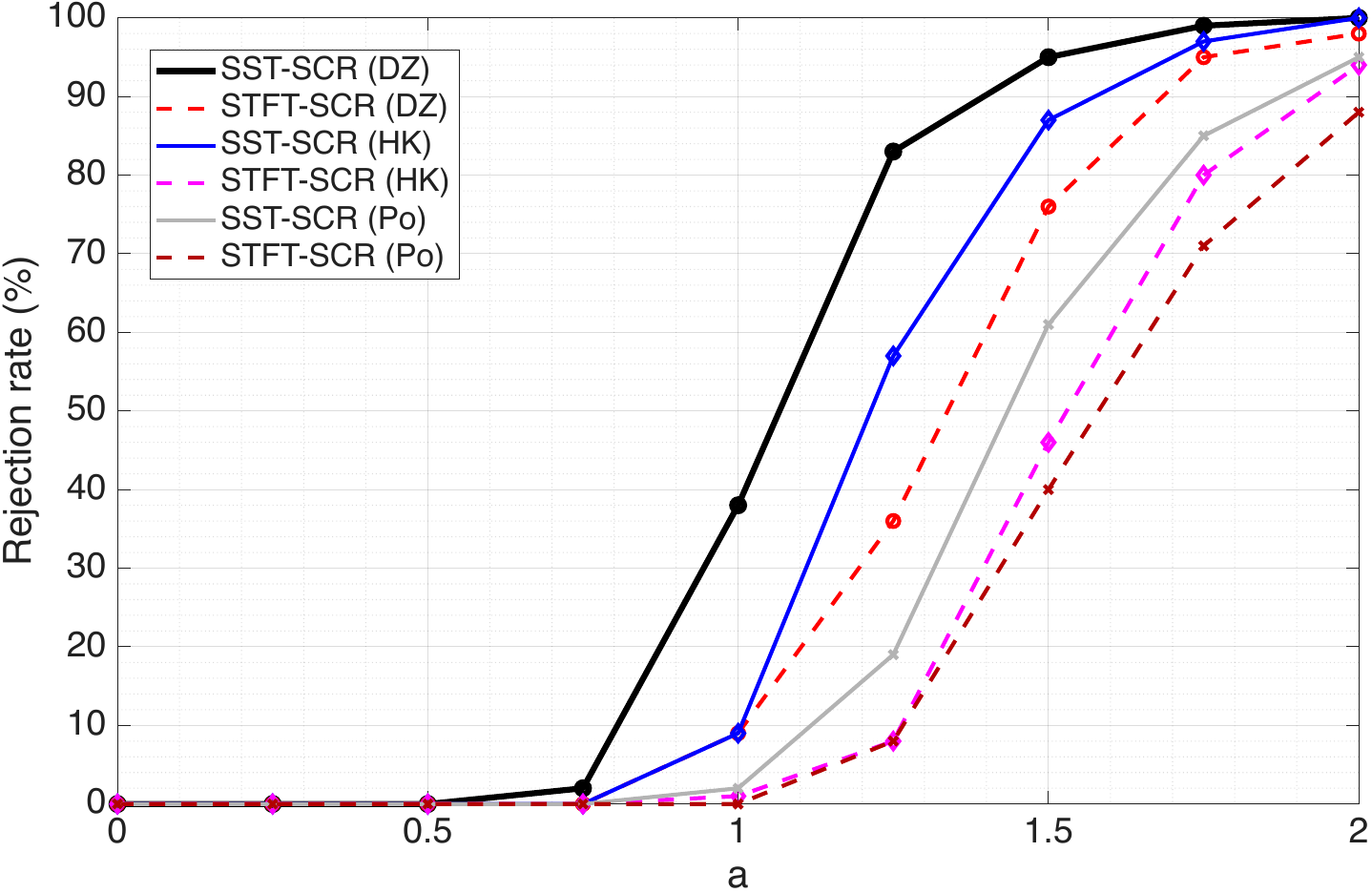}
\end{minipage}
\begin{minipage}{0.4\textwidth}
\caption{\small The rejection rate of the STFT-SCR and SST-SCR over a series of simulated signals, with the signal amplitude $a$ ranging from 0 to 2. The dashed (solid resp.) curves are based on the STFT-SCR (SST-SCR resp.) with different bootstrapping algorithms, where DZ, HK, and Po indicates bootstrapping algorithm from \cite{ding2023autoregressive}, \cite{HafnerKirch2017}, and \cite{Poskitt2025}.  \label{Figure: SCR result}}  
\end{minipage}
\end{figure}

\subsection{Application to sleep spindle analysis}

\begin{figure}[bht!]
\centering
\includegraphics[trim=5 175 5 10, clip,width=\textwidth]{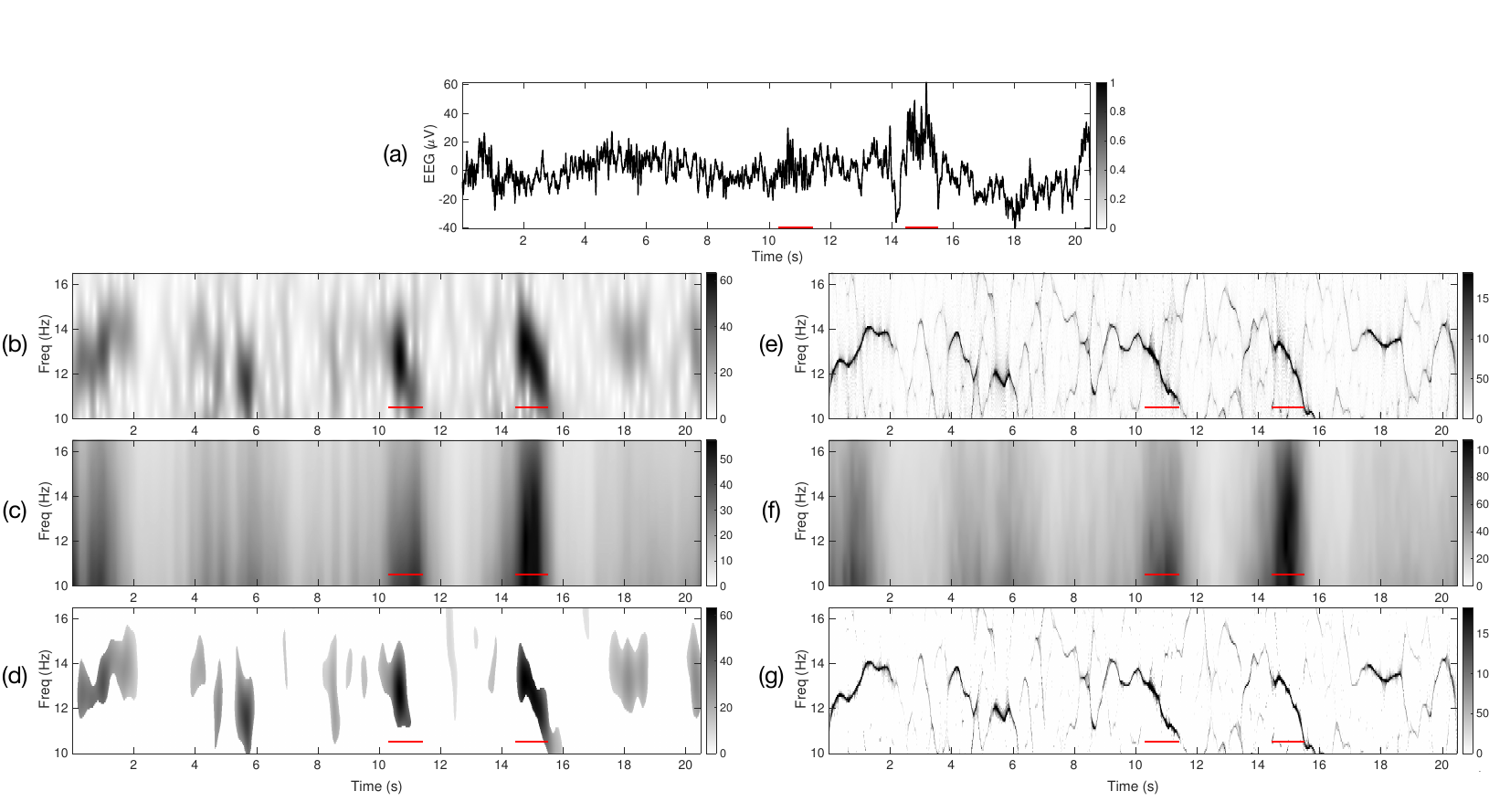}
\caption{\small An illustration of spindle analysis using the proposed bootstrapping algorithm. (a) the EEG signal recorded during the N2 sleep stage with the experts' spindle label marked in red. (b)-(d): the spectrogram, the threshold determined by the 95\% percentile of noise distribution determined by the bootstrap algorithm, and the thresholded TFR. (e)-(g): the synchrosqueezing transform, the threshold determined by the 95\% percentile of noise distribution determined by the bootstrap algorithm, and the thresholded TFR. \label{Figure: Spindle figure}}  
\end{figure}

Sleep spindles are brief bursts of activity in the sigma frequency range (around 11-16 Hz) of EEG, lasting 0.5 to 2 seconds \cite{iber2007aasm}. They are characteristics of the N2 stage of non-rapid eye movement sleep, marking the transition between light and deep sleep. Understanding sleep spindles is key to decoding sleep architecture, memory consolidation, and cognitive functions, and can shed light on healthy sleep patterns and neurological disorders \cite{weiner2016spindle}. Notably, spindle frequencies vary and are characterized by IF. The absence of spindle deceleration is linked to disorders like sleep apnea \cite{carvalho2014loss} and autism in children \cite{tessier2015intelligence}. Inter-expert agreement on spindle identification is limited, but reliability can be improved using qualitative confidence scores \cite{wendt2015inter}.  Estimating spindle duration is particularly challenging due to the thin-tail structure of spindles. A nonlinear multitaper method, a variation of SST called concentration of frequency and time (ConceFT), has been shown efficient in handling this issue \cite{shimizu2024unveil}. We conjecture that our proposed bootstrap algorithm could help handle this challenge and provide quantitative confidence. 
We show how the proposed bootstrap algorithm works on a segment from the open-access DREAMS database\footnote{\url{https://zenodo.org/record/2650142}}. Figure \ref{Figure: Spindle figure} shows the spectrogram, SST, and threshold determined by the 95th percentile of the noise distribution via the bootstrap algorithm. The red bars indicate spindles identified by experts. In the spectrogram, spindle dynamics, especially the IF, are difficult to discern. However, the dynamics, particularly the IF, can be visualized in the SST as curves. Both labeled spindles show decreasing IFs. The thresholding further clarifies the TFR, reducing background noise and revealing clear IF traces. Although not labeled, there is likely a spindle in the first 2 seconds with increasing, nonlinear IF dynamics. Given the raw EEG signal's complexity, it is understandable why this spindle was missed in expert's label--it overlaps with a bump and does not resemble a typical spindle. We conjecture that the UQ could assist expert annotation and explore spindle dynamics, and this clinical topic will be explored in future work.

\section{Theoretical support}

\subsection{Gaussian approximation of discrete STFT of general nonstationary noise}\label{section GA of STFT on NSN}

The distribution of discrete STFT of a locally stationary random process has been studied \cite{yang2022spectral} when we consider a fixed number of time-frequency pairs. However, when the input is a more general nonstationary time series, like the nonstationary random process satisfying the NSN model, or when the number of time-frequency pairs is growing as $n\to \infty$, it is still an open problem. 
In this section, we extend the results of \cite{yang2022spectral}  to the NSN model with diverging number of time-frequency points. 
The main technical tool we rely on is the recently developed high-dimensional Gaussian approximation \cite{LiuYangZhou2024} for HDNS time series that generalizes earlier results from \cite{Mies2024}.

First, we extend the sequential Gaussian approximation to the HDNS model improving the results in \cite[Theorem 2.2]{Mies2024}. 
The proof is postponed to Section \ref{subsection: proof of sequential ga our version}.

\begin{thm}\label{thm:OurSequentialGaussianApproximationTheorem}
Suppose the HDNS time series $\{z_i\}_{i=1}^n$ satisfy Assumption \ref{model assumption 4}  with $A > \sqrt{\chi} + 1$. Further, assume that $\|z_i\|_p\leq B_d$ for some $B_d>0$. On a sufficiently rich probability space, there exist $\{\hat z_i\}_{i=1}^n$ so that $\hat z_i \stackrel{\mathcal{D}}{=} z_i$, and Gaussian random vectors $y_i$ such that $y_i\sim \mathcal{N}(0,\texttt{cov}(z_i))$, such that 
\begin{equation}
 \left\|\max_{k=1,\ldots,n}\left|\frac{1}{\sqrt{n}}\sum_{j=1}^k(\hat{z}_j-y_j)\right|\right\|_2 =O\left(\left(\frac{d^{\frac{3}{4}-\frac{1}{2s}-\frac{1}{p}}}{n^{\frac{1}{4}-\frac{1}{2p}-\frac{1}{2s}+\frac{1}{ps}}}\right)^{\frac{1}{1-\frac{1}{s}-\frac{1}{p}}} \log(n)^2\right)\,,
\end{equation}
where the implied constant depends on $p$ and the dependence and moment of $z_i$.

\end{thm}

This result generalizes \cite[Theorem 3.1]{Mies2024}, offering an improved convergence rate. When $z_i$ has i.i.d. entries, $d$ is large, and $p,s\to\infty$, the bound simplifies to $d^{3/4}n^{-1/4}$, up to a logarithmic factor. This suggests that the bound asymptotically converges to $0$ when $d\asymp n^{1/3-\gamma}$ for any small $\gamma$. On the other hand, for time series with sufficiently short memory and light tail, the approximation bound established in \cite{Mies2024} asymptotically vanishes for $d$ as large as $O(n^{1/4-\gamma})$ for any small $\gamma$ and their bound converges at the rate $n^{-1/6}$ (within a logarithm factor) as the sample size increases. 
While it is possible that this bound could be further improved using alternative techniques, we find it sufficient for our current application and leave the investigation of sharper rates to future work.

To obtain a Gaussian approximation for the discrete STFT, we require the following lemma. When a nonstationary time series satisfying the NSN model is analyzed via the discrete STFT, it is transformed into a HDNS time series with $d>1$. The lemma shows that this transformed series preserves the moment and dependence conditions of the original process. The proof is deferred to Section~\ref{subsection: property of the rp for stft}.

\begin{lemma} \label{lemme: for theorem of STFT distribution like Gaussian}
Suppose $\{\epsilon_i\}_{i=1}^n$ satisfies the NSN model and Assumption  \ref{model assumption 4}, with the condition $A > \sqrt{\chi} + 1$ fulfilled.
Suppose $d=d(n)$ that diverges when $n\to\infty$ and
consider $0<\eta_1<\ldots<\eta_d$. Define a high dimensional time series
$\mathbf X_j:=[\epsilon_j \cos(2\pi \eta_1 j),\ldots,\epsilon_j \cos(2\pi \eta_d j),\,\epsilon_j \sin(2\pi \eta_1 j),\ldots,\epsilon_j \sin(2\pi \eta_d j)]^\top\in \mathbb{R}^{2d}$, 
where $j=1,\ldots,n$. Then, the time series $\mathbf X_j$ is HDNS satisfying Assumption \ref{model assumption 4}, with the condition $A > \sqrt{\chi} + 1$ fulfilled. 
\end{lemma}

With the above preparation, we are ready to state our main theorem about the approximate Gaussianity of discrete STFT if the input is a non-Gaussian and non-stationary random process. The proof is postponed to Section \ref{subsection proof of main theorem of stft}.

\begin{thm} \label{Theorem: STFT distribution like Gaussian} 
Suppose $\{\epsilon_i\}_{i=1}^n$ satisfies the NSN model, Assumption  \ref{model assumption 4}, and $A > \sqrt{\chi} + 1$.
Suppose $d=d(n)\to \infty$ when $n\to\infty$ and
consider $0<\eta_1<\ldots<\eta_d$. Denote a complex random vector associated with the discrete STFT  as
$\mathbf V_l:=[\mathbf V^{(\mathbf{h})}_\epsilon(l,\eta_1),\ldots,\mathbf V^{(\mathbf{h})}_\epsilon(l,\eta_d),\mathbf V^{(\mathbf{Dh})}_\epsilon(l,\eta_1),\ldots,\mathbf V^{(\mathbf{Dh})}_\epsilon(l,\eta_d)]^\top\in \mathbb{C}^{2d}$,
where $\mathbf{h}$ and $\mathbf{Dh}$ are defined in Section \ref{section background SST STFT} 
with $m:=\lceil \beta q\rceil$.
Suppose $\{\hat{\epsilon}_i\}_{i=1}^n$ is a Gaussian process defined on a potentially different probability space that shares the same covariance structure of $\{\epsilon_i\}_{i=1}^n$ and denote the associated discrete STFT as $\hat{\mathbf V}_l$. 
We have
\begin{align*}
\mathbb{E}\left(\max_l \Big|\frac{1}{\sqrt{n}}(\mathbf V_l-\hat{\mathbf V}_l)\Big|^2\right)
\leq {C'}\left(\frac{d^{\frac{3}{4}-\frac{1}{2s}-\frac{1}{p}}}{n^{\frac{1}{4}-\frac{1}{2p}-\frac{1}{2s}+\frac{1}{ps}}}\right)^{\frac{2}{1-\frac{1}{s}-\frac{1}{p}}} \log(n)^4 m^{-2}\,,\nonumber
\end{align*}
where $C'$ depends on $\mathsf{h}$ and the moment and dependence structures of $\epsilon_i$.
\end{thm}

When $p,s\to\infty$, the bound becomes $d^{3/2}n^{-1/2}m^{-2}$ up to a logarithmic factor. Thus, depending on the window size $m$, the number of frequencies we can control is up to the order of $n^{1/3}m^{4/3}$ up to a logarithmic factor. Also, by construction, the variance of each entry of $\hat{\mathbf{V}}_l$ is of order $1$, and the quantity $\frac{d^{\frac{3}{4}-\frac{1}{2s}-\frac{1}{p}}}{n^{\frac{1}{4}-\frac{1}{2p}-\frac{1}{2s}+\frac{1}{ps}}}$ decays polynomially to zero. We conclude that when $m\asymp n^{1/2}$, $\mathbb{E}\left(\max_l \Big|\mathbf V_l-\hat{\mathbf V}_l\Big|\right)\leq {C'}\left(\frac{d^{\frac{3}{4}-\frac{1}{2s}-\frac{1}{p}}}{n^{\frac{1}{4}-\frac{1}{2p}-\frac{1}{2s}+\frac{1}{ps}}}\right)^{\frac{1}{1-\frac{1}{s}-\frac{1}{p}}} \log(n)^2$, which means the discrete STFT of a NSN time series can be well approximated by its Gaussian companion under mild conditions.

\subsection{Robustness of reconstruction formula}\label{section:reconstruction formula proof}

We impose the following assumption on the window function.

\begin{assumption}\label{assumption window function main theorem}
Let $\mathsf{h}_0$ be a nonnegative and symmetric Schwartz function compactly supported on $[-1,1]$ and normalized to have unit $L^2$ norm.
Assume the Fourier transform of $\mathsf{h}_0$, denoted as $\widehat{\mathsf{h}}_0$, satisfies $|\widehat{\mathsf{h}}_0(\eta)|>\delta_1$ on $[-\Delta_0,\Delta_0]$, where $\delta_1>0$ and $\Delta_0>1$, and $\int_{\Delta_0}^\infty |\widehat{\mathsf{h}}_0(\eta)| d\eta\leq \delta_2$, where $\delta_2>0$ is a small constant. 
Take $\beta>0$ so that $\Delta:=\Delta_0/\beta<\Xi_s/2$ and define $\mathsf{h}(t):=\mathsf{h}_0(t/\beta)/\sqrt{\beta}$.  
\end{assumption}

These seeming complicated assumptions have clear interpretation and can be easily fulfilled in practice. While the conditions in this assumption can be relaxed (e.g. allowing non-compact or lower regularity), doing so would would only complicate notation and proofs without additional insight. To simplify the proof, we retain these assumptions. 
The parameters $\delta_1$ and $\delta_2$  jointly describe the shape of the window function $\mathsf{h}$. By the Schwartz condition, we can choose $\Delta_0$ to be of order one so that both $\delta_1$ and $\delta_2$ are small. As shown below, these parameters quantify the reconstruction accuracy. Recall that for any $\delta\in (0,1)$ and $C>0$, there exists a real, nonnegative, symmetric Schwartz function $h$ supported on $[-1,1]$ such that $C_1e^{-(1+\epsilon)C\xi^\delta}\leq \hat{h}(\xi)\leq C_2e^{-(1-\epsilon)C\xi^\delta}$ for any small $\epsilon>0$ and some $C_1,C_2>0$ \cite{tlas2022bump}. For such a kernel, we may choose $\delta_1= C_1e^{-(1+\epsilon)C\Delta_0^\delta}$ and $\delta_2=\frac{C_2}{(1-\epsilon)C\delta}e^{-(1-\epsilon)C\Delta_0^\delta}$.
The parameter $\beta$ controls the spectral concentration of $\mathsf{h}$, mitigating spectral interference, which is important since the signal under consideration is real. This reflexes the common practice of using wider windows to analyze lower-frequency components. Note that $\beta$ scales with $\Delta_0$: achieving smaller $\delta_1$ and $\delta_2$ requires a larger $\Delta_0$, and hence typically a larger $\beta$. 
In practice, $\Delta_r$ in \eqref{alogithm:sst:reconstruction continuous formula} is much smaller than $\Delta$, since the TFR concentration is sharpened by the SST.

The robustness result of the reconstruction formula is stated in following theorem, and the proof is postponed to Section \ref{subsection proof of sst reconstruction robustness}.

\begin{thm}\label{section:theorem:stability}
Suppose $Y_i$ follows \eqref{eq:model} and \eqref{discretization of AHM function sqrt n Hz}, $f$ satisfies the AHM and Assumption \eqref{model assumption 1}, and $\epsilon_i$  satisfies the NSN model, Assumption  \ref{model assumption 4}, and $A > \sqrt{\chi} + 1$. Set $m=\lceil \beta \sqrt{n}\rceil$, $d= n^{1/3-\gamma}$, and $\sigma=\sigma(n)=n^{1/4-\gamma'}$, where $\gamma,\gamma'>0$ are small constants. Set the sampled frequencies as $\eta_k=\frac{k\Xi}{d}$, where $k=1,\ldots,d$. Suppose the window function $\mathsf{h}$ satisfies Assumption \ref{assumption window function main theorem} and $\Delta$ satisfies $\frac{2(E_f'+\Xi E_f)}{\Xi_a\sqrt{\beta}\delta_1}\leq 1/2$, where $E_f$ and $E'_f$ are constants defined in Lemma \ref{Lemma: SST on deterministic function AHM}. For SST, set $\alpha:=(\Delta_r/C_\alpha)^2$ for some $C_\alpha>1$, and $\nu :=\Xi_a\sqrt{\beta}\delta_1/2+\zeta_n$,
where $\zeta_n\asymp n^{-\gamma'}\sqrt{\log n}$. Then, when $\varepsilon\geq 0$ is sufficiently small and $n$ is sufficiently large, with probability greater than $1-n^{-2}$, the SST reconstruction formula with the threshold $\nu/2$ satisfies
\begin{align}
\max_{l=m+1\ldots,n-m}\left|\widetilde{f}_k^{\mathbb{C}}(t_l)-A_k(t_l)e^{\ii2\pi\phi_k(t_l)}\right|\leq  C_1\Xi_A+C_2\zeta_n\,,
\end{align}
where $C_1>0$ is a small constant comprises of terms linearly depending on $\delta_1,\delta_2,\varepsilon$ and $\texttt{erfc}(C_\alpha)$ respectively, $\texttt{erfc}(\cdot)$ is the complementary error function, and $C_2>0$, which might not be small. Details of $C_1$ and $C_2$ can be found in \eqref{proof SST recon noisy setup final bound er definition}.
\end{thm}

The maximum is taken over indices $m+1$ to $n-m$ to mitigate boundary effects, since only partial window support is available near the boundaries and reconstruction quality deteriorates there. Treating boundary effects is a separate problem (e.g., \cite{meynard2021efficient}); in many settings they are asymptotically negligible. We therefore omit them to keep the focus of the paper.

The noise $\sigma \epsilon_i$ with $\sigma=\sigma(n)=n^{1/4-\gamma'}$ may seem counterintuitive since its magnitude grows with $n$. As shown in the proof, the standard deviation of $V^{(\mathbf{h})}_{\epsilon}(t_l,\xi_k)$ is of order $n^{-1/4}$ up to a logarithmic factor. Thus, scaling by $\sigma$ normalizes the TFR so that its fluctuations are of order one as a random field.
This normalization is consistent with the continuous framework in \cite{chen2014non}, where the noise is modeled as a generalized random process $\Phi$ such that $\Phi(\mathsf{h})$ has standard deviation of order one. In this sense, our result extends the pointwise robustness analysis in \cite{chen2014non} to a uniform result in the discrete setting. When $\gamma'=1/4$, we recover the intuitive case in which the noise magnitude remains bounded.

The reconstruction error comprises two main contributions. $C_1\Xi_A$ comes from the AHM and window truncation in the reconstruction formula, through the definition of $R_l$. The second term, $C_2\zeta_n$, reflects the effect of  noise. Since $E_f$ and $E_f'$ in the constraint $\frac{2(E_f'+\Xi E_f)}{\Xi_a\sqrt{\beta}\delta_1}\leq 1/2$ depend linearly on $\Xi_A$, this constraint effectively imposes a lower bound on $\sqrt{\beta}\delta_1$, indicating that $\delta_1$ cannot be chosen arbitrarily small. Consequently, even in the purely harmonic case ($\varepsilon=0$), the term $C_1$ does not vanish. If we choose $\delta_2=C_{\mathsf{h}}\Delta \mathsf{h}(0)\beta\delta_1$, the error terms in $C_1$ can be simplified. In the absence of noise, that is, when $\epsilon_i=0$, we have $C_2\zeta_n=0$, and the theorem reduces to the noise-free setting. In this case, the error scales linearly with $\Xi_A$.

\subsection{Guarantees of the bootstrap algorithm}\label{section bootstrap theorem proof}
 we present a bootstrapping theorem under the locally stationary assumption. A process $Z_j$, $j\in\mathbb{Z}$ satisfies the {\em uniformly positive definite in covariance} (UPDC) condition if, for sufficiently large $n$, the smallest eigenvalue of the covariance matrix of $(Z_1,\ldots,Z_n)$ is bounded below by a constant $\kappa > 0$. Furthermore, if $\max_k |\mathrm{cov}(Z_k, Z_{k+r})| \leq r^{-\tau}$ for all $r\in\mathbb{N}$ and some $\tau \geq 0$, we call $\tau$ the {\em covariance decay speed}. Further background on covariance locally stationary processes, UPDC, and short-memory conditions can be found in \cite{ding2023autoregressive}. The proof is deferred to Section~\ref{subsection proof of bootstrapping for sst}. 

\begin{thm}\label{Bootstrapping main theorem}
Assume $\{\epsilon_i\}_{i=1}^n$ is locally stationary, the associated covariance function is smooth, the local spectral density is lower bounded and fulfills the UPDC condition, and its covariance decay speed is $\tau>2$. Grant conditions in Theorem \ref{Theorem: STFT distribution like Gaussian}.
Then there exists a probability space $(\Omega, \mathsf F, \mathbb P)$, where we could construct a Gaussian tvAR random process following Section \ref{section bootstrap algorithm}, $\{\epsilon^{(*)}_{i}\}_{i=1}^n$ from $\{\epsilon_i\}_{i=1}^n$.  
Take $d$ frequencies, $\mathcal{G}:=(\xi_1,\ldots,\xi_d)$, where $d=n^a$ and $a\geq 0$. Then, when 
$a=\min\left\{\gamma,\, \frac{\frac{1}{4}-\frac{1}{2p}-\frac{1}{2s}+\frac{1}{ps}}{\frac{5}{4}-\frac{1}{s}-\frac{3}{2p}}\right\}-\vartheta$, 
where $\vartheta>0$ is a small constant so that $a\geq0$, $\gamma=\frac{\tau-2}{\tau+1}\in(0, 1]$ and $b$ is chosen to fulfill $\frac{b}{\log(b)}\asymp n^{1/(\tau+1)}$, we have 
\begin{align*}
&\max_{l=1,\ldots,d}\max_{i=1,\ldots,n}\Big|S^{(\mathbf{h})}_\epsilon(t_i,\xi_l)-S^{(\mathbf{h})}_{\epsilon^{(*)}}(t_i,\xi_l)\Big|=o_p(1)\,.
\end{align*}
\end{thm}

When $\epsilon_i$ satisfies $p,s\to\infty$ and $\gamma$ is close to $1$, the exponent $a$ approaches $1/5$. This rate is slightly worse than the Gaussian approximation for STFT, as we used crude bounds to control the approximation error to simplify the proof. According to \cite{sourisseau2022asymptotic}, if we further assume that $\epsilon_i$ is stationary, $S^{(\mathbf{h})}_{\epsilon^{(*)}}(t_i,\xi_l)$ is complex Gaussian with nontrivial variation. Therefore, $\max\limits_{l=1,\ldots,d}\max\limits_{i=1,\ldots,n}\Big|S^{(\mathbf{h})}_\epsilon(t_i,\xi_l)\Big|$ is away from zero, and the resulting error control is meaningful. 

\bibliographystyle{plain}
\bibliography{reference.bib}

\clearpage

\setcounter{page}{1}
	\setcounter{equation}{0}
	\renewcommand{\thepage}{SI.\arabic{page}}
	\renewcommand{\thesection}{SI.\arabic{section}}
	\renewcommand{\theequation}{SI.\arabic{equation}}
	\renewcommand{\thelemma}{SI.\arabic{lemma}}
	\renewcommand{\theprop}{SI.\arabic{proposition}}
	\renewcommand{\thetable}{SI.\arabic{table}}
	\renewcommand{\thefigure}{SI.\arabic{figure}}

\section{Proofs}\label{Section proof supp}

\subsection{Proof of Theorem \ref{thm:OurSequentialGaussianApproximationTheorem}: Sequantial Gaussian approximation under HDNS}\label{subsection: proof of sequential ga our version}

The main technical tool we need is a recently developed high-dimensional Gaussian approximation \cite{LiuYangZhou2024} for HDNS time series that generalizes results in \cite{Mies2024}, which is based on the martingale embedding technique \cite{eldan2020clt}. The following theorem is a generalization of \cite[Theorem 2.1]{Mies2024}.

\begin{thm}\cite[Theorem 3.2]{LiuYangZhou2024} \label{LiuYangZhou2024RecentResult}
Suppose the $d$-dimensional HDNS time series $\{z_i\}_{i=1}^n$ satisfy Assumption  \ref{model assumption 4}  with $A > \sqrt{\chi} + 1$. 
Consider $Z_n:=\sum_{i=1}^nz_i\in \mathbb{R}^d$, and assume that the smallest eigenvalue of $\texttt{cov}(\frac{1}{\sqrt{n}}Z_n)$ is bounded below by some constant $\lambda_*>0$.
Then, on a sufficiently rich probability space,  one can construct $\hat Z_n$ such that $\hat Z_n \stackrel{\mathcal{D}}{=}Z_n$ and a Gaussian random vector $Y_n$ with the same mean and covariance matrix as $Z_n$. We have
\begin{equation}
\left\|\frac{1}{\sqrt{n}}(\hat Z_n- Y_n)\right\|_2 =O(dn^{1/s-1/2} \log(n)), 
\end{equation}
where $1/s:=\max\left\{\frac{1}{p},\,\frac{1}{\sqrt{\chi+1}p}+
\left(\frac{1}{2}-\frac{1}{\sqrt{\chi+1}p}\right)\max \left\{\frac{2}{\sqrt{\chi+1}p},\,
\frac{1}{\chi}\left(\frac{1}{2}-\frac{1}{p}\right)\right\}\right\}$.
\end{thm}

Note that the authors used the 2-Wasserstein distance in the statement of \cite[Theorem 3.2]{LiuYangZhou2024}, while in the proof it is the $L^2$ distance that is proved. 
When $\chi$ is sufficiently large, the convergence rate becomes $dn^{1/p-1/2}$, which is further reduced to $dn^{-1/2}$ when $p$ is large. This allows us the nearly optimal convergence rate $d\asymp n^{1/2}$ \cite{portnoy1986central}.

We also need a generalized Rosenthal inequality for the HDNS model, as developed in \cite[Section 4]{liu2013probability}, which generalizes \cite[Theorem 1]{liu2013probability} and is of independent interest. This result can also be viewed as a different version of \cite[Theorem 3.2]{Mies2024}. 

\begin{lemma}[Rosenthal inequality]\label{lemma: rosenthal ineq we need}
Suppose the $d$-dimensional HDNS time series $\{z_i\}_{i=1}^n$ satisfy Assumption  \ref{model assumption 4}. Denote $S_i:=\sum_{l=1}^iz_l$. Then, for $2\leq p<\infty$, we have
\begin{align}
\left\|\max_{1\leq i\leq n}|S_i|\right\|_p\leq C\sqrt{dn}\,,
\end{align}
where $C=C_p\left[\Theta_p(1)+B^{1/p}\right]$ for some constant $C_p$ that depends only on $p$.
\end{lemma}

We establish some technical lemmas.
\begin{lemma}\label{lem:sub}
Under the same assumptions as those in Theorem \ref{LiuYangZhou2024RecentResult}, denote  $Z_{i,m}:=\sum_{j=(i-1)m+1}^{im}z_{i}$, $i=1,2,\cdots$, where $m=m(n)\rightarrow\infty$ with $m/n\rightarrow 0$. Then there exists  a sequence of zero mean Gaussian random vectors $\{y_i\}_{i=1}^n$ that preserves the covariance structure of $\{z_i\}_{i=1}^n$, such that 
\begin{eqnarray}\label{eq:sub1}
\max_{a,b, a\le b}\left\|\sum_{j=a}^b(Z_{j,m}-Y_{j,m})\right\|_2=O(d\sqrt{b-a+1}m^{1/s}\log m),
\end{eqnarray}
where $Y_{i,m}=\sum_{j=(i-1)m+1}^{im}y_{i}$, $i=1, 2,\cdots$ and $s$ is defined in Theorem \ref{LiuYangZhou2024RecentResult}.
\end{lemma}
\begin{proof}
The proof of this lemma follows the same martingale embedding and boundary conditioning arguments as those used in establishing \cite[Theorem 3.1]{LiuYangZhou2024}. In particular, Wasserstein Gaussian approximations are jointly applied to each block sum $Z_{i,m}$, $i=1,2,\cdots$ to obtain
$$\max_{i}\|Z_{i,m}-Y_{i,m}\|_2=O(dm^{1/s}\log m).$$
Then the lemma follows from the construction where the covariance between block sums $Z_{i,m}-Y_{i,m}$ for adjacent blocks are negligible. See \cite[Theorem 3.1]{LiuYangZhou2024} for the details.
\end{proof}

\begin{theorem}\label{prop:submax}
Under the same setup as that in Lemma \ref{lem:sub}, we have that
\begin{eqnarray*}
\left\|\max_{i}\Big|\sum_{j=1}^i(Z_{j,m}-Y_{j,m})\Big|\right\|_2=O(d\sqrt{n/m}m^{1/s}\log_2(2n/m)\log m).
\end{eqnarray*}
\end{theorem}
\begin{proof}
The proof of this Theorem is mainly based on the ``bisection'' technique in \cite[page 102]{billingsley1968convergence} with an induction argument. See also \cite{moricz1976moment}. In particular, Theorem \ref{prop:submax} follows from similar proofs as those of \cite[Theorems 3 and 4]{moricz1976moment}. For the sake of completeness, we present the proof below.

Without loss of generality, assume that $n$ is a multiple of $m$. Let $D_i=Z_{i,m}-Y_{i,m}$, $i=1, 2\cdots,n/m$. Define $S_{l,k}=\sum_{j=l+1}^{l+k}D_j$ ($S_{l,0}=0$) and $M_{l,k}=\max_{1\le j\le k}|S_{l,j}|$, $l\ge 0$. 
By Lemma \ref{lem:sub}, we have that
\begin{eqnarray}\label{eq:sub2}
\max_{l}\|S_{l,k}\|_2\le Cd\sqrt{k}m^{1/s}\log m
\end{eqnarray}
for some finite constant $C>0$, where $Y_{i,m}=\sum_{j=im-m+1}^{im}y_{i}$, $i=1, 2,\cdots$ and $s$ is defined in Theorem \ref{LiuYangZhou2024RecentResult}. 
We will use induction to show that 
\begin{eqnarray}\label{eq:sub3}
\|M_{l,k}\|_2\le C\sqrt{k}\log_2 (2k)dm^{1/s}\log m \mbox{ for all } l\ge 0, k\ge 1.
\end{eqnarray}
The claim obviously holds for $k=1$. Let $m^*$ be the integer part of $n^*/2+1$ where $n^*=n/m$. Observe that if $m^*\le k\le n^*$, then
$$|S_{l,k}|\le |S_{l,m^*}|+|S_{l+m^*,k-m^*}|.$$ 
As a result, we have $|S_{l,k}|\le |S_{l,m^*}|+M_{l+m^*,n^*-m^*}$. On the other hand, $|S_{l,k}|\le M_{l,m^*-1}$ if $k\le m^*-1$.  Therefore, for $k\in [1,n^*]$, we have
$$|S_{l,k}|\le |S_{l,m^*}|+[M_{l,m^*-1}^2+M^2_{l+m^*,n^*-m^*}]^{1/2}.$$
Hence
$$|M_{l,n^*}|\le |S_{l,m^*}|+[M_{l,m^*-1}^2+M^2_{l+m^*,n^*-m^*}]^{1/2}.$$
By Minkowski's Inequality, we obtain that
\begin{eqnarray}\label{eq:sub4}
\|M_{l,n^*}\|_2\le \|S_{l,m^*}\|_2+[\|M_{l,m^*-1}\|_2^2+\|M_{l+m^*,n^*-m^*}\|_2^2]^{1/2}.
\end{eqnarray}
Suppose that the conclusion in \eqref{eq:sub3} holds for all integers $k<n^*$. We have by the induction hypothesis that
$\|M_{l,m^*-1}\|_2\le C\sqrt{m^*-1}\log_2 (2(m^*-1))dm^{1/s}\log m$ and $\|M_{l+m^*,n^*-m^*}\|_2\le C\sqrt{n^*-m^*}\log_2 (2(n^*-m^*))dm^{1/s}\log m$. Simple calculations yield that
$$[\|M_{l,m^*-1}\|_2^2+\|M_{l+m^*,n^*-m^*}\|_2^2]^{1/2}\le C\sqrt{n^*}\log_2 (2(m^*-1))dm^{1/s}\log m.$$
By \eqref{eq:sub3}, we have 
$$\|S_{l,m^*}\|_2\le Cd\sqrt{m^*}m^{1/s}\log m\le Cd\sqrt{n^*}m^{1/s}\log m.$$ 
Hence by \eqref{eq:sub4}, we have
\begin{eqnarray*}
\|M_{l,n^*}\|_2\le C\{\log_2 (2(m^*-1))+1\}\sqrt{n^*}dm^{1/s}\log m\le C\log_2 (2n^*)\sqrt{n^*}dm^{1/s}\log m.
\end{eqnarray*}
The theorem follows by setting $l=0$.
\end{proof}

\begin{proof} [Proof of Theorem \ref{thm:OurSequentialGaussianApproximationTheorem}]
Divide the sequence $\{z_i\}_{i=1}^n$ into blocks of size $L\leq n$, where $L$ will be determined in the end, and denote $M:=\lfloor n/L\rfloor$. The goal is applying Theorem \ref{LiuYangZhou2024RecentResult} to gain an approximation over each block, and control the Gaussian approximation using the weakly dependent structure of the blocks to obtain the final Gaussian approximation. 

By Theorem \ref{LiuYangZhou2024RecentResult} and Theorem \ref{prop:submax}, on a potentially richer probability space, we can find sequences of Gaussian random vectors $y_i\sim N(0,\texttt{cov}(z_i))$, and random vectors $\hat{z}_i \stackrel{\mathcal{D}}{=} z_i$,  $i=1,\dots,n$, such that 
\begin{align}
\mathbb{E}\left[ \max_{r=1,\ldots,M}\left| \frac{1}{\sqrt{n}}\sum_{k=1}^r Q_k \right|^2\right] &\, =\frac{1}{n}\left\| \max_{r=1,\ldots,M}\left| \sum_{k=1}^r Q_k \right| \right\|_2^2\nonumber\\
&\, =O(d^2 L^{2/s-1}\log_2^2(2n/L)\log^2(L))\,.\label{proof: seg ga our version Doob's ineq}
\end{align}
where $Q_k:=\sum_{j=(k-1)L+1}^{kL\wedge n}(\hat z_j-y_j)$ and $s$ is defined in Theorem \ref{LiuYangZhou2024RecentResult}. 

By the Rosenthal inequality in Lemma \ref{lemma: rosenthal ineq we need} and the fact that 
\[
\mathbb{E}[|y_i|^p]\leq C'(\mathbb{E}[|y_i|^2])^{p/2}=C'(\mathbb{E}[|\hat z_i|^2])^{p/2}\leq C'\mathbb{E}[|\hat z_i|^p]
\]
for some constant $C'$ depending only on $p$ \cite[Lemma 6.1]{Mies2024}, we have
\begin{align}
&\,\left(\mathbb{E}\left[ \max_{r=(j-1)L+1,\ldots,jL\wedge n}\left| \sum_{k=(j-1)L+1}^r \hat z_j \right|^p\right]\right)^{1/p}\leq C\sqrt{dL}\,,\label{proof: seg ga our version rosenthal ineq applied}\\
&\,\left(\mathbb{E}\left[ \max_{r=(j-1)L+1,\ldots,jL\wedge n}\left| \sum_{k=(j-1)L+1}^r y_j \right|^p\right]\right)^{1/p}\leq CC'\sqrt{dL}\,,
\end{align}
where $C=C_p\left[\Theta_p(1)+B^{1/p}\right]$ for some constant $C_p$ that depends only on $p$.

Next, note that if $\max_{k=1,\ldots, n}\left|\sum_{t=1}^k a_t \right|^2$ for a realization of random vectors $a_t:=\hat z_t-y_t$, $t=1,\ldots,n$ is achieved at $k^*$, where $(j^*-1)L+1\leq k^*<j^*L$, then 
\[
\left|\sum_{t=1}^{k^*} a_t \right|^2=\left|\sum_{t=1}^{(j^*-1)L} a_t+\sum_{(j^*-1)L+1}^{k^*} a_t \right|^2\leq 2\left|\sum_{t=1}^{(j^*-1)L} a_t\right|^2+2\left|\sum_{(j^*-1)L+1}^{k^*} a_t \right|^2\,.
\] 
Thus, 
\begin{align*}
\max_{k=1,\ldots,n}\left|\sum_{t=1}^{k} a_t \right|^2\leq 2\max_{r=1,\ldots,M}\left|\sum_{t=1}^{(r-1)L} a_t\right|^2+2\max_{1\leq j\leq M}\max_{(j-1)L+1\leq k\leq jL}\left|\sum_{(j-1)L+1}^{k} a_t \right|^2
\end{align*}
and hence
\begin{align*}
&\left(\mathbb{E}\left[\max_{k=1,\ldots,n}\left|\sum_{t=1}^{k} a_t \right|^2\right]\right)^{1/2}\\
\leq &\, \sqrt{2}\left(\mathbb{E}\left[\max_{1\leq r\leq M}\left|\sum_{t=1}^{(r-1)L} a_t\right|^2\right]\right)^{1/2}+\sqrt{2}\left(\mathbb{E}\left[\max_{1\leq r\leq M}\max_{(r-1)L+1\leq k\leq rL}\left|\sum_{(r-1)L+1}^{k} a_t \right|^2\right]\right)^{1/2}.
\end{align*}
Note that
\begin{align*}
&\left(\mathbb{E}\left[\max_{1\leq r\leq M}\max_{(r-1)L+1\leq k\leq rL}\left|\sum_{(r-1)L+1}^{k} a_t \right|^2\right]\right)^{1/2}\\
=&\left(\mathbb{E}\left[\left(\max_{1\leq r\leq M}\max_{(r-1)L+1\leq k\leq rL}\left|\sum_{(r-1)L+1}^{k} a_t \right|\right)^2\right]\right)^{1/2}\\
=&\left\|\max_{1\leq r\leq M}\max_{(r-1)L+1\leq k\leq rL}\left|\sum_{(r-1)L+1}^{k} a_t \right|\right\|_2\\
\le &\left\|\max_{1\leq r\leq M}\max_{(r-1)L+1\leq k\leq rL}\left|\sum_{(r-1)L+1}^{k} a_t \right|\right\|_p\\
=&\left(\mathbb{E}\left[\max_{1\leq r\leq M}\max_{(r-1)L+1\leq k\leq rL}\left|\sum_{(r-1)L+1}^{k} a_t \right|^p\right]\right)^{1/p}.
\end{align*}
Therefore,
\begin{align}
&\left(\mathbb{E}\left[ \max_{k=1,\ldots, n}\left| \frac{1}{\sqrt{n}}\sum_{t=1}^k (\hat z_t-y_t) \right|^2\right]\right)^{1/2}\label{proof main theorem 1, first control}\\
\leq\,&\left(2\mathbb{E}\left[ \max_{r=1,\ldots,M}\left| \frac{1}{\sqrt{n}}\sum_{k=1}^r Q_k \right|^2\right]\right)^{1/2}\nonumber\\
&+\sqrt{2}\left(\mathbb{E}\left[\max_{1\leq j\leq M}\max_{\substack{(j-1)L+1\leq \\k\leq jL}}\left|\frac{1}{\sqrt{n}}\sum_{t=(j-1)L+1}^{k\wedge n}(\hat z_t-y_t)\right|^p\right]\right)^{1/p}\nonumber\\
\leq\,&\sqrt{2}\left(\mathbb{E}\left[ \max_{r=1,\ldots,M}\left| \frac{1}{\sqrt{n}}\sum_{k=1}^r Q_k \right|^2\right]\right)^{1/2}\nonumber\\
&+\sqrt{2}M^{1/p}\max_{1\leq j\leq M}\left(\mathbb{E}\left[\max_{\substack{(j-1)L+1\leq k\\ \leq jL\wedge n}}\left|\frac{1}{\sqrt{n}}\sum_{t=(j-1)L+1}^{k\wedge n}(\hat z_t-y_t)\right|^p\right]\right)^{1/p}\nonumber\,,
\end{align}
where the second inequality comes from the bound 
\[
(\mathbb{E}\max_{1\leq t\leq M}|\xi_t|^p)^{1/p}\leq M^{1/p}\max_{1\leq t\leq M}(\mathbb{E}[|\xi_t|^p])^{1/p}
\] 
for any random variables $\xi_1,\ldots,\xi_M$.
The
first term in the right hand side of \eqref{proof main theorem 1, first control} can be bounded by \eqref{proof: seg ga our version Doob's ineq} and the second term can be bounded by \eqref{proof: seg ga our version rosenthal ineq applied}. As a result, \eqref{proof main theorem 1, first control} is controlled by
\begin{align}
&\left(\mathbb{E}\left[ \max_{k=1,\ldots, n}\left| \frac{1}{\sqrt{n}}\sum_{t=1}^k (\hat z_t-y_t) \right|^2\right]\right)^{1/2}\label{proof main theorem 1, second control}\\
\leq\,&C_1 d L^{1/s-1/2}\log_2\left(\frac{2n}{L}\right)\log(L)+ \sqrt{2}C(C'+1)M^{1/p}\sqrt{\frac{Ld}{n}}\nonumber\\
\leq\,&C_1 d \left(\frac{n}{M}\right)^{1/s-1/2}\log_2(2M)\log\left(\frac{n}{M}\right)+ \sqrt{2}C(C'+1)M^{1/p-1/2}\sqrt{d}\nonumber
\end{align}
for some $C_1>0$.
By choosing $L$ so that the right hand side is balanced, we have $M=(\frac{n^{1/2-1/s}}{d^{1/2}})^{\frac{1}{1-1/s-1/p}}$, and hence
\begin{align}
&\left(\mathbb{E}\left[ \max_{k=1,\ldots, n}\left| \frac{1}{\sqrt{n}}\sum_{t=1}^k (\hat z_t-y_t) \right|^2\right]\right)^{1/2}\label{proof main theorem 1, final control}
=O\left(\left(\frac{d^{\frac{3}{4}-\frac{1}{2s}-\frac{1}{p}}}{n^{\frac{1}{4}-\frac{1}{2p}-\frac{1}{2s}+\frac{1}{ps}}}\right)^{\frac{1}{1-\frac{1}{s}-\frac{1}{p}}} \log(n)^2\right)\nonumber\,,
\end{align}
where the implied constant depends on $p$ and the dependence structure and moment control of $z_i$.
\end{proof}

\subsection{Proof of Lemma \ref{lemme: for theorem of STFT distribution like Gaussian}}\label{subsection: property of the rp for stft}
\begin{proof}
The $p$-norm bound of Assumptions \ref{model assumption 4} is immediate by the boundedness of sine and cosine functions. Next, note that $\mathbf X_j$ is generated by
$\mathbf X_j=\tilde{\mathcal{G}}_i(\mathcal{F}_i)$,  
where $\tilde{\mathcal{G}}_i:\mathbb{R}^\infty\to \mathbb{R}^{2d}$ via
$\tilde{\mathcal{G}}_i(\mathcal{F}_i):=\begin{bmatrix}\mathcal{G}(\mathcal{F}_i) \cos(2\pi \eta_1 j)\\ \vdots \\ \mathcal{G}(\mathcal{F}_i) \sin(2\pi \eta_d j) \end{bmatrix}$. Thus, the polynomial decay of the dependence measure of $\mathbf X_j$ in Assumption \ref{model assumption 4} holds based on the assumption of $X_i$. Finally, the moment bound holds since $\cos$ and $\sin$ functions are bounded by $1$. 
\end{proof}

\subsection{Proof of Theorem \ref{Theorem: STFT distribution like Gaussian}: the Gaussianity of STFT} \label{subsection proof of main theorem of stft}

\begin{proof}
Denote $\alpha_t \in\mathbb{C}^{2d\times 4d}$ so that 
$\alpha_t(k,k+2d)=e^{\ii2\pi \eta_k t}$ and $\alpha_t(k+d,k+3d)=\ii e^{\ii2\pi \eta_k t}$ for $k=1,\ldots,d$ and $0$ otherwise. 
Consider $\mathbf X_j\in \mathbb{R}^{2d}$ in Lemma \ref{lemme: for theorem of STFT distribution like Gaussian}. 
Denote the partial sum
\[
\mathbf S_j=\sum_{l=1}^j\mathbf X_l \in \mathbb{R}^{2d}
\]
with $\mathbf S_j=0$ when $j\leq 0$. 
We then have
\begin{align}
\mathbf V_t
&\,=\alpha_t \begin{bmatrix}
\sum_{j=t-m}^{t+m}\mathbf X_j   \mathbf{h}(j-t)\\ \sum_{j=t-m}^{t+m}\mathbf X_j   \mathbf{Dh}(j-t)\end{bmatrix}\nonumber\\
&\,=\alpha_t\begin{bmatrix}\sum_{j=t-m}^{t+m}(\mathbf S_j-\mathbf S_{j-1})   \mathbf{h}(j-t)\\ \sum_{j=t-m}^{t+m}(\mathbf S_j-\mathbf S_{j-1})   \mathbf{Dh}(j-t)\end{bmatrix}\nonumber\\
&\,=\alpha_t \begin{bmatrix}\sum_{j=t-m}^{t+m}\mathbf S_j (\mathbf{h}(j-t)-\mathbf{h}(j-1-t))-\mathbf S_{t-m-1}\mathbf h(-m) \\ \sum_{j=t-m}^{t+m}\mathbf S_j (\mathbf{Dh}(j-t)-\mathbf{Dh}(j-1-t))-\mathbf S_{t-m-1}\mathbf{D h}(-m)\end{bmatrix}\nonumber
\end{align}
Suppose $\{Z_i\}_{i\in \mathbb{Z}}$ is the Gaussian approximation of $X_i$ on a potentially richer probability space, $\hat{\mathbf S}_j$ is the associated partial sum, and $\hat{\mathbf V}_t\in \mathbb{C}^{2d}$ is the associated complex Gaussian random vector defined like $\mathbf{V}_t$. We have
\begin{align*}
&\mathbf V_t-\hat{\mathbf V}_t\\
=&\,\alpha_t \begin{bmatrix}\sum_{j=t-m}^{t+m}(\mathbf S_{j}-\hat{\mathbf S}_{j}) (\mathbf{h}(j-t)-\mathbf{h}(j-1-t))-({\mathbf S}_{t-m-1}-\hat {\mathbf S}_{t-m-1})\mathbf h(-m)\\
\sum_{j=t-m}^{t+m}(\mathbf S_{j}-\hat{\mathbf S}_{j}) (\mathbf{Dh}(j-t)-\mathbf{Dh}(j-1-t))-({\mathbf S}_{t-m-1}-\hat {\mathbf S}_{t-m-1})\mathbf{Dh}(-m)
\end{bmatrix}\,.
\end{align*}
A direct bound leads to
\begin{align*}
&\Big|\frac{1}{\sqrt{n}}(\mathbf V_t-\hat{\mathbf V}_t)\Big|^2\\
\leq&\,2\sum_{j=t-m}^{t+m}\Big|\frac{1}{\sqrt{n}}(\mathbf S_j-\hat{\mathbf S}_j)\Big|^2( |\mathbf{h}(j-t)-\mathbf{h}(j-1-t)|^2+|\mathbf{Dh}(j-t)-\mathbf{Dh}(j-1-t)|^2)\nonumber\\
&\,+2\Big|\frac{1}{\sqrt{n}}({\mathbf S}_{t-m-1}-\hat {\mathbf S}_{t-m-1})\Big|^2(\mathbf h(-m)^2+\mathbf{Dh}(-m)^2)\nonumber\\
=&\,2\sum_{j=t-m}^{t+m}\Big|\frac{1}{\sqrt{n}}(\mathbf S_j-\hat{\mathbf S}_j)\Big|^2( |\mathbf{h}(j-t)-\mathbf{h}(j-1-t)|^2+|\mathbf{Dh}(j-t)-\mathbf{Dh}(j-1-t)|^2)\nonumber\,,
\end{align*}
where the last equality holds since $\mathsf{h}$ is assumed to be supported on $[-1/2,1/2]$.
We then control $\frac{1}{\sqrt{n}}(\mathbf S_j-\hat{\mathbf S}_j)$ by $\max_{t=1,\ldots, n} \big|\frac{1}{\sqrt{n}}(\mathbf S_t-\hat{\mathbf S}_t)\big|^2$ using Theorem \ref{thm:OurSequentialGaussianApproximationTheorem}. Note that $\mathbf X_i\in \mathbb{R}^{2d}$ fulfills the necessary assumptions by Lemma \ref{lemme: for theorem of STFT distribution like Gaussian}.
As a result, we obtain
\begin{align}
\left[\mathbb{E}\left(\max_{t=1,\ldots, n} \Big|\frac{1}{\sqrt{n}}(\mathbf S_t-\hat{\mathbf S}_t)\Big|^2\right)\right]^{1/2}\leq C\left(\frac{d^{\frac{3}{4}-\frac{1}{2s}-\frac{1}{p}}}{n^{\frac{1}{4}-\frac{1}{2p}-\frac{1}{2s}+\frac{1}{ps}}}\right)^{\frac{1}{1-\frac{1}{s}-\frac{1}{p}}} \log(n)^2\,,\nonumber
\end{align} 
where $C>0$ is a constant depending on $p$ and the dependence structure and moment control of $X_i$. 
On the other hand, by the definition of $\mathbf{h}(j-t)$, we have
\begin{align*}
&\,\max_{t=1,\ldots,n}\sum_{j=t-m}^{t+m} (|\mathbf{h}(j-t)-\mathbf{h}(j-1-t)|^2+|\mathbf{Dh}(j-t)-\mathbf{Dh}(j-1-t)|^2)\\
= &\,\max_t\frac{1}{m}\sum_{j=t-m}^{t+m}\left(\left|\mathsf{h}\left(\frac{j-t}{2m+1}\right)-\mathsf{h}\left(\frac{j-t-1}{2m+1}\right)\right|^2+\left|\mathsf{h}'\left(\frac{j-t}{2m+1}\right)-\mathsf{h}'\left(\frac{j-t-1}{2m+1}\right)\right|^2\right) \nonumber\\
=&\,\Theta(m^{-2})\,,\nonumber
\end{align*}
where the implied constant depends on  $\|\mathsf{h}'\|_2$, $\|\mathsf{h}''\|_2$, $\|\mathsf{h}''\|_\infty$ and  $\|\mathsf{h}'''\|_\infty$ since $m\left(\mathsf{h}\left(\frac{j-t}{2m+1}\right)-\mathsf{h}\left(\frac{j-t-1}{2m+1}\right)\right)$ is a finite difference approximation of $\mathsf{h}'\left(\frac{j-t-1}{2m+1}\right)$ and $\frac{1}{m}\sum_{j=t-m}^{t+m}\left|\mathsf{h}'\left(\frac{j-t}{2m+1}\right)\right|^2$ is a Riemann sum approximation of $\|\mathsf{h}'\|_2$.
By putting the above together, we have the claim by
\begin{align*}
&\mathbb{E}\left(\max_t \Big|\frac{1}{\sqrt{n}}(\mathbf V_t-\hat{\mathbf V}_t)\Big|^2\right)\\
\leq&\,
\mathbb{E}\,\bigg(\max_t \sum_{j=t-m}^{t+m}\max_{j\leq n}\Big|\frac{1}{\sqrt{n}}(\mathbf S_j-\hat{\mathbf S}_j)\Big|^2 \\
&\quad\times(|\mathbf{h}(j-t)-\mathbf{h}(j-1-t)|^2+|\mathbf{Dh}(j-t)-\mathbf{Dh}(j-1-t)|^2)\bigg)\nonumber\\
\leq&\,
\mathbb{E}\left(\max_{j\leq n}\Big|\frac{1}{\sqrt{n}}(\mathbf S_j-\hat{\mathbf S}_j)\Big|^2\right)\\
&\quad\times \max_t 
\sum_{j=t-m}^{t+m} (|\mathbf{h}(j-t)-\mathbf{h}(j-1-t)|^2+|\mathbf{Dh}(j-t)-\mathbf{Dh}(j-1-t)|^2)\nonumber\\
\leq &\, {C'}\left(\frac{d^{\frac{3}{4}-\frac{1}{2s}-\frac{1}{p}}}{n^{\frac{1}{4}-\frac{1}{2p}-\frac{1}{2s}+\frac{1}{ps}}}\right)^{\frac{2}{1-\frac{1}{s}-\frac{1}{p}}} \log(n)^4 m^{-2}\,,\nonumber
\end{align*}
where $C'>0$ is a constant depending on $\mathsf{h}$ and the moment and dependence structures of $\epsilon_i$. 
\end{proof}

\subsection{Proof of Theorem \ref{section:theorem:stability}: robustness of SST-based reconstruction formula}\label{subsection proof of sst reconstruction robustness}

We require the following lemma concerning the magnitude of the discrete STFT of noise. This lemma is an immediate consequence of Theorem \ref{Theorem: STFT distribution like Gaussian}. This result may be of independent interest.

\begin{lem}\label{Lemma: STFT on HDNS noise control size}
Suppose $\{\epsilon_i\}_{i=1}^n$ satisfies the NSN model and Assumption  \ref{model assumption 4}, with the condition $A > \sqrt{\chi} + 1$ fulfilled. Suppose $d=n^{1/3-\gamma}$ for a small $\gamma\in (0,1/3)$ and
consider $0<\eta_1<\ldots<\eta_d$. 
Denote a complex random vector associated with the discrete STFT \eqref{STFT definition discrete}:
\begin{align*}
\mathbf V_l:=&\,[\mathbf V^{(\mathbf{h})}_\epsilon(l,\eta_1),\ldots, \mathbf V^{(\mathbf{h})}_\epsilon(l,\eta_d),\mathbf V^{(\mathbf{Dh})}_\epsilon(l,\eta_1),\ldots,\mathbf V^{(\mathbf{Dh})}_\epsilon(l,\eta_d)]^\top\in \mathbb{C}^{2d}\,,
\end{align*}
where 
$m:=\lceil \beta \sqrt{n}\rceil$.
Then we have
\[
\max_l\left|\mathbf V_l\right|=o_p(\sqrt{d\log d})\,,
\]
and
\[
\max_l |{\mathbf V}_l|_\infty \leq C'\sqrt{\log d}
\] 
with probability higher than $1-d^{-3}$.  
\end{lem}

\begin{proof}   
By Theorem \ref{Theorem: STFT distribution like Gaussian}, we know that $\mathbf V^{(\mathbf{h})}_\epsilon(l,\eta_k)$ is approximately Gaussian in the following sense. Denote
\begin{align}
\hat{\mathbf V}_l:=&\,[\mathbf V^{(\mathbf{h})}_{\hat{\epsilon}}(l,\eta_1),\ldots,\mathbf V^{(\mathbf{h})}_{\hat{\epsilon}}(l,\eta_d),\mathbf V^{(\mathbf{Dh})}_{\hat{\epsilon}}(l,\eta_1),\ldots,\mathbf V^{(\mathbf{Dh})}_{\hat{\epsilon}}(l,\eta_d)]^\top\in \mathbb{C}^{2d}, \nonumber
\end{align}
where $\hat{\epsilon}$ is the Gaussian approximation of $\epsilon$ stated in Theorem \ref{Theorem: STFT distribution like Gaussian}.
Then, we have
\begin{align*}
\mathbb{E}\left(\max_l \Big|\frac{1}{\sqrt{n}}(\mathbf V_l-\hat{\mathbf V}_l)\Big|^2\right)
\leq {C'}d \left(\frac{d^{1/2}}{n^{1/2-1/s}}\right)^{\frac{1-2/p}{3/2-1/s-1/p}} (\log n)^2m^{-2}\,,\nonumber
\end{align*}
where $C'$ is a universal constant. 
Since
\[
|\mathbf V^{(\mathbf{h})}_\epsilon(l,\eta_k)|\leq |\mathbf V^{(\mathbf{h})}_{\hat{\epsilon}}(l,\eta_k)|+|\mathbf V^{(\mathbf{h})}_\epsilon(l,\eta_k)-\mathbf V^{(\mathbf{h})}_{\hat{\epsilon}}(l,\eta_k)|
\]
and a similar bound holds when the superscript $\mathsf{h}$ is replaced by $\mathsf{h}'$,
we have
\[
\max_l|\mathbf V_l|\leq \max_l|\hat{\mathbf V}_l|+\max_l|\mathbf V_l-\hat{\mathbf V}_l|\,.
\]
By the Chebychev inequality, we have
\[
\max_l \Big|\frac{1}{\sqrt{n}}(\mathbf V_l-\hat{\mathbf V}_l)\Big|
=o_p\left(\left(\frac{d^{\frac{3}{4}-\frac{1}{2s}-\frac{1}{p}}}{n^{\frac{1}{4}-\frac{1}{2p}-\frac{1}{2s}+\frac{1}{ps}}}\right)^{\frac{1}{1-\frac{1}{s}-\frac{1}{p}}} \frac{\log(n)}{m}\right)\,.
\]
On the other hand, 
since $\hat{\mathbf V}_l$ is a non-degenerate $2d$-dim complex Gaussian random vector whose entrywise variance is of  order $1$, we have
\[
|\hat{\mathbf V}_l|_\infty \leq C\sqrt{\log d}\,,
\] 
and hence
\[
|\hat{\mathbf V}_l|\leq C\sqrt{d\log d}\,,
\]
with probability higher than $1-d^{-3}$. 
By a union bound, $\max_l \Big|\frac{1}{\sqrt{n}}\hat{\mathbf V}_l\Big|^2$ can be bounded by $\frac{Cd\log d}{n}$ for some constant $C>0$ with probability higher than $1-nd^{-3}$. 
Therefore, we have
\[
\max_l\left|\mathbf V_l\right|=o_p\left(\left(\frac{d^{\frac{3}{4}-\frac{1}{2s}-\frac{1}{p}}}{n^{\frac{1}{4}-\frac{1}{2p}-\frac{1}{2s}+\frac{1}{ps}}}\right)^{\frac{1}{1-\frac{1}{s}-\frac{1}{p}}} \frac{(\log n)^2\sqrt{n}}{m}+\sqrt{d\log d}\right)\,.
\]

To obtain the uniform bound, note that by construction, $\hat{\mathbf V}_l$ and ${\mathbf V}_l$ share the same covariance matrix. Thus, ${\mathbf V}_l$ a non-degenerate $2d$-dim complex random vector whose entrywise variance is of order $1$, by the same argument used in \cite{wu2024frequency} following the technique for maxima control \cite{chernozhukov2013,chernozhukov2015}, we obtain
\[
\max_l |{\mathbf V}_l|_\infty \leq C'\sqrt{\log d}
\] 
with probability higher than $1-d^{-3}$, where $C'>0$ depends on the moment and dependence structure of $\epsilon_i$.

\end{proof}

The following result is about how STFT behaves on an oscillatory signal satisfying AHM in the discretized setup \eqref{STFT definition continuous discretization}. We precisely describe its behavior and the associated numerical approximation error.

\begin{lem}\label{Lemma: SST on deterministic function AHM}
Suppose $\{f_i\}_{i=1}^n$ in model \eqref{eq:model} satisfies $f_i=f(t_i)$, where $f$ satisfies the AHM with Assumption \eqref{model assumption 1} fulfilled and $t_i=i/\sqrt{n}$. Denote $\mathbf f:=[f(t_1),\ldots,f(t_n)]^\top\in \mathbb{R}^n$. Assume Assumption \eqref{assumption window function main theorem} holds. Denote 
$J^{(m)}_k:=\int_{-\infty}^\infty |u|^k|\mathsf{h}^{(m)}(u)|du$, 
the $k\in \mathbb{N}\cup\{0\}$ th absolute moment of the $m$-th derivative of $\mathsf{h}$. 
Then, for any $\eta\in (0,\sqrt{n}/2)$,   
\begin{align}
E_{f,l,\eta}:=V^{(\mathsf{h})}_{\mathbf f}(t_l,\eta)-A(t_l)e^{\ii2\pi \phi(t_l)}\widehat{\mathsf{h}}(\eta-\phi'(t_l))\,,
\end{align}
satisfies
\begin{align}
E_f:=\max_l\sup_{\eta\in (0,\Xi)} |E_{f,l,\eta}|\leq \,&\varepsilon \Xi_A(J^{(0)}_1+\pi \Xi J^{(0)}_2)+\frac{\Xi_A((\varepsilon+\Xi)\|\mathsf{h}\|_\infty+\|\mathsf{h}'\|_\infty)\beta}{\sqrt{n}}\,,\nonumber
\end{align}
and 
\begin{align}
E'_{f,l,\eta}:= \frac{1}{2\pi \ii}V_{\mathbf{f}}^{(\mathsf{h}')}(t_l,\eta)-(\eta-\phi'(t_l))A(t_l)\widehat{\mathsf{h}}(\eta-\phi'(t_l))e^{\ii 2\pi\phi(t)} 
\end{align}
satisfies
\begin{align}
E'_f:=\,&\max_l\sup_{\eta\in (0,\Xi)} |E'_{f,l,\eta}|\nonumber\\
\leq \,&\varepsilon \Xi_A\left(\frac{J^{(0)}_0}{2\pi}+2\Xi J^{(0)}_1+\pi\Xi^2 J^{(0)}_2\right)+\frac{\Xi_A((\varepsilon+\Xi)\|\mathsf{h}'\|_\infty+\|\mathsf{h}''\|_\infty)\beta}{\sqrt{n}}\,.\nonumber
\end{align}
The reassignment rule satisfies that when $|V^{(\mathsf{h})}_{\mathbf f}(t_l,\eta)|>\tau$ for some $\tau\geq 0$, we have
\begin{align}\label{Proof stability recon SST Of}
E_{O,l,\eta}:=O^{(\tau)}_{\mathbf f}(t_l,\eta)-\phi'(t_l)
\end{align}
satisfies 
\begin{align}
E^{(\tau)}_O:=\max_{l}\max_{\substack{\eta\in(0,\Xi)\\|V^{(\mathsf{h})}_{\mathbf f}(t_l,\eta)|>\tau}}|E_{O,l,\eta}|\leq
\frac{E'_f+\Xi E_f}{\tau} \,.\nonumber
\end{align}
Set  $B_l:=[\phi'(t_l)-\Delta,\,\phi'(t_l)+\Delta]$. 
Then, we have
\begin{align}
E_{r,l}:=\frac{1}{\mathsf{h}(0)}\frac{\Xi}{d}\sum_{k\in B_l}V_{\mathbf{f}}^{(\mathsf{h})}(t_l,\eta_k)-A(t_l)e^{\ii2\pi \phi(t_l)}\label{definition Er error for STFT reconstruction}
\end{align}
satisfying
\begin{align*}
E_r:=\max_l|E_{r,l}|
\leq \,& \frac{\pi^2\Xi^2\Xi_A\Delta}{3d^2\mathsf{h}(0)}\left((1+ \varepsilon \beta)  J^{(0)}_2+\frac{4\beta((\varepsilon \|\mathsf{h}\|_\infty+\|\mathsf{h}'\|_\infty) \beta+\|\mathsf{h}\|_\infty)}{\sqrt{n}}\right)\\
& + \frac{2(E_f \Delta +\Xi_A \beta^{-1/2} \delta_2)}{\mathsf{h}(0)}\,,
\end{align*}
where $\frac{\Xi}{d}\sum_{k\in B_l}V_{\mathbf{f}}^{(\mathsf{h})}(t_l,\eta_k)$ is the reconstruction formula associated with the discretized STFT.
\end{lem}

\begin{remark}
In this lemma, we explicitly identify each error term, as each one plays a distinct role in establishing the uniform error bound required later. For instance, consider \eqref{definition Er error for STFT reconstruction}. Since $d=d(n)\to \infty$ as $n\to\infty$, when $n\to \infty$, the reconstruction error asymptotically approaches $ \frac{2(E_f \Delta +\Xi_A \beta^{-1/2} \delta)}{\mathsf{h}(0)}$. Moreover, when the signal is harmonic; that is, when $\varepsilon=0$, this error simplifies to $ \frac{2\Xi_A \beta^{-1/2} \delta}{\mathsf{h}(0)}$, which corresponds to the spectral truncation in the reconstruction formula. Notably, all error terms scale with $\Xi_A$, highlighting that the reconstruction error scales properly according to the signal's amplitude. 
\end{remark}

\begin{proof}
The behavior of $V^{(\mathsf{h})}_f(t_1,\eta)$ and $V^{(\mathsf{h}')}_f(t_1,\eta)$ is well known in the literature. See, for example, \cite{Wu2011adaptive}. It is worth noting that in the main result of interest, \cite[Theorem 2.3.14]{Wu2011adaptive}, the control over $|A'(t)|$ differs from ours. Here, we choose to emphasize the role of $A(t)$ directly, as it allows us to achieve a clearer and more uniform bound in the final result.

The study of $V^{(\mathbf{h})}_{\mathbf f}(t_l,\eta)$ and $V^{(\mathbf{Dh})}_{\mathbf f}(t_l,\eta)$ is via viewing them as Riemann sums and studying the numerical approximation error. Note that there are $\lfloor 2\beta/\sqrt{n}\rfloor$ sampled points over $[-\beta,\beta]$ since $\mathsf{h}$ is supported on $[-\beta,\beta]$. Also, we have $\|(f\mathsf{h})'\|_\infty\leq \|A'\|_\infty\|\mathsf{h}\|_\infty+\|A\|_\infty\|\phi'\|_\infty\|\mathsf{h}\|_\infty+\leq \|A\|_\infty\|\mathsf{h}'\|_\infty\leq \Xi_A((\varepsilon+\Xi)\|\mathsf{h}\|_\infty+\|\mathsf{h}'\|_\infty)$. By the Riemann sum error approximation, we have
\begin{align}
&\left|V_{\mathbf{f}}^{(\mathsf{h})}(t_l,\eta)-V_f^{(\mathsf{h})}(t_l,\eta)\right|\label{STFT discretization approximation by continous STFT by Riemann sum}\\
=\,&\left|\frac{1}{\sqrt{n}}\sum_{j=1}^n f(t_j)\mathsf{h}(t_j-t_l)e^{-\ii 2\pi \eta(t_j-t_l)}-V_f^{(\mathsf{h})}(t_l,\eta)\right|\nonumber\\
\leq\,& \frac{\Xi_A((\varepsilon+\Xi)\|\mathsf{h}\|_\infty+\|\mathsf{h}'\|_\infty)\beta}{\sqrt{n}}\nonumber\,,
\end{align}
and hence the claimed bound
\begin{align}
&\left|V_{\mathbf{f}}^{(\mathsf{h})}(t_l,\eta)-A(t_l)\widehat{\mathsf{h}}(\eta-\phi'(t_l))e^{\ii 2\pi\phi(t)}\right|\nonumber\\
\leq &\, \varepsilon \Xi_A(J^{(0)}_1+\pi \Xi J^{(0)}_2)+\frac{\Xi_A((\varepsilon+\Xi)\|\mathsf{h}\|_\infty+\|\mathsf{h}'\|_\infty)\beta}{\sqrt{n}}=:E_f\,,\nonumber
\end{align}
where the bound comes from \cite[Theorem 2.3.14]{Wu2011adaptive}. Note that in \cite[Theorem 2.3.14]{Wu2011adaptive}, the control of $|A'(t)|$ is different from what we consider here. Similarly, we have
\begin{align} 
\,&\left|\frac{1}{2\pi \ii}V_{\mathbf{f}}^{(\mathsf{h}')}(t_l,\eta)-\frac{1}{2\pi \ii}V_f^{(\mathsf{h}')}(t_l,\eta)\right|\leq \frac{\Xi_A((\varepsilon+\Xi)\|\mathsf{h}'\|_\infty+\|\mathsf{h}''\|_\infty)\beta}{\sqrt{n}}\nonumber
\end{align}
and hence with \cite[Theorem 2.3.14]{Wu2011adaptive}, we obtain
\begin{align}
&\left| \frac{1}{2\pi \ii}V_{\mathbf{f}}^{(\mathsf{h}')}(t_l,\eta)-(\eta-\phi'(t_l))A(t_l)\widehat{\mathsf{h}}(\eta-\phi'(t_l))e^{\ii 2\pi\phi(t)}\right| \nonumber\\
\leq \,&\varepsilon \Xi_A\left(\frac{J^{(0)}_0}{2\pi}+2\Xi J^{(0)}_1+\pi\Xi^2 J^{(0)}_2\right)+\frac{\Xi_A((\varepsilon+\Xi)\|\mathsf{h}\|_\infty+\|\mathsf{h}'\|_\infty)\beta}{\sqrt{n}}=:E_f'\nonumber\,.
\end{align}

For the reassignment rule, note that by the above Riemann sum approximation, when $n$ is sufficiently large, we have when $|V^{(\mathsf{h})}_{\mathbf f}(t_l,\eta)|> \tau$ for $\tau>0$, where $\eta\in(0,\Xi)$, 
\begin{align*}
\,&|O^{(\tau)}_{\mathbf f}(t_l,\eta)-\phi'(t_l)|=\left|\frac{-1}{2\pi i} 
\frac{V^{(\mathsf{h}')}_{\mathbf f}(t_l,\eta)}{V^{(\mathsf{h})}_{\mathbf f}(t_l,\eta)} -\phi'(t_l)\right|\\
=\,&\left|\frac{-A(t_l)\phi'(t_l)e^{\ii2\pi \phi(t_l)}\widehat{h}(\eta-\phi'(t_l))-E'_{f,l,k}}{V^{(\mathsf{h})}_{\mathbf f}(t_l,\eta)}-\phi'(t_l)\right| \\
\leq\,&\frac{E'_f+\phi'(t_l)E_f}{|V^{(\mathsf{h})}_{\mathbf f}(t_l,\eta)|}\leq \frac{E'_f+\Xi E_f}{|V^{(\mathsf{h})}_{\mathbf f}(t_l,\eta)|}\,
\end{align*}
and hence the claim. 

For the final claim about the discretized STFT reconstruction formula, rewrite 
\[
\frac{\Xi}{d}\sum_{k\in B_l}V_{\mathbf{f}}^{(\mathsf{h})}(t_l,\eta_k)
=\frac{1}{\sqrt{n}}\sum_{j=1}^n f(t_j)\mathsf{h}(t_j-t_l)\left[\frac{\Xi}{d}\sum_{k\in B_l}e^{-\ii 2\pi \eta_k(t_j-t_l)}\right]\,.
\]
By the Riemann sum approximation, we obtain
\begin{align*}
&\left|\frac{\Xi}{d}\sum_{k\in B_l}e^{-\ii 2\pi \eta_k(t_j-t_l)}-\int_{B_l}e^{-\ii 2\pi \eta(t_j-t_l)}d\eta\right|\leq \frac{\pi^2\Xi^2\Delta}{3d^2}(t_j-t_l)^2\,.
\end{align*}
Therefore, since $|f(t)|\leq A(t)$ and $\mathsf{h}$ is nonnegative by assumption, 
\begin{align}
&\left|\frac{\Xi}{d}\sum_{k\in B_l}V^{(\mathsf{h})}_{\mathbf f}(t_l,\eta_k)-\int_{B_l}V^{(\mathsf{h})}_{\mathbf f}(t_l,\eta)d\eta\right|\label{STFT reconstruction discret STFT vs discrete discretized bound} \\
\leq&\,\frac{\pi^2\Xi^2\Delta}{3d^2}\frac{1}{\sqrt{n}}\sum_{j=1}^n A(t_j)\mathsf{h}(t_j-t_l)(t_j-t_l)^2\nonumber\\
\leq &\,\frac{\pi^2\Xi^2\Xi_A\Delta}{3d^2}\left((1+ \varepsilon \beta)  J^{(0)}_2+\frac{4\beta((\varepsilon \|\mathsf{h}\|_\infty+\|\mathsf{h}'\|_\infty) \beta+\|\mathsf{h}\|_\infty)}{\sqrt{n}}\right)\,,\nonumber
\end{align}
where the last bound comes from the Riemann sum approximation
\begin{align*}
&\left|\frac{1}{\sqrt{n}}\sum_{j=1}^n A(t_j)\mathsf{h}(t_j-t_l)(t_j-t_l)^2I(|t_j-t_l|\leq \beta)-\int_{-\beta}^{\beta} A(t+t_l)\mathsf{h}(t)t^2dt\right|\\
\leq \,&\frac{4\Xi_A\beta((\varepsilon \|\mathsf{h}\|_\infty+\|\mathsf{h}'\|_\infty) \beta+\|\mathsf{h}\|_\infty)}{\sqrt{n}}
\end{align*}
and the approximation that
\[
\left|\int_{-\beta}^{\beta} A(t+t_l)\mathsf{h}(t)t^2dt-A(t_l)J^{(0)}_2\right|\leq \varepsilon \beta \Xi_A J^{(0)}_2\,.
\]
Moreover, 
\begin{align}
&\left|\frac{1}{\mathsf{h}(0)} \int_{B_l}V^{(\mathsf{h})}_{\mathbf f}(t_l,\eta)d\eta-
 A(t_l)e^{\ii2\pi \phi(t_l)} \right| \label{STFT reconstruction discret STFT vs final bound} \\
=\,&\frac{1}{\mathsf{h}(0)}\left|\int_{B_l}V^{(\mathsf{h})}_{\mathbf f}(t_l,\eta)d\eta-
\int_{-\infty}^\infty A(t_l)e^{\ii2\pi \phi(t_l)}\widehat{\mathsf{h}}(\eta-\phi'(t_l)) d\eta\right| \nonumber\\
\leq \,&\frac{1}{\mathsf{h}(0)}\left| \int_{B_l} \left[V^{(\mathsf{h})}_{\mathbf f}(t_l,\eta) - A(t_l)e^{\ii2\pi \phi(t_l)}\widehat{\mathsf{h}}(\eta-\phi'(t_l))\right] d\eta\right|\nonumber \\
&+ \frac{1}{\mathsf{h}(0)}\left|\int_{\mathbb{R}\backslash B_l} A(t_l)e^{\ii2\pi \phi(t_l)}\widehat{\mathsf{h}}(\eta-\phi'(t_l)) d\eta\right|\nonumber\\
\leq \,& \frac{2  E_f \Delta}{\mathsf{h}(0)} + \frac{\Xi_A}{\mathsf{h}(0)} \int_{\mathbb{R}\backslash B_l} |\widehat{\mathsf{h}}(\eta-\phi'(t_l))| d\eta\leq  \frac{2E_f \Delta +\Xi_A \beta^{-1/2} \delta_2}{\mathsf{h}(0)}\,.\nonumber
\end{align}
By putting \eqref{STFT reconstruction discret STFT vs discrete discretized bound}  and \eqref{STFT reconstruction discret STFT vs final bound} together, when $n$ is sufficiently large, we obtain 
\begin{align*}
&\left|\frac{1}{\mathsf{h}(0)}\frac{\Xi}{d}\sum_{k\in B_l}V_{\mathbf{f}}^{(\mathsf{h})}(t_l,\eta_k)-A(t_l)e^{\ii2\pi \phi(t_l)}\right|\\
\leq \,& \frac{\pi^2\Xi^2\Xi_A\Delta}{3d^2\mathsf{h}(0)}\left((1+ \varepsilon \beta)  J^{(0)}_2+\frac{4\beta((\varepsilon \|\mathsf{h}\|_\infty+\|\mathsf{h}'\|_\infty) \beta+\|\mathsf{h}\|_\infty)}{\sqrt{n}}\right) + \frac{2(E_f \Delta +\Xi_A \beta^{-1/2} \delta)}{\mathsf{h}(0)}\,,
\end{align*}
and hence the claimed bound.  
\end{proof}

\begin{corro}\label{corro: SST on deterministic function AHM}
Grant the same notation and assumptions in Lemma \ref{Lemma: SST on deterministic function AHM}. 
Set $B_l:=[\phi'(t_l)-\Delta,\,\phi'(t_l)+\Delta]$. 
Then, when $n$ is sufficiently large, we have  
\begin{align}
\min_l \min_{\eta\in B_l} |V^{(\mathsf{h})}_{\mathbf f}(t_l,\eta)|\geq \Xi_a\sqrt{\beta}\delta_1/2=:\nu_0 
\end{align}
for some constant $C'_{\mathsf{h}}>2$ depending on $\mathsf{h}$.

\end{corro}

\begin{proof}
The proof is an immediate consequence of Lemma \ref{Lemma: SST on deterministic function AHM}. Since for any $l=1,\ldots,n$, by the assumption of $\mathsf{h}$, when $\eta\in B_l$, we have
\begin{align}
|V^{(\mathsf{h})}_{\mathbf f}(t_l,\eta)|\geq A(t_l) |\widehat{\mathsf{h}}(\eta-\phi'(t_l))|- 2\varepsilon \Xi_A(J^{(0)}_1+\pi \Xi J^{(0)}_2)\geq \Xi_a\sqrt{\beta}\delta_1/2\nonumber\,.
\end{align}

\end{proof}

The following lemma prepares a control of a key quantity that we need when we analyze the SST reconstruction formula. We start with a clean case, which will serve as the base for the noisy case.
\begin{lem}\label{Lemma theorem:stability g_alpha integration clean case}
Suppose  $f_i=f(t_i)$, $f$ satisfies the AHM with $\varepsilon>0$ and Assumption \eqref{model assumption 1} fulfilled, $t_i=i/\sqrt{n}$. Set $m=\lceil \beta \sqrt{n}\rceil$ and $d=d(n)\to \infty$ when $n\to \infty$.
Set $\eta_k=\frac{k\Xi}{d}$, where $k=1,\ldots,d$. Grant the kernel assumption in Assumption \ref{assumption window function main theorem} and choose $\Delta$ so that $\frac{2(E_f'+\Xi E_f)}{\Xi_a\sqrt{\beta}\delta_1}\leq 1/2$.
For $l\in\{1,\ldots,n\}$, denote 
\begin{align}\label{definition Ql, Rl and Bl in the proof}
R_l:=\,&[\phi'(t_l)-\Delta_r,\,\phi'(t_l)+\Delta_r]\\
B_l:=\,&[\phi'(t_l)-\Delta,\,\phi'(t_l)+\Delta]\nonumber\\
Q_l:=\,&\{\eta\in(0,\Xi)|\, |V^{(\mathsf{h})}_{\mathbf f}(t_l,\eta)|>\nu_0\}\,. \nonumber
\end{align}
We have $B_l\subset Q_l$. 
Set 
\begin{align}\label{definition of nu0 and alpha}
\nu_0=\Xi_a\sqrt{\beta}\delta_1/2\ \mbox{ and }\ \alpha=\left(\frac{\Delta_r}{2C_\alpha}\right)^2
\end{align}
for some $C_\alpha>1$.
Then, when $n$ is sufficiently large, we have
\begin{align}
E_{g,l,k}:=\frac{\Xi}{d}\sum_{\eta_j\in R_l} g_\alpha\big(\eta_j- O^{(\nu_0)}_{\mathbf{f}}(t_l,\eta_k)\big)-1\,,\label{definition of error term of galpha sum of clean signal Eg}
\end{align}
satisfying
\begin{align}
E_g:=\max_l\max_{\eta_k\in Q_l}|E_{g,l,k}|\leq 2\texttt{erfc}(C_\alpha)+\frac{\Xi C_\alpha^3}{\sqrt{\pi}\Delta_r^3 d}\,.\nonumber
\end{align}
When $\eta_k\notin Q_l$, we have
\[
\frac{\Xi}{d}\sum_{\eta_j\in R_l} g_\alpha\big(\eta_j- O^{(\nu_0)}_{\mathbf f}(t_l,\eta_k)\big)=0\,.
\]
\end{lem}

\begin{remark}
Note that there are two dominant terms controlling $E_g$. The first term is regarding the truncation of integration in the frequency domain, and the second term is regarding the discretization of the frequency domain. 
\end{remark}

\begin{proof}

Since $\eta_j$ is a uniform discretization in the frequency axis, we approximate $\frac{\Xi}{d}\sum_{\eta_j\in R_l} g_\alpha\big(\eta_j- O^{(\nu_0)}_{\mathbf f}(t_l,\eta_k)\big)$ by $\int_{R_l} g_\alpha(\eta-O^{(\nu_0)}_{\mathbf f}(t_l,\eta_k))d\eta$. %
By the Riemann sum approximation, when $|V^{(\mathsf{h})}_{\mathbf f}(t_l,\eta)|>\nu_0$, we have  
\[
\left|\frac{\Xi}{d}\sum_{\eta_j\in R_l} g_\alpha\big(\eta_j- O^{(\nu_0)}_{\mathbf f}(t_l,\eta_k)\big) - \int_{R_l} g_\alpha(\eta-O^{(\nu_0)}_{\mathbf f}(t_l,\eta_k))d\eta\right|\leq \frac{\Xi C_\alpha^3}{\sqrt{\pi}\Delta_r^3 d}
\]
since the integrant is smooth and over $R_l$ there are $\lfloor \frac{2\Delta_r d}{\Xi}\rfloor$ uniform discrete points.

Take the reassignment control in Corollary \ref{corro: SST on deterministic function AHM} that 
\[
\max_{l}\max_{\substack{\eta\in(0,\Xi)\\|V^{(\mathsf{h})}_{\mathbf f}(t_l,\eta)|>\tau}}|O^{(\nu_0)}_{\mathbf f}(t_l,\eta_k)-\phi'(t_l)|\leq \frac{E_f'+\Xi E_f}{\nu_0}\,.
\]
By assumption, $\frac{E_f'+\Xi E_f}{\nu_0}\leq 1/2$, so we obtain
\[
\left|\int_{R_l} g_\alpha(\eta-O^{(\nu_0)}_{\mathbf{f}}(t_l,\eta_k))d\eta-1\right|\leq
2\texttt{erfc}(C_\alpha) \,,
\]
where we simply bound $\frac{E_f'+\Xi E_f}{\nu_0}<C_\alpha$ when $\varepsilon$ is sufficiently small. 
As a result, when $\eta_k\in Q_l$, we have that
\begin{align}\label{proof sum g_alpha when input is f result}
E_{g,l,k}:=\frac{\Xi}{d}\sum_{\eta_j\in R_l} g_\alpha\big(\eta_j- O^{(\nu_0)}_{\mathbf f}(t_l,\eta_k)\big)-1
\end{align}
satisfies
\begin{align}
E_{g}:=\max_l\max_{k\in Q_l}|E_{g,l,k}|\leq 2\texttt{erfc}(C_\alpha)+\frac{\Xi C_\alpha^3}{\sqrt{\pi}\Delta_r^3 d}\,.
\end{align}

When $\eta_k\notin Q_l$, 
we have $O^{(\nu_0)}_{\mathbf f}(t_l,\eta_k)=-\infty$, and hence $g_\alpha\big(\eta_j- O^{(\nu_0)}_{\mathbf f}(t_l,\eta_k)\big)=0$, in which case we have 
\[
\frac{\Xi}{d}\sum_{\eta_j\in R_l} g_\alpha\big(\eta_j- O^{(\nu_0)}_{\mathbf f}(t_l,\eta_k)\big)= 0\,.
\]
\end{proof}

\begin{lem}\label{Lemma theorem:stability g_alpha integration}
Grant the setup in Lemma \ref{Lemma theorem:stability g_alpha integration clean case}.
Consider $\{Y_i\}_{i=1}^n$ follows model \eqref{eq:model}, where $f_i=f(t_i)$, $f$ satisfies the AHM with $\varepsilon>0$ and Assumption \eqref{model assumption 1} fulfilled, $t_i=i/\sqrt{n}$, $\{\epsilon_i\}_{i=1}^n$ satisfies the NSN model and Assumption  \ref{model assumption 4}, with the condition $A > \sqrt{\chi} + 1$ fulfilled, and $\sigma=\sigma(n)=n^{1/4-\gamma'}$, where $\gamma'\in (0,1/4]$ is a small constant. %
Set
\[
\nu :=\nu_0+\zeta_n\,,
\]
where $\nu_0$ is defined in \eqref{definition of nu0 and alpha} and $\zeta_n\asymp n^{-\gamma'}\sqrt{\log n}$ is specified in \eqref{proof sum galpha Vepsilon V'epsilon uniform bound}. 
Then, when $n$ is sufficiently large, with probability greater than $1-n^{-2}$, we have
\begin{align}
e_{g,l,k}:=\frac{\Xi}{d}\sum_{\eta_j\in R_l} g_\alpha\big(\eta_j- O^{(\nu/2)}_Y(t_l,\eta_k)\big)-1\,,
\end{align}
satisfying
\[
e_g:=\max_l\max_{\eta_k\in Q_l}|e_{g,l,k}|\leq E_g+\frac{108(1+E_g)\Xi^2C_\alpha^2}{\nu \Delta_r^2} \zeta_n \,,
\]
where $Q_l$ and $R_l$ are defined in \eqref{definition Ql, Rl and Bl in the proof} and $E_g$ is defined in \eqref{definition of error term of galpha sum of clean signal Eg}.
When $\eta_k\notin Q_l$, we have
\begin{align*}
\left|\frac{\Xi}{d}\sum_{\eta_j\in R_l} g_\alpha\big(\eta_j- O^{(\nu/2)}_{Y}(t_l,\eta_k)\big)-1\right| \leq e_g
\end{align*} 
when $|V^{(\mathsf{h})}_{\mathbf f}(t_l,\eta_k)|\geq \nu/2+\zeta_n$
and
\begin{align*}
&\left|\frac{\Xi}{d}\sum_{\eta_j\in R_l} g_\alpha\big(\eta_j- O^{(\nu/2)}_{Y}(t_l,\eta_k)\big)\right|\\
\leq\,& \left\{
\begin{array}{ll}
0&\mbox{ when } |V^{(\mathsf{h})}_{\mathbf f}(t_l,\eta_k)|< \nu/2-\zeta_n\\
\frac{2C_\alpha}{\pi} & \mbox{ when } \nu/2-\zeta_n\leq |V^{(\mathsf{h})}_{\mathbf f}(t_l,\eta_k)|< \nu/2+\zeta_n.
\end{array}
\right.
\end{align*}
\end{lem}

\begin{remark}

In general, note that the noise magnitude is controlled by $\sigma\zeta_n$, which assumed to decay to $0$ when $n\to \infty$. Thus, the impact of noise is asymptotically negligible. Particularly, the trivial bound over the region of $\nu/2-\zeta_n\leq |V^{(\mathsf{h})}_{\mathbf f}(t_l,\eta_k)|< \nu/2+\zeta_n$ does not contribute significantly in our upcoming analysis since the measure of this region decays to $0$ when $n\to \infty$. 

When the data is noise-free, we have $e_g=E_g$, since $\zeta_n$ in the last term in $e_g$ arises solely from controlling the noise magnitude. When $\eta_k\notin Q_l$, the bound 
$\left|\frac{\Xi}{d}\sum_{\eta_j\in R_l} g_\alpha\big(\eta_j- O^{(\nu/2)}_{Y}(t_l,\eta_k)\big)-1\right| \leq E_g$ holds
when $|V^{(\mathsf{h})}_{\mathbf f}(t_l,\eta_k)|\geq \nu_0/2$. 
Therefore, when $\epsilon=0$, we recover Lemma \ref{Lemma theorem:stability g_alpha integration clean case} with the truncation threshold set to $\nu_0/2$. 
\end{remark}

\begin{proof}
We start with some observations.
By the linearity of STFT and Lemma \ref{Lemma: STFT on HDNS noise control size}, we know that with probability greater than $1-n^{-3}$, $\max_l |\mathbf{V}_l|_\infty \leq C'\sqrt{\log d}$, and hence 
\begin{align}\label{proof sum galpha Vepsilon V'epsilon uniform bound}
&\max_{l,k}\max\{|V^{(\mathbf{h})}_{\sigma\epsilon}(t_l,\eta_k)|,\,|V^{(\mathbf{Dh})}_{\sigma\epsilon}(t_l,\eta_k)|\}\\
\leq \,&C'\sigma\sqrt{\log d}\lceil \beta \sqrt{n}\rceil^{-1/2}=:\zeta_n\asymp n^{-\gamma'}\sqrt{\log n}\,,\nonumber
\end{align}
where $C'>0$, with probability higher than $1-n^{-3}$. 
Denote the event subspace that \eqref{proof sum galpha Vepsilon V'epsilon uniform bound} holds as $\Omega$.
On the other hand, by Corollary \ref{corro: SST on deterministic function AHM}, we know 
$\min_l \min_{\eta\in Q_l} |V^{(\mathsf{h})}_{\mathbf f}(t_l,\eta)|\geq \nu_0$.
Clearly,  
since $\nu\geq\zeta_n$, when conditional on the event $\Omega$, we have 
\[
\max_{l,k}|V^{(\mathbf{h})}_{\sigma\epsilon}(t_l,\eta_k)|\leq \nu\,,
\]
and hence for any $\eta_k\in Q_l$, $|V^{(\mathbf{h})}_Y(t_l,\eta_k)|>||V^{(\mathbf{h})}_{\mathbf{f}}(t_l,\eta_k)|-|V^{(\mathbf{h})}_{\sigma\epsilon}(t_l,\eta_k)||\geq \nu/2$ holds since $\zeta_n\to 0$, and hence
\begin{align}\label{proof sum galpha VY lower bounded}
\min_l\min_{k\in Q_l}|V^{(\mathbf{h})}_Y(t_l,\eta_k)|> \nu/2\,.
\end{align}

With the above preparation, we consider two cases: when $\eta_k\in Q_l$ and when $\eta_k\notin Q_l$. For the first case when $\eta_k\in Q_l$, since
\begin{align}
&\frac{\Xi}{d}\sum_{\eta_j\in R_l}  g_\alpha\big(\eta_j- O^{(\nu/2)}_Y(t_l,\eta_k)\big)=\frac{\Xi}{d}\sum_{\eta_j\in R_l} g_\alpha\big(\eta_j- O^{(\nu_0)}_{\mathbf f}(t_l,\eta_k)\big)\nonumber\\
&\qquad\qquad+\frac{\Xi}{d}\sum_{\eta_j\in R_l} \left[g_\alpha\big(\eta_j- O^{(\nu/2)}_Y(t_l,\eta_k)\big)-g_\alpha\big(\eta_j- O^{(\nu_0)}_{\mathbf f}(t_l,\eta_k)\big)\right]\label{summation of galpha in the noisy case key control term}
\end{align}
and we have controlled $\frac{\Xi}{d}\sum_{\eta_j\in R_l} g_\alpha\big(\eta_j- O^{(\nu_0)}_{\mathbf f}(t_l,\eta_k)\big)$ in Lemma \ref{Lemma theorem:stability g_alpha integration clean case}, the result is obtained by controlling 
the second term \eqref{summation of galpha in the noisy case key control term}
caused by the NSN. Rewrite \eqref{summation of galpha in the noisy case key control term} as
\begin{align*}
&\frac{\Xi}{d}\sum_{\eta_j\in R_l} \left[g_\alpha\big(\eta_j- O^{(\nu/2)}_Y(t_l,\eta_k)\big)-g_\alpha\big(\eta_j- O^{(\nu_0)}_f(t_l,\eta_k)\big)\right]\\
=\,&\frac{\Xi}{d}\sum_{\eta_j\in R_l} \frac{1}{\sqrt{\pi\alpha}} e^{-\frac{1}{\alpha}|\eta_j- O^{(\nu_0)}_{\mathbf f}(t_l,\eta_k)|^2}\left(e^{-\frac{1}{\alpha}(|\eta_j- O^{(\nu/2)}_Y(t_l,\eta_k)|^2-|\eta_j- O^{(\nu_0)}_f(t_l,\eta_k)|^2)}-1\right)
\end{align*}
For any $l=1,\ldots,n$, by \eqref{Proof stability recon SST Of} and \eqref{proof sum galpha Vepsilon V'epsilon uniform bound}, we have
\begin{align*}
& |O^{(\nu/2)}_{Y}(t_l,\eta_k)-O^{(\nu_0)}_{\mathbf f}(t_l,\eta_k)|\\
=\,& \left|\frac{ V^{(\mathbf{h}')}_{\sigma \epsilon}(t_l,\eta_k)}{V^{(\mathbf{h})}_Y(t_l,\eta_k)}-\frac{ V^{(\mathbf{h})}_{\sigma\epsilon}(t_l,\eta_k)}{V^{(\mathbf{h})}_Y(t_l,\eta_k)}O^{(\nu_0)}_{\mathbf f}(t_l,\eta_k)\right|\leq  \frac{3\Xi}{\nu} \zeta_n \end{align*}
when conditional on the event $\Omega$, where we use the trivial bound $|O^{(\nu_0)}_{\mathbf f}(t_l,\eta_k)|+1\leq 1.5\Xi$ when $\varepsilon$ is sufficiently small and \eqref{proof sum galpha VY lower bounded}. 
Therefore, since $\max_l\max_{k\in Q_l}|O^{(\nu_0)}_{\mathbf{f}}(t_l,\eta_k)+\eta_k|\leq 2\phi'(t_l)$ when $\varepsilon$ is sufficiently small, we have $\max_l\max_{k\in Q_l}|O^{(\nu/2)}_{Y}(t_l,\eta_k)+\eta_k|\leq 3\phi'(t_l)$ when $n$ is sufficiently large. On the other hand, since $\eta_j\in R_l$, we have
$|O^{(\nu_0)}_{\mathbf{f}}(t_l,\eta_k)-\eta_j|\leq |O^{(\nu_0)}_{\mathbf{f}}(t_l,\eta_k)+\eta_k|+|\eta_k+\eta_j|\leq 4\phi'(t_l)$ and hence $|O^{(\nu/2)}_Y(t_l,\eta_k)-\eta_j|\leq  5\phi'(t_l)$.
Therefore,
\begin{align*}
&\max_l\max_{k\in Q_l}\Big||\eta_j- O^{(\nu/2)}_Y(t_l,\eta_k)|^2-|\eta_j- O^{(\nu_0)}_{\mathbf f}(t_l,\eta_k)|^2\Big|\\
\leq\,& \Big(|O^{(\nu/2)}_Y(t_l,\eta_k)-\eta_j|+|O^{(\nu_0)}_{\mathbf f}(t_l,\eta_k)-\eta_j|\Big)\Big|O^{(\nu/2)}_Y(t_l,\eta_k)-O^{(\nu_0)}_{\mathbf f}(t_l,\eta_k)\Big|\\
\leq\,& \frac{27\Xi^2}{\nu} \zeta_n\,.
\end{align*}
By Taylor expansion, we conclude that 
\begin{align*}
&\max_l\max_{k\in Q_l}\left|e^{-\frac{1}{\alpha}(|\eta_j- O^{(\nu/2)}_Y(t_l,\eta_k)|^2-|\eta_j- O^{(\nu_0)}_f(t_l,\eta_k)|^2)}-1\right|
\leq \frac{108\Xi^2C_\alpha^2}{\nu \Delta_r^2} \zeta_n\,.
\end{align*}
Hence, with \eqref{proof sum g_alpha when input is f result}, we conclude the desired bound
\begin{align*}
&\max_l\max_{k\in Q_l}\left|\frac{\Xi}{d}\sum_{\eta_j\in R_l} \Big(g_\alpha\big(\eta_j- O^{(\nu/2)}_Y(t_l,\eta_k)\big)-g_\alpha\big(\eta_j- O^{(\nu_0)}_{\mathbf f}(t_l,\eta_k)\big)\Big)\right|\\
\leq\,&\left(\max_l\max_{k\in Q_l}\frac{\Xi}{d}\sum_{\eta_j\in R_l} \frac{1}{\sqrt{\pi\alpha}} e^{-\frac{1}{\alpha}|\eta_j- O^{(\nu_0)}_{\mathbf f}(t_l,\eta_k)|^2}\right)\\
&\times \max_l\max_{k\in Q_l}\left|e^{-\frac{1}{\alpha}(|\eta_j- O^{(\nu/2)}_Y(t_l,\eta_k)|^2-|\eta_j- O^{(\nu_0)}_{\mathbf f}(t_l,\eta_k)|^2)}-1\right|\\
\leq\,& \frac{108(1+E_g)\Xi^2C_\alpha^2}{\nu \Delta_r^2} \zeta_n \,.
\end{align*}

For the second case when $\eta_k\notin Q_l$, note that we get $O^{(\nu_0)}_{\mathbf f}(t_l,\eta_k)=-\infty$ and hence $g_\alpha\big(\eta_j- O^{(\nu_0)}_{\mathbf f}(t_l,\eta_k)\big)=0$. 
By Corollary \ref{corro: SST on deterministic function AHM}, this case can only happen when $\eta_k\notin B_l$. 
We discuss three subcases: 
$|V^{(\mathsf{h})}_{\mathbf f}(t_l,\eta_k)|< \nu/2-\zeta_n$,
$\nu/2+\zeta_n\leq |V^{(\mathsf{h})}_{\mathbf f}(t_l,\eta_k)|< \nu_0$, and $\nu/2-\zeta_n\leq |V^{(\mathsf{h})}_{\mathbf f}(t_l,\eta_k)|< \nu/2+\zeta_n$.
In the first subcase, since we have $|V^{(\mathsf{h})}_{\epsilon}(t_l,\eta_k)|<\zeta_n$ when conditional on $\Omega$, $|V^{(\mathsf{h})}_{Y}(t_l,\eta_k)|\geq \nu/2$ cannot happen, and hence
\eqref{summation of galpha in the noisy case key control term} is reduced to $0$. 
In the second subcase, again, since we have $|V^{(\mathsf{h})}_{\epsilon}(t_l,\eta_k)|<\zeta_n$ when conditional on $\Omega$, $|V^{(\mathsf{h})}_{Y}(t_l,\eta_k)|\geq \nu/2$ always happen, and $O^{(\nu/2)}_Y(t_l,\eta_k)\neq -\infty$. We need to evaluate $O^{(\nu/2)}_Y(t_l,\eta_k)$. Since $|V^{(\mathsf{h})}_{\mathbf f}(t_l,\eta_k)|\geq \nu/2+\zeta_n$ and $E_f$ is of order $\varepsilon$, by the same argument as that for \eqref{Proof stability recon SST Of}, when $n$ is sufficiently large, 
\begin{align*}
|O^{(\nu/2)}_Y(t_l,\eta_k)-\phi'(t_l)|\leq &\,\frac{\varepsilon  \Xi_A\big(\frac{J^{(0)}_0}{2\pi}+3\Xi J^{(0)}_1+2\pi\Xi^2 J^{(0)}_2\big)+(1+\Xi)\zeta_n}{\nu/2} \\
&+\frac{\Xi_A(\varepsilon+\Xi+\|\mathsf{h}'\|_\infty+\|\mathsf{h}''\|_\infty) \beta}{\sqrt{n}\nu/2}
\end{align*}
Then, by the same argument as that for \eqref{proof sum g_alpha when input is f result}, we obtain
\begin{align*}
\left|\frac{\Xi}{d}\sum_{\eta_j\in R_l} g_\alpha\big(\eta_j- O^{(\nu/2)}_Y(t_l,\eta_k)\big)-1\right|\leq 2\texttt{erfc}(C_\alpha)+\frac{2\Xi C_\alpha^3}{ \sqrt{\pi}\varepsilon^{1/3} d}
\end{align*}
and hence the desired bound.
The third subcase is in general uncertain. Therefore, we only consider a trivial bound that
\begin{align*}
\left|\frac{\Xi}{d}\sum_{\eta_j\in R_l} g_\alpha\big(\eta_j- O^{(\nu/2)}_Y(t_l,\eta_k)\big)\right|\leq \frac{\Xi}{d}\sum_{\eta_j\in R_l}\frac{1}{\sqrt{\pi\alpha}}=\frac{2C_\alpha}{\pi}\,.
\end{align*}

\end{proof}

With the above lemmas, we are ready to prove the robustness theorem of SST reconstruction formula.

\begin{proof}[Proof of Theorem \ref{section:theorem:stability}]
Rewrite the targeting reconstruction formula at time $t_l$ as
\begin{align*}
&\widetilde{f}^{\mathbb{C}}(t_l)
:=\frac{1}{\mathsf{h}(0)}\frac{\Xi}{d}\sum_{\eta_j\in R_l} S^{(\mathbf{h})}_Y(t_l,\eta_j)\\
=  &\,\frac{1}{\mathsf{h}(0)}\frac{\Xi}{d}\sum_{\eta_j\in R_l}  \frac{\Xi}{d}\sum_{k=1}^{d}
V^{(\mathbf{h})}_Y(t_l,\eta_k)
\, g_\alpha\big(\eta_j- O^{(\nu/2)}_Y(t_l,\eta_k)\big)\\
=&\, \frac{1}{\mathsf{h}(0)}\frac{\Xi}{d}\sum_{k=1}^{d}
V^{(\mathbf{h})}_Y(t_l,\eta_k) I\Big(|V^{(\mathsf{h})}_{\mathbf f}(t_l,\eta_k)|>\frac{\nu}{2}-\zeta_n\Big)
\left[\frac{\Xi}{d}\sum_{\eta_j\in R_l}  g_\alpha\big(\eta_j- O^{(\nu/2)}_Y(t_l,\eta_k)\big)\right]\,,
\end{align*}
where $I(|V^{(\mathsf{h})}_{\mathbf f}(t_l,\eta_k)|> \frac{\nu}{2}-\zeta_n)$ in the last equality comes from Lemma \ref{Lemma theorem:stability g_alpha integration}.
By the control of $\frac{\Xi}{d}\sum_{\eta_j\in B_l}  g_\alpha\big(\eta_j- O^{(\nu/2)}_Y(t_l,\eta_k)\big)$ in  Lemma \ref{Lemma theorem:stability g_alpha integration} and the lower bound of $|V^{(\mathsf{h})}_{\mathbf f}(t_l,\eta_k)|$ in Lemma \ref{corro: SST on deterministic function AHM}, we obtain
\begin{align}
\widetilde{f}^{\mathbb{C}}(t_l)
 =&\, \frac{1}{\mathsf{h}(0)}\frac{\Xi}{d}\sum_{k\in B_l}V^{(\mathbf{h})}_Y(t_l,\eta_k)(1+e_{g,l,k})
\nonumber\\
&+\frac{1}{\mathsf{h}(0)}\frac{\Xi}{d}\sum_{k\in Q_l\backslash B_l}V^{(\mathbf{h})}_Y(t_l,\eta_k)(1+e_{g,l,k})\nonumber\\
&+\frac{1}{\mathsf{h}(0)}\frac{\Xi}{d}\sum_{k\notin Q_l} V^{(\mathbf{h})}_Y(t_l,\eta_k)I\Big(|V^{(\mathsf{h})}_{\mathbf f}(t_l,\eta_k)|> \frac{\nu}{2}-\zeta_n\Big) e'_{g,l,k}\,,\label{Proof recon stability expansion rewriten 1} 
\end{align}
where $e'_{g,l,k}$ satisfies
\begin{align}\label{sst recon proof definition and bound of e'glk}
|e'_{g,l,k}|\leq 
\left\{
\begin{array}{ll}
e_g& \mbox{ when }|V^{(\mathsf{h})}_{\mathbf f}(t_l,\eta_k)|\geq \frac{\nu}{2}+\zeta_n\\
 \frac{2C_\alpha}{\pi} & \mbox{ when }\frac{\nu}{2}-\zeta_n\leq |V^{(\mathsf{h})}_{\mathbf f}(t_l,\eta_k)|<\frac{\nu}{2}+\zeta_n
\end{array}
\right.
\end{align}
To obtain the desired claim, we control the right hand side of \eqref{Proof recon stability expansion rewriten 1} term by term. The first term to control is $\frac{1}{\mathsf{h}(0)}\frac{\Xi}{d}\sum_{k\in B_l}V^{(\mathbf{h})}_Y(t_l,\eta_k)$.
By the linearity of STFT, consider the bound
\begin{align}
\left|\frac{\Xi}{d}\sum_{k\in B_l}\left(V^{(\mathbf{h})}_Y(t_l,\eta_k) -V^{(\mathbf{h})}_{\mathbf f}(t_l,\eta_k)\right)\right|
\leq \frac{\Xi}{d}\sum_{k\in B_l}\left|V^{(\mathbf{h})}_{\epsilon}(t_l,\eta_k)\right|\leq 2\Delta \zeta_n\,,\nonumber
\end{align}
where we use the fact that $\max_{l,k}|V^{(\mathbf{h})}_\epsilon(t_l,\eta_k)|\leq \zeta_n$ \eqref{proof sum galpha Vepsilon V'epsilon uniform bound}. Therefore, we have when $n$ is sufficiently large,
\begin{align}
&\frac{1}{\mathsf{h}(0)}\frac{\Xi}{d}\sum_{k\in B_l}V^{(\mathbf{h})}_Y(t_l,\eta_k)=\frac{1}{\mathsf{h}(0)}\frac{\Xi}{d}\sum_{k\in B_l}V^{(\mathbf{h})}_{\mathbf f}(t_l,\eta_k) + e_l\label{proof SST recon noisy setup bound0}\,,
\end{align}
where $\max_l|e_l|\leq 2\Delta \zeta_n$. With \eqref{proof SST recon noisy setup bound0}, we obtain 
\begin{align*}
\left|\frac{1}{\mathsf{h}(0)}\frac{\Xi}{d}\sum_{k\in B_l}V^{(\mathbf{h})}_Y(t_l,\eta_k)e_{g,l,k}\right|\leq (\Xi_A+E_r+2\Delta \zeta_n)e_g\,,
\end{align*}
where we use the STFT-based reconstruction formula analysis in \eqref{definition Er error for STFT reconstruction}. Therefore, the first term of \eqref{Proof recon stability expansion rewriten 1} is controlled by
\begin{align}
\left|\frac{1}{\mathsf{h}(0)}\frac{\Xi}{d}\sum_{k\in B_l}V^{(\mathbf{h})}_Y(t_l,\eta_k)(1+e_{g,l,k})-A(t_l)e^{\ii 2\pi \phi(t_l)}\right|\leq (E_r+2\Delta \zeta_n)(1+e_g)+\Xi_A e_g\,.\label{proof SST recon noisy setup bound1}
\end{align}

Since $e'_{g,l,k}$ and $V^{(\mathbf{h})}_Y(t_l,\eta_k)$ are small in the second and third terms of \eqref{Proof recon stability expansion rewriten 1}, we control them by trivial bounds. To this end, we prepare some quantities. We also need to control the size of $Q_l\backslash B_l$.  
By the decay property of $\widehat{\mathsf{h}}$, $Q_l\backslash B_l\subset [\Delta,\,C_{\mathsf{h}}\Delta]$, where $C_{\mathsf{h}}>1$ depends on the property of $\mathsf{h}$.  
Therefore, with \eqref{sst recon proof definition and bound of e'glk}, we obtain the control for the second term:
\begin{align}
&\left|\frac{\Xi}{d}\sum_{k\in Q_l\backslash B_l}V^{(\mathbf{h})}_Y(t_l,\eta_k)(1+e_{g,l,k})\right|\label{proof SST recon noisy setup bound2}\\
\leq \,&\frac{\Xi(1+e_{g})}{d}\sum_{k\in Q_l\backslash B_l} (|V^{(\mathbf{h})}_{\mathbf{f}}(t_l,\eta_k)|+\zeta_n)\nonumber\\
\leq \,&(1+e_{g})(C_{\mathsf{h}}-1)\Delta(\nu_0+\zeta_n)\,.\nonumber
\end{align}

For the third term, we also need to control the size of $ |V^{(\mathsf{h})}_{\mathbf f}(t_l,\eta_k)|\geq \nu/2+\zeta_n$ when $k\notin Q_l$. Again, by the decay property of $\mathsf{h}$,  and the fact that $Q^c_l\subset B_l^c$, the size of the set $\{k\notin Q_l|\, |V^{(\mathsf{h})}_{\mathbf f}(t_l,\eta_k)|\geq \nu/2+\zeta_n\}$ is bounded by $C_{\mathsf{h}}'\Delta$ for some constant $C_{\mathsf{h}}'>C_{\mathsf{h}}$.  
Similarly, 
the size of the set $\{k\notin Q_l|\, \nu/2-\zeta_n<|V^{(\mathsf{h})}_{\mathbf f}(t_l,\eta_k)|\leq \nu/2+\zeta_n\}$ is bounded by $C''\zeta_n$ for some $C''>0$ due to the smoothness of $|V^{(\mathsf{h})}_{\mathbf f}(t_l,\eta)|$ as a function of $\eta$, where $C''$ depends on $\Xi_A$, $\Xi$ and $\mathsf{h}$. As a result, we have
\begin{align}
&\left|\frac{\Xi}{d}\sum_{k\notin Q_l} V^{(\mathbf{h})}_Y(t_l,\eta_k) e'_{g,l,k}\right|\label{proof SST recon noisy setup bound3}\\
\leq \,&\Bigg|\frac{\Xi}{d}\sum_{\substack{k\notin Q_l \mbox{ s.t. }\\  |V^{(\mathsf{h})}_{\mathbf f}(t_l,\eta_k)|\geq \nu/2+\zeta_n}} (|V^{(\mathsf{h})}_{\mathbf f}(t_l,\eta_k)|+\zeta_n)\Bigg| e_{g}\nonumber\\
&+\Bigg|\frac{\Xi}{d}\sum_{\substack{k\notin Q_l \mbox{ s.t. }\\  \nu/2-\zeta_n<|V^{(\mathsf{h})}_{\mathbf f}(t_l,\eta_k)|\leq \nu/2+\zeta_n}} (|V^{(\mathsf{h})}_{\mathbf f}(t_l,\eta_k)|+\zeta_n)\Bigg|\frac{2C_\alpha}{\pi}\nonumber\\
\leq \,&  \nu_0\big[C_{\mathsf{h}}'\Delta e_g+ C''C_\alpha  \zeta_n\big]\,.\nonumber
\end{align}
Putting \eqref{proof SST recon noisy setup bound1}, \eqref{proof SST recon noisy setup bound2} and \eqref{proof SST recon noisy setup bound3} together, we conclude that
\begin{align}
e_{r,l}:=\widetilde{f}^{\mathbb{C}}(t_l)-A(t_l)e^{\ii2\pi\phi(t_l)}\,,
\end{align}
satisfies
\begin{align*}
e_r:=\,&\max_l|e_{r,l}|
\leq
(E_r+2\Delta \zeta_n)(1+e_g)+\Xi_A e_g
\nonumber\\
&+(1+e_{g})(C_{\mathsf{h}}-1)\Delta(\nu_0+\zeta_n)
+\nu_0\big[C_{\mathsf{h}}'\Delta e_g+ C''C_\alpha  \zeta_n\big]\nonumber\\
\leq\,& 2E_r+2C_{\mathsf{h}}\Delta\nu_0+\Xi_Ae_g+(4+2C_{\mathsf{h}}+C''C_\alpha\nu_0)\zeta_n\,,
\end{align*}
where the last bound is an immediate simplification by taking $e_g<1$ into account.
We further simplify this complicated bound to gain some insights.
Since $m\asymp n^{1/2}$, $d\asymp n^{1/3-\gamma}$ and $\zeta_n\asymp n^{-\gamma'}$, 
when $n$ is sufficiently large and conditional on the event $\Omega$, we simplify $e_r$ by keeping how noise impact the final result by keeping terms involving $\zeta_n$. Simplify each error terms with the following trivial bounds:
\[
e_g\leq 3\texttt{erfc}(C_\alpha)+\frac{162\Xi^2C_\alpha^2}{\nu_0 \Delta_r^2} \zeta_n\,,
\]
where we use the trivial bound that $E_g\leq 3\texttt{erfc}(C_\alpha)\leq 1.5$ and 
\[
E_r\leq  2\Xi_A\frac{\varepsilon J(\Xi+1) \Delta + \beta^{-1/2} \delta_2}{\mathsf{h}(0)}\,,
\]
where we use the trivial bound $E_f\leq \varepsilon J\Xi_A(\Xi+1)$ with $J:=2\pi\max_{k=1,2,3}J^{(0)}_k$.
We thus obtain the desired simplified bound
\begin{align}\label{proof SST recon noisy setup final bound er definition}
e_r\leq\,& C_1\Xi_A+C_2\zeta_n\,,
\end{align}
where 
\[
C_1=3\texttt{erfc}(C_\alpha)+C_{\mathsf{h}}\Delta\sqrt{\beta}\delta_1+\frac{2(\varepsilon J(\Xi+1) \Delta + \beta^{-1/2} \delta_2)}{\mathsf{h}(0)}
\]
and
\[
C_2:=4+2C_{\mathsf{h}}+C''C_\alpha\Xi_a\sqrt{\beta}\delta_1+\frac{162\Xi_A\Xi^2C_\alpha^2}{\Xi_a\sqrt{\beta}\delta_1 \Delta_r^2} \,.
\]
\end{proof}

\subsection{Proof of Theorem \ref{Bootstrapping main theorem}}\label{subsection proof of bootstrapping for sst}

\begin{proof}
Recall the existence of a Gaussian random process $\hat{\epsilon}_i$ that shares the same covariance structure of $\epsilon_i$ shown in Theorem \ref{Theorem: STFT distribution like Gaussian}.
Under the given assumptions, we can well approximate $\hat{\epsilon}_i$, where $i=1,\ldots,n$, by a Gaussian tvAR process of order $b\in \mathbb{N}$ \cite[Theorem 2.11]{ding2023autoregressive}, denoted as ${\epsilon}^{(*)}_i$, where $i=1,\ldots,n$, so that ${\epsilon}^{(*)}_i=\hat{\epsilon}_i$ in law for $i=1,\ldots,b$.
Since we can find a sufficiently rich probability space to host $\epsilon^*_i$ and $\hat{\epsilon}_i$, $\hat{\epsilon}_i-\epsilon^*_i$ is Gaussian with mean $0$, and the error is controlled by
\begin{equation}
{\epsilon}^{(*)}_i-\hat{\epsilon}_i=O_{\ell^2}(\log(b)^\tau b^{-(\tau-2)} + b^{2.5}n^{-1})
\end{equation}
for $1\leq i\leq n$. Since $b$ is chosen to fulfill $\frac{b}{\log(b)}\asymp n^{1/(\tau+1)}$ to balance errors between the truncation and smooth approximation, the error becomes $O_{\ell^2}(n^{-\gamma})$, where $\gamma=\frac{\tau-2}{\tau+1}\in(0, 1]$; that is, we have
\begin{equation}\label{equation tvAR var bound control}
\max_{i=1,\ldots,n}\|{\epsilon}^{(*)}_i-\hat{\epsilon}_i\|_2\leq Cn^{-\gamma}
\end{equation}
for some constant $C>0$. Hence, if we denote the complex random vector associated with the discretized STFT of ${\epsilon}^{(*)}-\hat{\epsilon}$ by
\begin{align*}
E_1:=[&\,\Re V^{(\mathbf{h})}_{{\epsilon}^{(*)}-\hat{\epsilon}}(1,\eta_1),\, \Re V^{(\mathbf{h}')}_{{\epsilon}^{(*)}-\hat{\epsilon}}(1,\eta_1),\ldots,\,\Re V^{(\mathbf{h})}_{{\epsilon}^{(*)}-\hat{\epsilon}}(n,\eta_1),\,\Re V^{(\mathbf{h}')}_{{\epsilon}^{(*)}-\hat{\epsilon}}(n,\eta_1),\\
&\,\Im V^{(\mathbf{h})}_{{\epsilon}^{(*)}-\hat{\epsilon}}(1,\eta_1),\, \Im V^{(\mathbf{h}')}_{{\epsilon}^{(*)}-\hat{\epsilon}}(1,\eta_1),\ldots,\,\Im V^{(\mathbf{h})}_{{\epsilon}^{(*)}-\hat{\epsilon}}(n,\eta_1),\,\Im V^{(\mathbf{h}')}_{{\epsilon}^{(*)}-\hat{\epsilon}}(n,\eta_1),\\
&\,\Re V^{(\mathbf{h})}_{{\epsilon}^{(*)}-\hat{\epsilon}}(1,\eta_2),\,\Re V^{(\mathbf{h}')}_{{\epsilon}^{(*)}-\hat{\epsilon}}(1,\eta_2),\ldots,\,\Re V^{(\mathbf{h})}_{{\epsilon}^{(*)}-\hat{\epsilon}}(n,\eta_2),\,\Re V^{(\mathbf{h}')}_{{\epsilon}^{(*)}-\hat{\epsilon}}(n,\eta_2),\\
&\,\ldots \\
&\,\Im V^{(\mathbf{h})}_{{\epsilon}^{(*)}-\hat{\epsilon}}(1,\eta_d),\,\Im V^{(\mathbf{h}')}_{{\epsilon}^{(*)}-\hat{\epsilon}}(1,\eta_d),\ldots,\,\Im V^{(\mathbf{h})}_{{\epsilon}^{(*)}-\hat{\epsilon}}(n,\eta_d),\,\Im V^{(\mathbf{h}')}_{{\epsilon}^{(*)}-\hat{\epsilon}}(n,\eta_d)]^\top\in \mathbb{R}^{4nd}\,,
\end{align*}
where $d\geq 1$ is the number of frequencies we have interest, we know $E_1$ is a Gaussian vector. 
Clearly, $\mathbb{E}E_1=0$.
By \eqref{equation tvAR var bound control}, we claim that the entrywise standard deviation is controlled by
\[
\beta_i:=(\mathbb{E}[|E_1(i)|^2])^{1/2} \leq C'n^{-\gamma}\,,
\]
where $C'>0$ is a constant, for any $i=1,\ldots, 4nd$. Indeed, by taking $i=1$ as an example, we have
\begin{align*}
&\|E_1(1)\|_2\leq \frac{1}{n^{1/4}}\sum_{j=1-\lceil\beta \sqrt{n}\rceil}^{1+\lceil\beta \sqrt{n}\rceil}\left\|({\epsilon}^{(*)}_j-\hat{\epsilon}_j)\mathbf{h}(j-1)\cos(2\pi\eta_1 (j-1))\right\|_2\\
\leq &\, \frac{1}{n^{1/4}}\sum_{j=1-\lceil\beta \sqrt{n}\rceil}^{1+\lceil\beta \sqrt{n}\rceil}\left\|{\epsilon}^{(*)}_j-\hat{\epsilon}_j\right\|_2\frac{1}{n^{1/4}}\mathsf{h}\left(-\beta+\frac{j-1}{\sqrt{n}}\right)\\
\leq &\,\max_{i=1,\ldots,n}\left\|{\epsilon}^{(*)}_j-\hat{\epsilon}_j\right\|_2\sum_{j=1-\lceil\beta \sqrt{n}\rceil}^{1+\lceil\beta \sqrt{n}\rceil}\frac{1}{\sqrt{n}}\mathsf{h}\left(-\beta+\frac{j-1}{\sqrt{n}}\right)\\
\leq &\, C'n^{-\gamma}\,,
\end{align*}
where $C'\leq 2\hat{h}(0)C$ by a Riemann sum approximation. Other entries are controlled by the same way, while noting that the Riemann sum approximation of $\widehat{\mathsf{h}'}(0)$ decays to $0$ at the rate $n^{-1/2}$.
A uniform bound of $E_1$ using the Gaussian tail bound is
\begin{align}
&\mathbb{P}\left\{\max_{i=1,\ldots,4nd}|E_1(i)|>t\right\}\leq \sum_{i=1}^{4nd}\mathbb{P}\left\{|E_1(i)|>t\right\}
\leq\sum_{i=1}^{4nd} e^{-\frac{t^2}{2\sigma_{i}^2}}\leq 4nde^{-\frac{t^2 n^{2\gamma}}{2{C'}^2}}\,. \nonumber
\end{align}
Thus, for any large constant $D>1$, we have 
\[
\mathbb{P}\left\{\max_{i=1,\ldots,4nd}|E_1(i)|>Dn^{-\gamma}\sqrt{\log(n)}\right\}\leq n^{-2} 
\]
when $n$ is sufficiently large. Denote $\tilde{\mathsf V}_i$ and $\hat{\mathsf V}_i$ to be the discretized STFT coefficients at time $t$ associated with $\{\epsilon^*_i\}$ and $\{\hat{\epsilon}_i\}$ respectively; that is,
\begin{align*}
\tilde{\mathsf V}_i:=[V^{(\mathbf{h})}_{{\epsilon}^{(*)}}(t_i,\eta_1),\ldots,V^{(\mathbf{h})}_{{\epsilon}^{(*)}}(t_i,\eta_d),V^{(\mathbf{h}')}_{{\epsilon}^{(*)}}(t_i,\eta_1),\ldots,V^{(\mathbf{h}')}_{{\epsilon}^{(*)}}(t_i,\eta_d)]^\top\in \mathbb{C}^{2d},
\\
\hat{\mathsf V}_i:=[V^{(\mathbf{h})}_{\hat{\epsilon}}(t_i,\eta_1),\ldots,V^{(\mathbf{h})}_{\hat{\epsilon}}(t_i,\eta_d),V^{(\mathbf{h}')}_{\hat{\epsilon}}(t_i,\eta_1),\ldots,V^{(\mathbf{h}')}_{\hat{\epsilon}}(t_i,\eta_d)]^\top\in \mathbb{C}^{2d},
\end{align*}
which are $2d$-dim complex Gaussian random vectors. The above calculation suggests that with probability higher than $1-O(n^{-2})$, we have
\[
\max_{i=1,\ldots,n}|\hat{\mathsf V}_i-\tilde{\mathsf V}_i|\leq C\sqrt{d}n^{-\gamma}\sqrt{\log(n)}
\]
for some constant $C>0$.

Next, by Theorem \ref{Theorem: STFT distribution like Gaussian}, if we denote ${\mathsf V}_i$ to be the discretized STFT coefficients  associated with $\{\epsilon_i\}$; that is,
\begin{align*}
{\mathsf V}_i:=[V^{(\mathbf{h})}_{\epsilon}(t_i,\eta_1),\ldots,V^{(\mathbf{h})}_{\epsilon}(t_i,\eta_d),V^{(\mathbf{h}')}_{\epsilon}(t_i,\eta_1),\ldots,V^{(\mathbf{h}')}_{\epsilon}(t_i,\eta_d)]^\top\in \mathbb{C}^{2d},
\end{align*} 
then, since $m=\lceil \beta \sqrt{n}\rceil$, we have
\begin{align*}
\mathbb{E}\left(\max_l |\mathsf V_t-\hat{\mathsf V}_t|\right)
\leq {C'}\left(\frac{d^{\frac{3}{4}-\frac{1}{2s}-\frac{1}{p}}}{n^{\frac{1}{4}-\frac{1}{2p}-\frac{1}{2s}+\frac{1}{ps}}}\right)^{\frac{1}{1-\frac{1}{s}-\frac{1}{p}}} \log(n)^2\,,\nonumber
\end{align*}
and hence by Chebychev's inequality, we obtain
\[
\max_l |\mathsf V_t-\hat{\mathsf V}_t|
=o_p\left(\left(\frac{d^{\frac{3}{4}-\frac{1}{2s}-\frac{1}{p}}}{n^{\frac{1}{4}-\frac{1}{2p}-\frac{1}{2s}+\frac{1}{ps}}}\right)^{\frac{1}{1-\frac{1}{s}-\frac{1}{p}}} \log(n)^3\right)\,.
\]
By the above calculation, we have the control between the STFT of the original time series and the bootstrapped Gaussian time series by tvAR by
\begin{align}
\max_{t=1,\ldots,n}|\mathsf V_t-\tilde{\mathsf V}_t|\nonumber
\leq &\,\max_{t=1,\ldots,n}|\mathsf V_t-\hat{\mathsf V}_t|+\max_{t=1,\ldots,n}|\hat{\mathsf V}_t-\tilde{\mathsf V}_t|\\
=&\,o_p\left(\sqrt{d}n^{-\gamma}\sqrt{\log(n)}+\left(\frac{d^{\frac{3}{4}-\frac{1}{2s}-\frac{1}{p}}}{n^{\frac{1}{4}-\frac{1}{2p}-\frac{1}{2s}+\frac{1}{ps}}}\right)^{\frac{1}{1-\frac{1}{s}-\frac{1}{p}}} \log(n)^3\right)\label{proof BS V difference}
\end{align}
since in general $\sqrt{d}n^{-\gamma}\sqrt{\log(n)}$ and $\left(\frac{d^{\frac{3}{4}-\frac{1}{2s}-\frac{1}{p}}}{n^{\frac{1}{4}-\frac{1}{2p}-\frac{1}{2s}+\frac{1}{ps}}}\right)^{\frac{1}{1-\frac{1}{s}-\frac{1}{p}}} \log(n)^3$ are not comparable.

Recall that the SST of $\{\epsilon_i\}$ is defined as
\[
S^{(\mathbf{h})}_{\epsilon}(t_i,\xi_l):=\frac{\Xi}{d}\sum_{k=1}^{d}
V^{(\mathsf{h})}_{\epsilon}(t_i,\eta_k)
\, g_\alpha\big(\xi_l- O^{(\nu)}_{\epsilon}(t_i,\eta_k)\big)\in \mathbb{C}
\]
for $i=1,\ldots,n$ and $\xi_l>0$, and $S^{(\mathbf{h})}_{{\epsilon}^{(*)}}(t_i,\xi_l)$ is defined similarly from $\{\epsilon^*_i\}$.
Denote $\mathsf V_{\epsilon,i}:=[V^{(\mathsf{h})}_{\epsilon}(t_i,\eta_1),\ldots,V^{(\mathsf{h})}_{\epsilon}(t_i,\eta_d)]^\top\in \mathbb{C}^{d}$ (the first $d$ entries of $\mathsf  V_{i}$), and $\mathsf V_{\epsilon^*,i}$ is defined similarly from $\{{\epsilon}^{(*)}_i\}$. Also denote
\[
\mathbf g_i:=\left[g_\alpha(\xi_l-O^{(\nu)}_{\epsilon}(t_i,\eta_1)),\ldots,g_\alpha(\xi_l-O^{(\nu)}_{\epsilon}(t_i,\eta_d))\right]^\top \in \mathbb{R}^{d}
\]
and $\tilde{\mathbf g}_j$ is the associated random vector defined with $\{{\epsilon}^{(*)}_i\}$.
We have
\begin{align*}
&\left|S^{(\mathbf{h})}_{\epsilon}(t_i,\xi_l)-S^{(\mathbf{h})}_{{\epsilon}^{(*)}}(t_i,\xi_l)\right|\\
=\,&\left|\frac{1}{d}\sum_{k=1}^{d}
\left[V^{(\mathbf{h})}_{\epsilon}(t_i,\eta_k)
\, g_\alpha\big(\xi_l- O^{(\nu)}_{\epsilon}(t_i,\eta_k)\big)-V^{(\mathbf{h})}_{{\epsilon}^{(*)}}(t_i,\eta_k)
\, g_\alpha\big(\xi_l-O^{(\nu)}_{{\epsilon}^{(*)}}(t_i,\eta_k)\big)\right]\right|\\
=\,&\frac{1}{d}
\left|(\mathsf V_{\epsilon,i}-\mathsf V_{\epsilon^*,i})\cdot {\mathbf g}_i+ \mathsf V_{\epsilon^*,i}\cdot (\mathbf g_i-\tilde{\mathbf g}_i)\right|\\
\leq\,&\frac{|\mathsf V_{Y,i}|}{d}
\left|\mathbf g_i-\tilde{\mathbf g}_i\right|+\frac{|{\mathbf g}_i|}{d}|\mathbf  V_{i}-\tilde{\mathbf V}_{i}| \,.
\end{align*}
Note that $|{\mathbf g}_i|\leq \sqrt{d/\alpha}$. Thus, combined with \eqref{proof BS V difference}, the second term on the right hand side is controlled by $o_p\left(n^{-\gamma}\sqrt{\log(n)}+\frac{d^{-\frac{1}{2}+(\frac{3}{4}-\frac{1}{2s}-\frac{1}{p})/(1-\frac{1}{s}-\frac{1}{p})}\log(n)^3}{n^{(\frac{1}{4}-\frac{1}{2p}-\frac{1}{2s}+\frac{1}{ps})/(1-\frac{1}{s}-\frac{1}{p})}} \right)$.  Since $n^{1/4}\mathsf V_{Y,i}$ is a non-degenerate $d$-dim complex Gaussian random vector with order $1$ entrywise variance, $|\mathsf V_{Y,i}|$ can be bounded by $C\sqrt{d\log(d)}n^{-1/4}$ for some constant $C>0$ with probability greater than $1-n^{-2}$. Also, we can trivially bound $| \mathbf g_i-\tilde{\mathbf g}_i|$ by $2\sqrt{d/\alpha}$. As a result, the first term on the right hand side is controlled by $o_p(dn^{-1/4}\sqrt{\log(d)})$. As a result, for $\xi_l$, 
\begin{align*}
&\max_{j=1,\ldots,n}|S^{(\mathbf{h})}_{\epsilon}(t_i,\xi_l)-S^{(\mathbf{h})}_{{\epsilon}^{(*)}}(t_i,\xi_l)|\\
=\,&o_p\left((n^{-\gamma}+dn^{-1/4})\sqrt{\log(d)}+\frac{d^{-\frac{1}{2}+(\frac{3}{4}-\frac{1}{2s}-\frac{1}{p})/(1-\frac{1}{s}-\frac{1}{p})}\log(n)^3}{n^{(\frac{1}{4}-\frac{1}{2p}-\frac{1}{2s}+\frac{1}{ps})/(1-\frac{1}{s}-\frac{1}{p})}} \right)\,.
\end{align*}
By a direct union bound, since $\frac{\frac{1}{4}-\frac{1}{2p}-\frac{1}{2s}+\frac{1}{ps}}{\frac{1}{2}(1-\frac{1}{s}-\frac{1}{p})+(\frac{3}{4}-\frac{1}{2s}-\frac{1}{p})}>1/8$ for any $p,s$, we conclude that for $\xi_1,\ldots,\xi_d$, when 
\[
a=\min\left\{\gamma,\, \frac{\frac{1}{4}-\frac{1}{2p}-\frac{1}{2s}+\frac{1}{ps}}{\frac{1}{2}(1-\frac{1}{s}-\frac{1}{p})+(\frac{3}{4}-\frac{1}{2s}-\frac{1}{p})}\right\}-\vartheta\,, 
\]
where $\vartheta$ is any small positive constant so that $a\geq 0$, we have  
\begin{align*}
&\max_{l=1,\ldots,d}\max_{i=1,\ldots,n}|S^{(\mathbf{h})}_{\epsilon}(t_i,\xi_l)-S^{(\mathbf{h})}_{{\epsilon}^{(*)}}(t_i,\xi_l)|\\
=\,&o_p\left((dn^{-\gamma}+d^2n^{-1/4})\sqrt{\log(d)}+\frac{d^{\frac{1}{2}+(\frac{3}{4}-\frac{1}{2s}-\frac{1}{p})/(1-\frac{1}{s}-\frac{1}{p})}\log(n)^3}{n^{(\frac{1}{4}-\frac{1}{2p}-\frac{1}{2s}+\frac{1}{ps})/(1-\frac{1}{s}-\frac{1}{p})}}\right)=o_p(1)\,.
\end{align*}
\end{proof}

\section{More Numerical Results}\label{Section more numerical results supp}

\subsection{Null case}\label{Section more numerical results supp: null case}

To numerically validate this result, fix $n=2048$, and construct a locally stationary random process in the following way. First, construct a tvAR process via 
\begin{align*}
\epsilon_i=
\left\{
\begin{array}{ll}
\eta_i & \mbox{ when }i=1,2\\
\sum_{k=1}^2\phi_k(i/n)\epsilon_{i-k}+\eta_i & \mbox{ when }i=3,4,\ldots,n
\end{array}
\right.
\end{align*}
where $\eta_i$ is an i.i.d. Gaussian process with standard deviation $1$, $\phi_1(i)=-0.5(0.7+0.3\cos(2\pi i/n)$ and $\phi_2(i)=0.3\sqrt{0.1+i/(4n)}$. 
Then, construct the simulated locally stationary random process via $x_i=(1+0.5\cos(2\pi i/n))\epsilon_i$.
See Figure \ref{Figure:tvARapprox} for an illustration of the estimated tvAR process and one realization of the bootstrap time series.

\begin{figure}[bht!]
\centering
\includegraphics[trim=0 0 0 0,clip,width=0.85\textwidth]{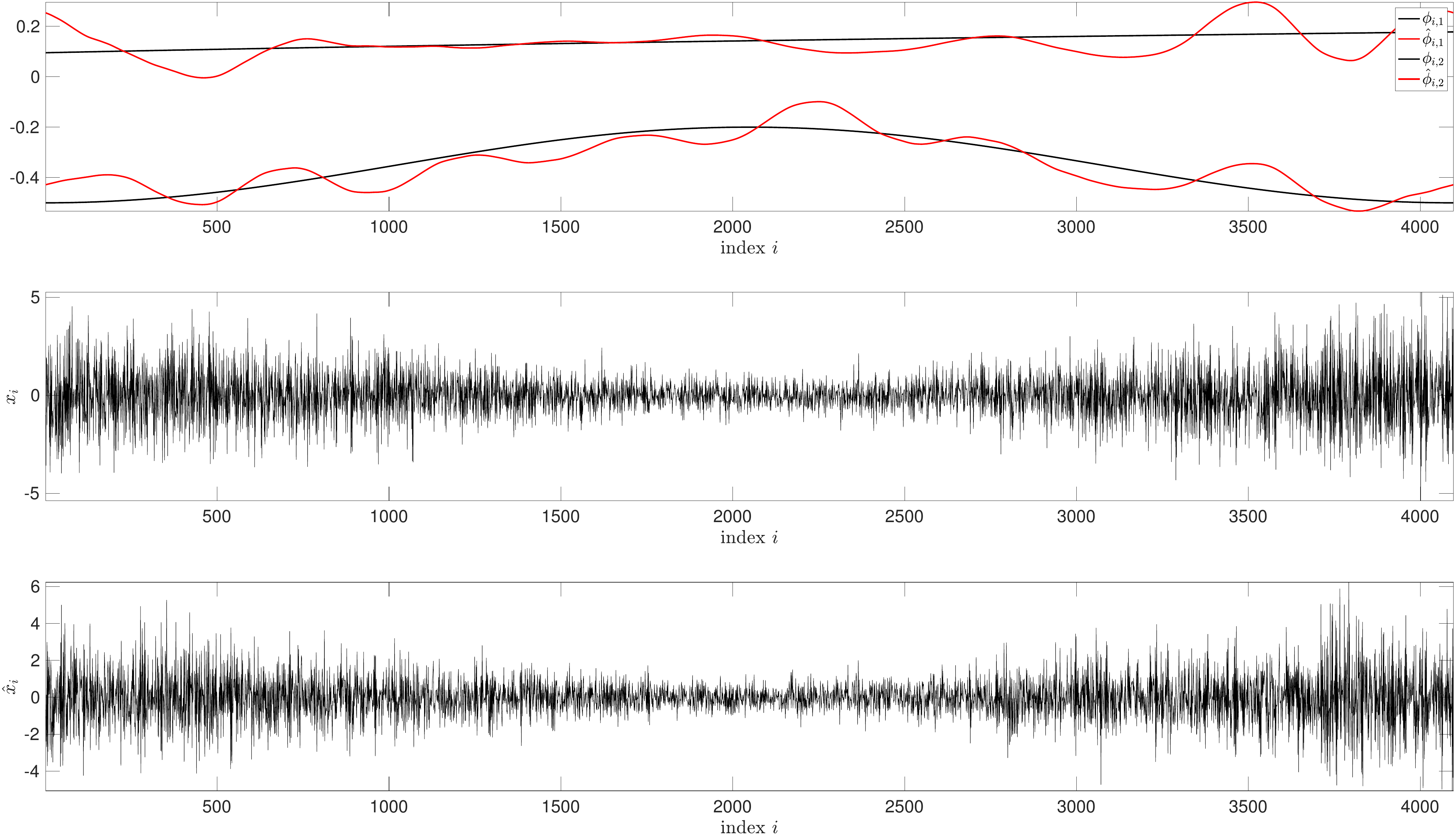}
\caption{Top: the true tvAR coefficients, and the estimated coefficients of the approximate tvAR process. Middle: one realization of $x_i$. Bottom: the bootstrapped $x_i$.\label{Figure:tvARapprox}}  
\end{figure}


\subsection{QQ plots for a comparison}

The true distributions of STFT and SST are collected by realizing $x_i$ for $M=1000$ times, and the bootstrapped distributions are generated from realizing $1000$ bootstraps of one realization of $x_i$.
Figures \ref{FigureS2:STFT BS real} and \ref{FigureS3:SST BS real} present QQ plots comparing the distributions of the STFT and SST coefficients under the true noise model and those obtained via bootstrapping, evaluated at various time-frequency pairs. The In these plots, columns correspond to time points and rows to frequency levels. We only show the real part and the imaginary part is similar. The results demonstrate that the bootstrap-based distributions closely approximate those from the true noise model, validating the effectiveness of the bootstrapping approach.

\begin{figure}[bht!]
\centering
\includegraphics[trim=0 0 0 0,clip,width=\textwidth]{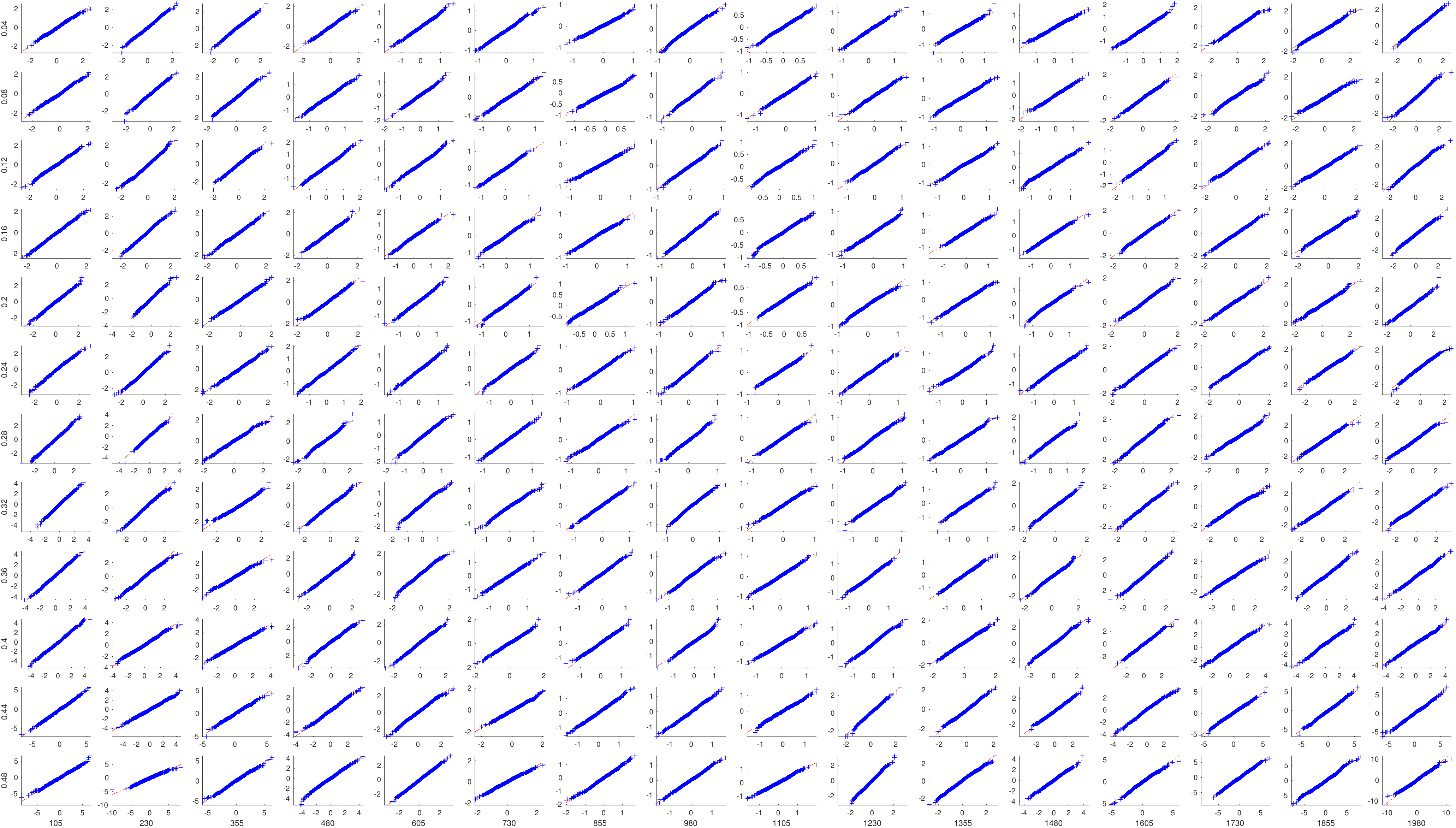}
\caption{The QQ plot of the distribution of the real part of STFT of $\epsilon_i$ (null case) and its bootstrap. \label{FigureS2:STFT BS real}}  
\end{figure}

\begin{figure}[bht!]
\centering
\includegraphics[trim=0 0 0 0,clip,width=\textwidth]{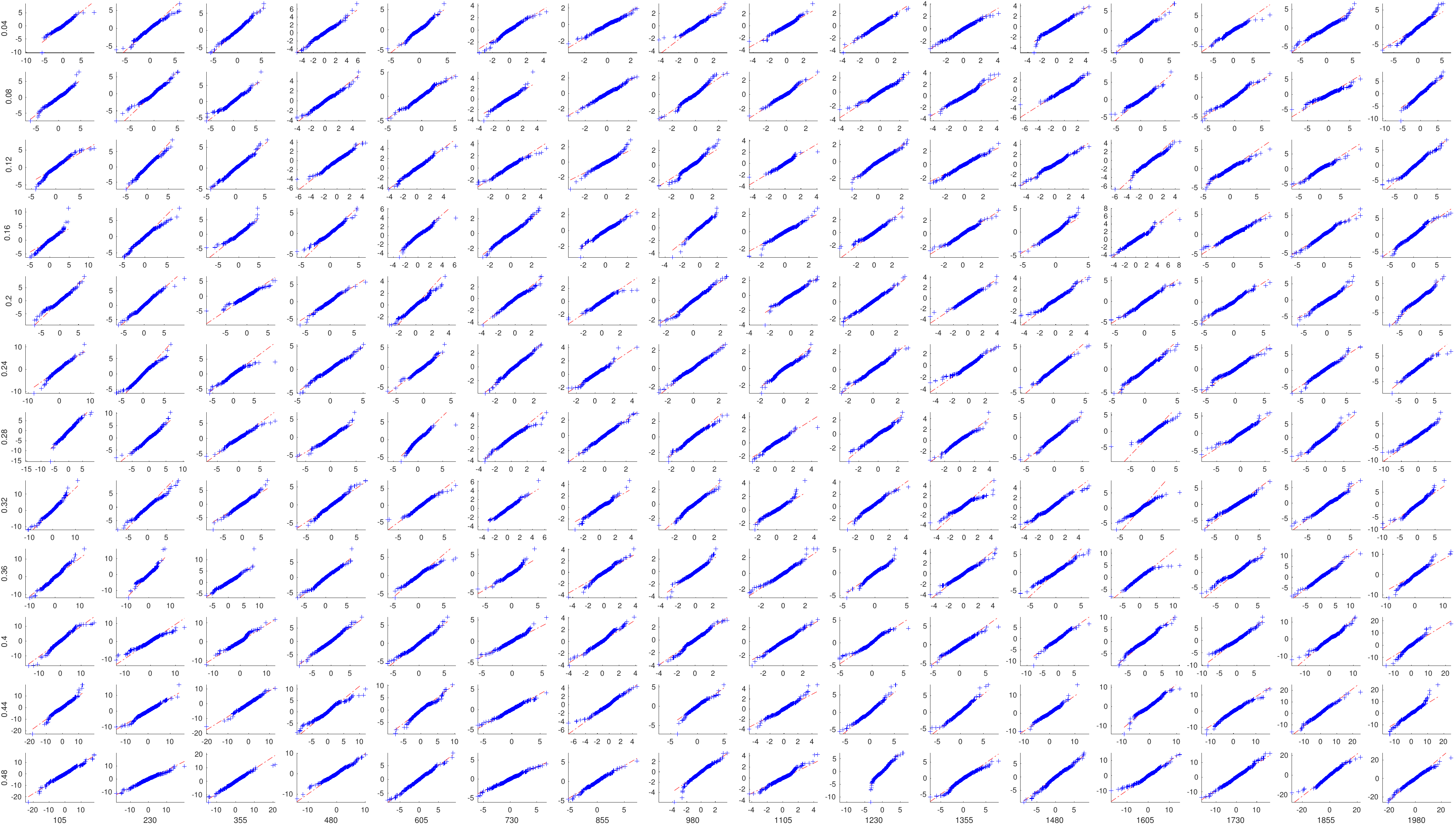}
\caption{The QQ plot of the distribution of the real part of SST of $\epsilon_i$ (null case) and its bootstrap.\label{FigureS3:SST BS real}}  
\end{figure}

Under the non-null setting, let $\tilde{f}(i/\sqrt{n})$, denote the reconstructed signal obtained via the SST-based reconstruction formula, and let the reconstructed noise be defined as $\tilde{\epsilon}_i:=X_i-\tilde{f}(i/\sqrt{n})$.
The QQ plots comparing the STFT and SST coefficients of the bootstrapped noises, denoted by $\epsilon^{(*m)}_{i}$ for $m=1,\ldots,1000$, with those of the noises generated from the underlying model are shown in Figures~\ref{FigureS4:STFT BS real nonnull} and \ref{FigureS5:SST BS real nonnull}, respectively. 
Similarly, Figures~\ref{FigureS6:STFT BS+signal real nonnull} and \ref{FigureS7:SST BS+signal real nonnull} display the QQ plots of the STFT and SST coefficients of the bootstrapped signals, defined as ${x}^{(*m)}_i:=\tilde{f}(i/\sqrt{n})+{\epsilon}^{(*m)}_{i}$, for $m=1,\ldots,1000$, together with those of the noisy signals from the underlying model. Only the real parts are shown; the imaginary parts exhibit similar patterns.
These results demonstrate that, owing to the robustness of the SST-based reconstruction algorithm and the well approximation of the noise structure, the bootstrap-based distributions closely approximate those derived from the true model. This provides strong evidence supporting the validity of combining SST-based reconstruction with the bootstrap framework for uncertainty quantification.

\begin{figure}[bht!]
\centering
\includegraphics[trim=0 0 0 0,clip,width=\textwidth]{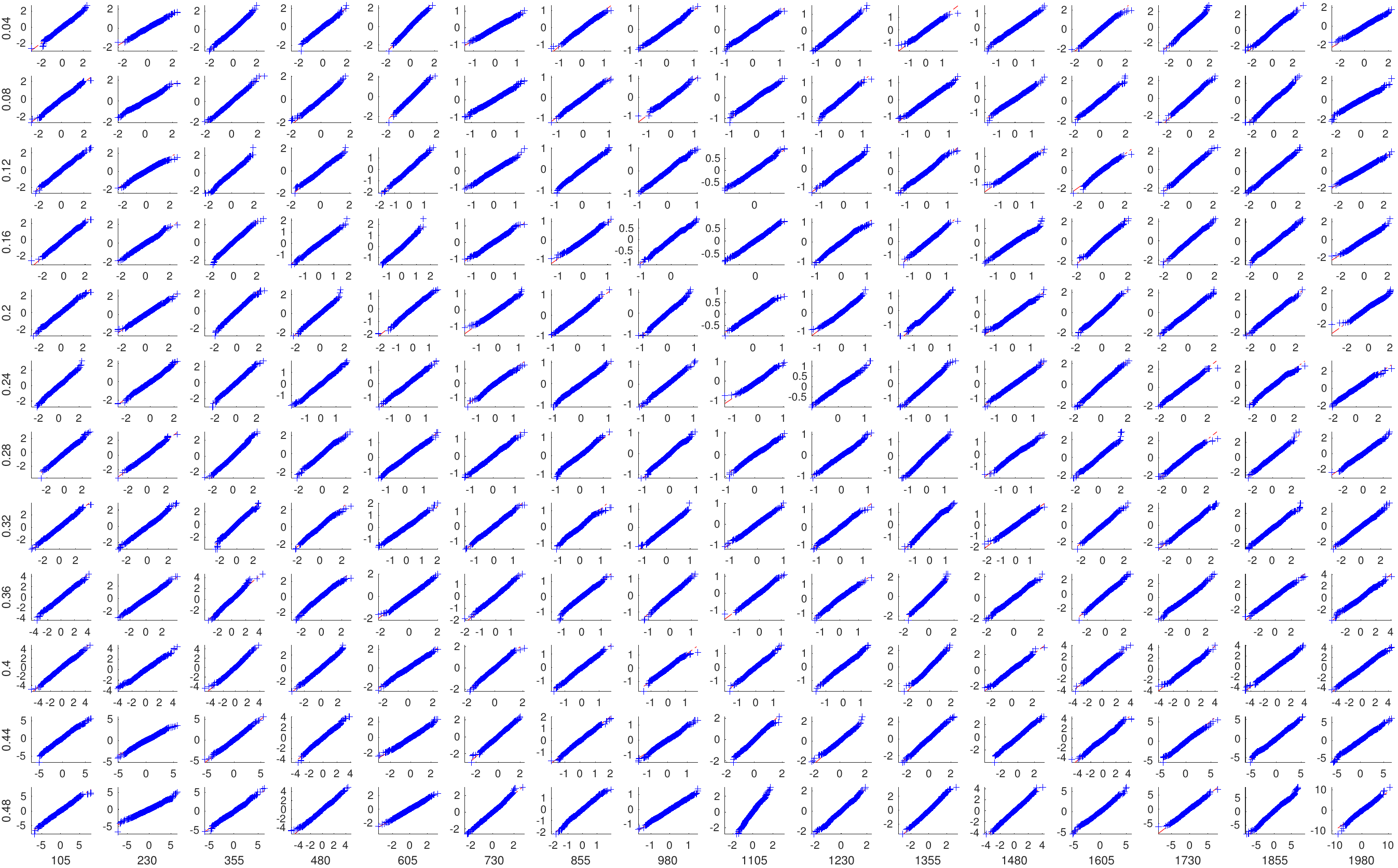}
\caption{The QQ plot of the distribution of the real part of STFT of $\epsilon_i$ under the nonnull case and its bootstrap using the reconstructed noise, $\epsilon^{(*m)}_{i}$, where $m=1,\ldots,1000$. \label{FigureS4:STFT BS real nonnull}}  
\end{figure}

\begin{figure}[bht!]
\centering
\includegraphics[trim=0 0 0 0,clip,width=\textwidth]{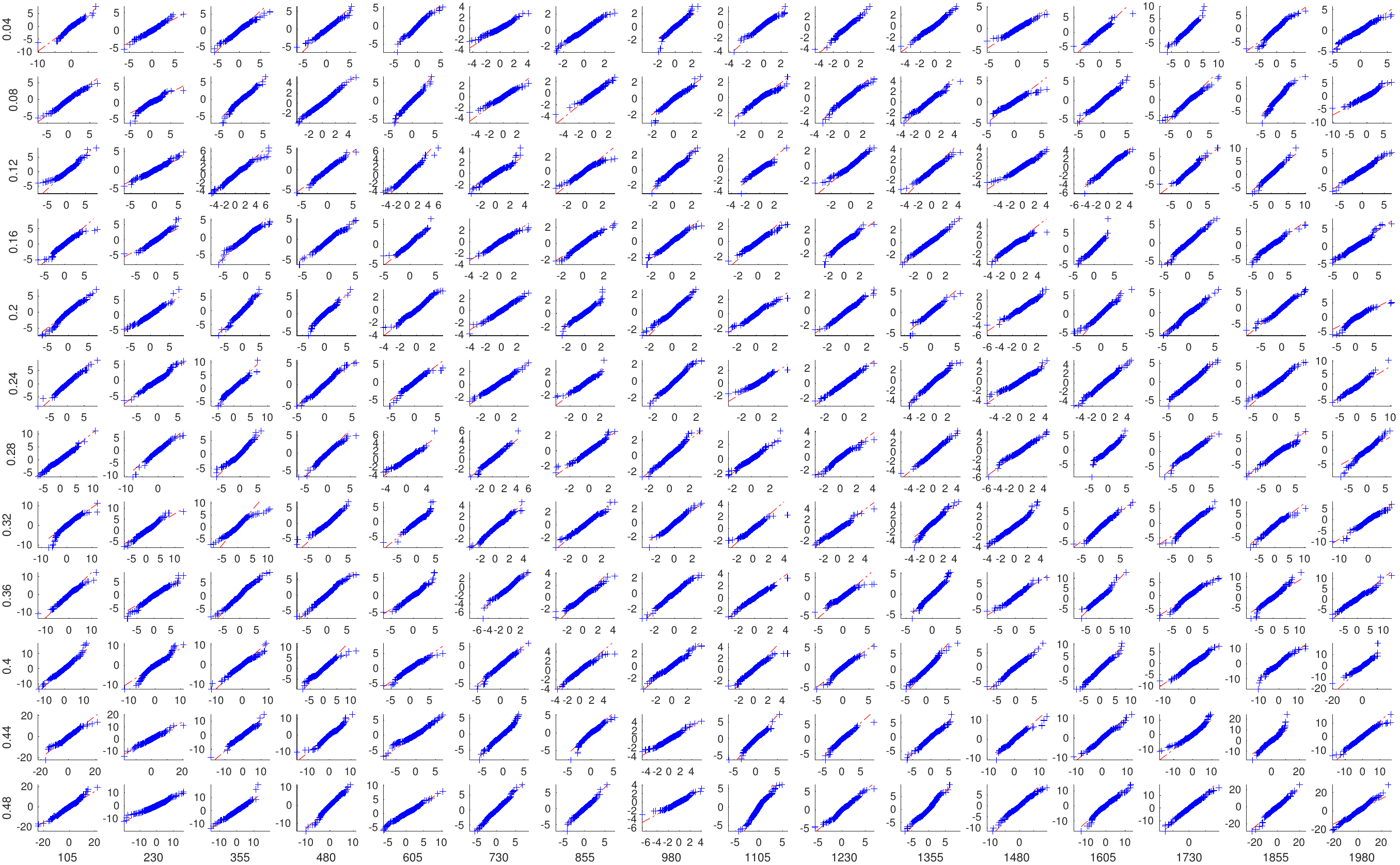}
\caption{The QQ plot of the distribution of the real part of SST of $\epsilon_i$ under the nonnull case and its bootstrap using the reconstructed noise, $\epsilon^{(*m)}_{i}$, where $m=1,\ldots,1000$.\label{FigureS5:SST BS real nonnull}}  
\end{figure}

\begin{figure}[bht!]
\centering
\includegraphics[trim=0 0 0 0,clip,width=\textwidth]{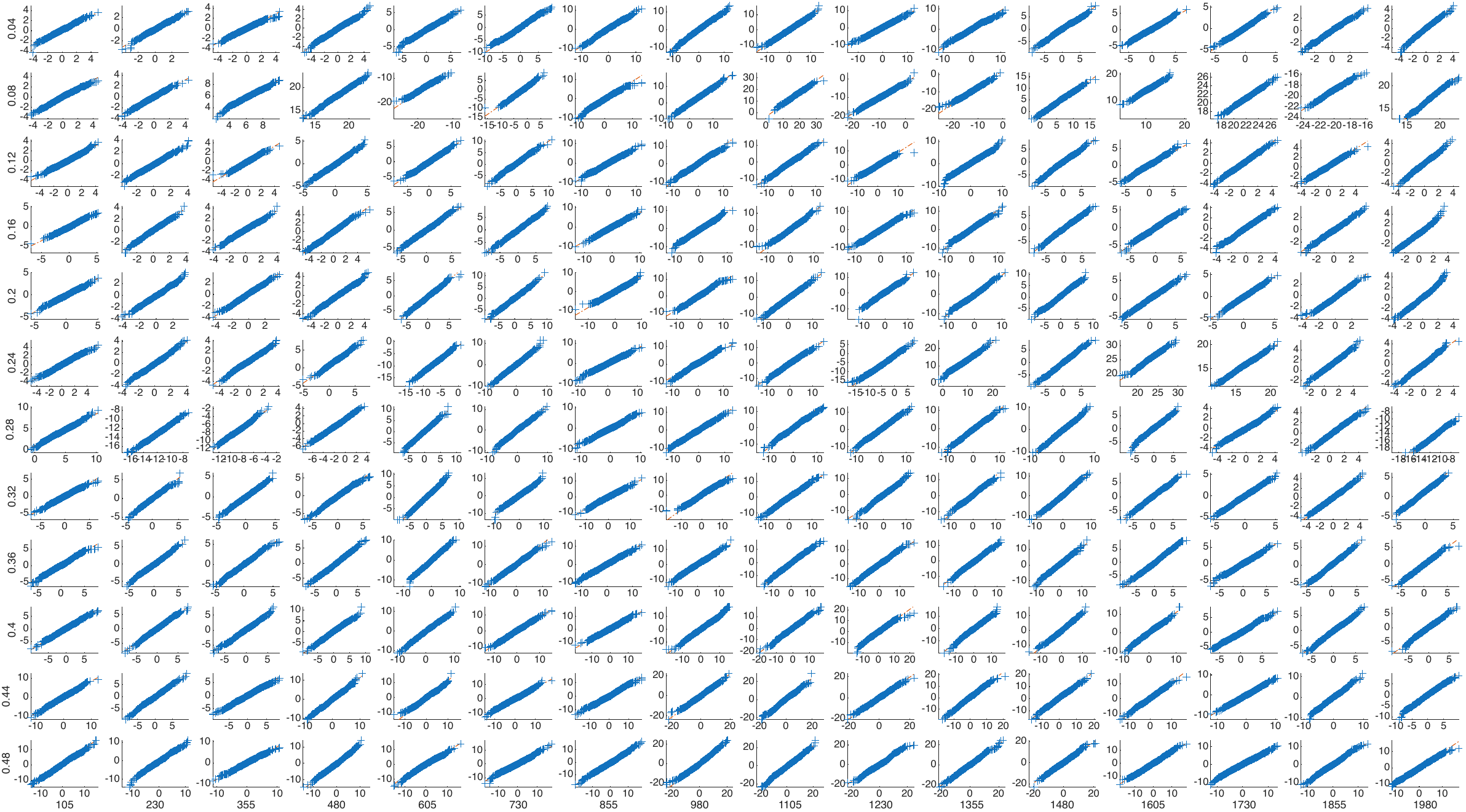}
\caption{The QQ plot of the distribution of the real part of STFT of $X_i$ under the nonnull case and its bootstrap using the reconstructed  signal and noise ${x}^{(*m)}_i$, where $m=1,\ldots,1000$. \label{FigureS6:STFT BS+signal real nonnull}}  
\end{figure}

\begin{figure}[bht!]
\centering
\includegraphics[trim=0 0 0 0,clip,width=\textwidth]{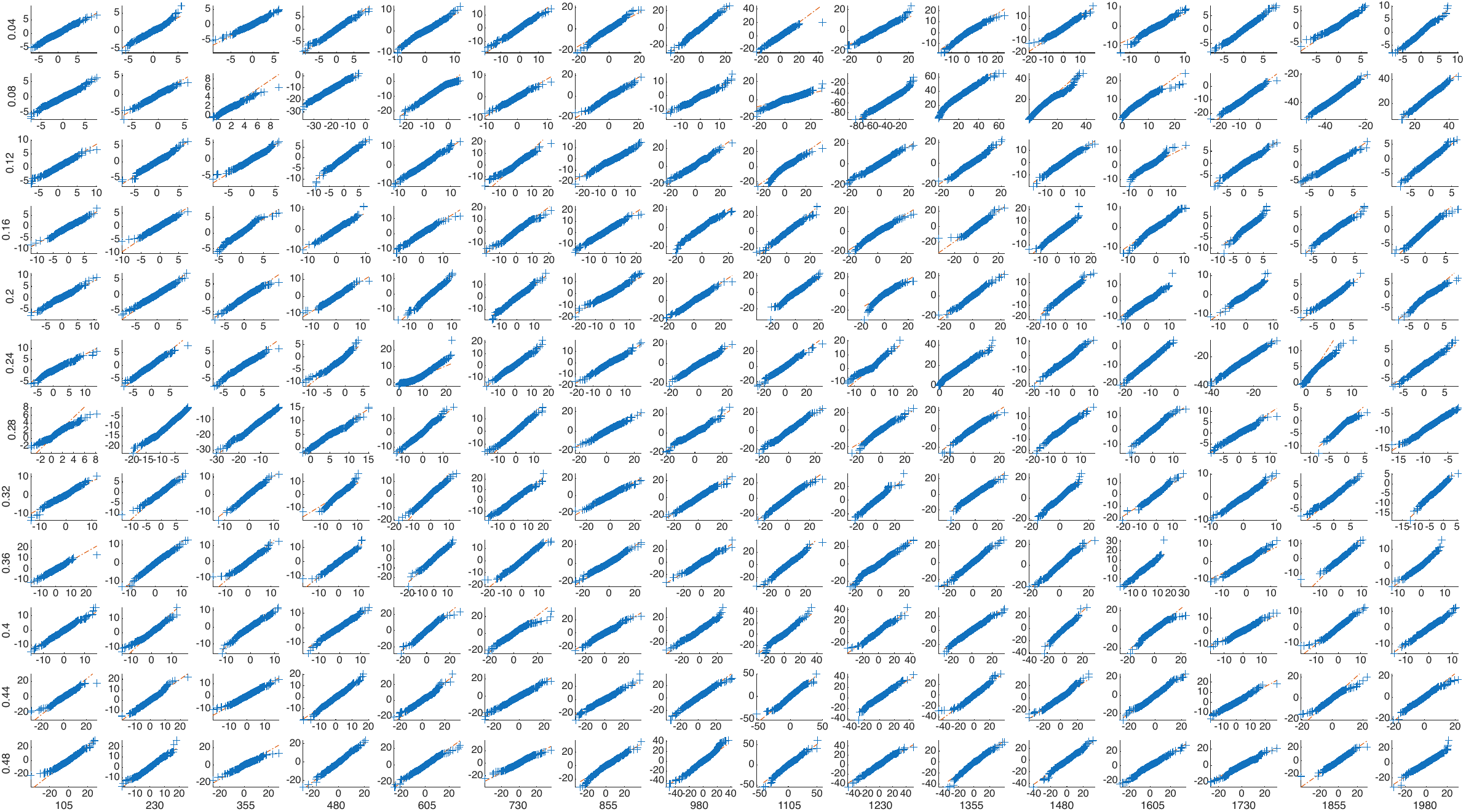}
\caption{The QQ plot of the distribution of the real part of SST of $X_i$ under the nonnull case and its bootstrap using the reconstructed signal and noise ${x}^{(*m)}_i$, where $m=1,\ldots,1000$.\label{FigureS7:SST BS+signal real nonnull}}  
\end{figure}

\end{document}